\documentclass[12pt]{article}

\def\tdotoggle{1}

\usepackage[outline]{contour}\usepackage[T1]{fontenc}

\usepackage{amsfonts,amssymb,amsmath,amsthm,mathtools} \usepackage{graphicx,microtype,latexsym,lpic,bm,xspace,booktabs,wrapfig,array,booktabs,multirow}

\usepackage[bf]{caption}

\usepackage{algorithmicx,algorithm}
\usepackage[noend]{algpseudocode}
\usepackage{thmtools,thm-restate}
\usepackage{tocloft}

\usepackage{dsfont}
\usepackage[greek,english]{babel}

\usepackage{enumitem}  \setlist[itemize]{noitemsep,label=$-$}
\setlist[enumerate]{noitemsep}

\usepackage[usenames,dvipsnames]{xcolor}
\definecolor{Gred}{RGB}{219, 50, 54}
\definecolor{Ggreen}{RGB}{60, 186, 84}
\definecolor{Gblue}{RGB}{72, 133, 237}
\definecolor{Gyellow}{RGB}{247, 178, 16}
\definecolor{ToCgreen}{RGB}{0, 128, 0}
\definecolor{myGold}{RGB}{231,141,20}
\definecolor{myBlue}{rgb}{0.19,0.41,.65}
\definecolor{myPurple}{RGB}{175,0,124}

\usepackage{hyperref}
\addto\extrasenglish{}
\addto\extrasenglish{}
\addto\extrasenglish{}
\hypersetup{
			colorlinks=true,
	citecolor=ToCgreen,
	linkcolor=Sepia,
	filecolor=Gred,
	urlcolor=Gred,
	pagebackref=true
	}

\usepackage{tikz}
\usetikzlibrary{arrows.meta, shapes}
\usepackage[framemethod=tikz]{mdframed}
\mdfsetup{
	roundcorner=4pt,
	nobreak=true,
	linewidth=.5,
	innerleftmargin=10pt,
	innerrightmargin=10pt,
	innertopmargin=10pt,
	innerbottommargin=10pt,
	skipabove=12,skipbelow=5pt
}

\usepackage[margin=2cm]{geometry}
\usepackage[colorinlistoftodos,prependcaption,textsize=footnotesize]{todonotes} \providecommand{\tdotoggle}{1}
\newcommand{\mytodo}[1]{\ifnum\tdotoggle=1{#1}\fi}

\newcommand{\tableoftodos}{\ifnum\tdotoggle=1 \listoftodos[Comments/To Do's] \fi}

\newcommand{\myignore}[1]{}

\newcommand{\TealBlue}[1]{\textcolor{TealBlue}{#1}}
\newcommand{\Peach}[1]{\textcolor{Peach}{#1}}
\newcommand{\Cyan}[1]{\textcolor{cyan}{#1}}
\newcommand{\Red}[1]{\textcolor{red}{#1}}
\newcommand{\Navy}[1]{\textcolor{Blue}{#1}}
\newcommand{\Blue}[1]{\textcolor{Blue}{#1}}
\newcommand{\Green}[1]{\textcolor{OliveGreen}{#1}}
\newcommand{\Black}[1]{\textcolor{black}{#1}}
\newcommand{\White}[1]{\textcolor{white}{#1}}
\newcommand{\Code}[1]{\Cyan{\texttt{#1}}}
\newcommand{\Gray}[1]{\textcolor{gray}{#1}}
\newcommand{\Magenta}[1]{\textcolor{magenta}{#1}}
\newcommand{\Maroon}[1]{\textcolor{Maroon}{#1}}

\newcommand{\tTealBlue}[1]{\ifnum\tdotoggle=1\TealBlue{#1}\else{#1}\fi}
\newcommand{\tPeach}[1]{\ifnum\tdotoggle=1\Peach{#1}\else{#1}\fi}
\newcommand{\tCyan}[1]{\ifnum\tdotoggle=1\Cyan{#1}\else{#1}\fi}
\newcommand{\tRed}[1]{\ifnum\tdotoggle=1\Red{#1}\else{#1}\fi}
\newcommand{\tNavy}[1]{\ifnum\tdotoggle=1\Navy{#1}\else{#1}\fi}
\newcommand{\tBlue}[1]{\ifnum\tdotoggle=1\Blue{#1}\else{#1}\fi}
\newcommand{\tGreen}[1]{\ifnum\tdotoggle=1\Green{#1}\else{#1}\fi}
\newcommand{\tBlack}[1]{\ifnum\tdotoggle=1\Black{#1}\else{#1}\fi}
\newcommand{\tWhite}[1]{\ifnum\tdotoggle=1\White{#1}\else{#1}\fi}
\newcommand{\tCode}[1]{\ifnum\tdotoggle=1\Code{#1}\else{#1}\fi}
\newcommand{\tGray}[1]{\ifnum\tdotoggle=1\Gray{#1}\else{#1}\fi}
\newcommand{\tMagenta}[1]{\ifnum\tdotoggle=1\Magenta{#1}\else{#1}\fi}
\newcommand{\tMaroon}[1]{\ifnum\tdotoggle=1\Maroon{#1}\else{#1}\fi}

\let\OLDthebibliography\thebibliography
\renewcommand\thebibliography[1]{
	\OLDthebibliography{#1}
	\setlength{\parskip}{1.3pt}
	}

\newcommand{\inparen}[1]{\left ( #1 \right )}

\newcommand{\infork}[1]{\left \{ \begin{matrix} #1 \end{matrix} \right .}

\newcommand{\inabs}[1]{\begin{vmatrix} #1 \end{vmatrix}}

\newcommand{\ceil}[1]{\left \lceil #1 \right \rceil}
\newcommand{\floor}[1]{\left \lfloor #1 \right \rfloor}

\newcommand{\set}[1]{\inbrace{#1}}

\renewcommand{\mod}[1]{~\inparen{\mathrm{mod}~#1}}

\let\boldm\bm
\renewcommand{\bm}{{\boldm m}}

\newcommand{\bx}{{\boldm x}}

\newcommand{\bA}{{\boldm A}}
\newcommand{\bB}{{\boldm B}}

\def\compactify{\itemsep=0pt \topsep=0pt \partopsep=0pt \parsep=0pt}
\let\latexusecounter=\usecounter

\newenvironment{Enumerate}
{\def\usecounter{\compactify\latexusecounter}
	\begin{enumerate}}
	{\end{enumerate}\let\usecounter=\latexusecounter}

\def\floor#1{\mathop{\left\lfloor#1\right\rfloor}}
\def\ceil#1{\mathop{\left\lceil#1\right\rceil}}
\def\abs#1{\left|#1\right|}
\def\p#1{\left(#1\right)}
\def\a#1{\langle#1\rangle}

\def\set#1{\left\{#1\right\}}

\newcommand{\parasc}[1]{\noindent \contour{black}{#1}}
\newcommand{\paragr}[1]{\noindent \textbf{\boldmath #1}}

\newcommand{\chara}{\mathds{1}}

\newcommand{\Z}{\ensuremath{\mathbb{Z}}}

\newcommand{\F}{\ensuremath{\mathbb{F}}}

\newcommand{\bbF}{\ensuremath{\mathbb{F}}}

\newcommand{\bbZ}{\ensuremath{\mathbb{Z}}}

\newcommand{\calC}{\mathcal{C}}

\newcommand{\calG}{\mathcal{G}}

\newcommand{\calM}{\mathcal{M}}

\newcommand{\calO}{\mathcal{O}}

\newcommand{\calQ}{\mathcal{Q}}
\newcommand{\calR}{\mathcal{R}}

\newcommand{\calT}{\mathcal{T}}

\newcommand{\calV}{\mathcal{V}}

\DeclarePairedDelimiter\bracks{[}{]}

\renewcommand{\vec}[1]{\boldsymbol{#1}}

\newcommand{\veca}{\ensuremath{\boldsymbol{a}}}

\newcommand{\vecf}{\ensuremath{\boldsymbol{f}}}
\newcommand{\vecg}{\ensuremath{\boldsymbol{g}}}
\newcommand{\vech}{\ensuremath{\boldsymbol{h}}}

\newcommand{\vecv}{\ensuremath{\boldsymbol{v}}}

\newcommand{\vecx}{\ensuremath{\boldsymbol{x}}}
\newcommand{\vecy}{\ensuremath{\boldsymbol{y}}}
\newcommand{\vecz}{\ensuremath{\boldsymbol{z}}}

\theoremstyle{plain}            \newtheorem{theorem}{Theorem}

\newtheorem{question}{Open Problem}

\newtheorem{lemma}{Lemma}[section]
\newtheorem{corollary}[lemma]{Corollary}

\newtheorem{claim}[lemma]{Claim}
\newtheorem{fact}[lemma]{Fact}

\newtheorem*{ntheorem}{Theorem}

\theoremstyle{definition}       \newtheorem{definition}[lemma]{Definition}

\theoremstyle{remark}           \newtheorem{remark}[lemma]{Remark}

\numberwithin{equation}{section}

\newenvironment{prevproof}[2]{\noindent {\bf {Proof of
			\autoref{#2}:}}}{\hfill $\blacksquare$\vskip \belowdisplayskip}

\contourlength{0.2pt}
\contournumber{20}

\newcommand{\bit}{\ensuremath{\set{0,1}}}

\DeclareMathOperator*{\poly}{\ensuremath\mathrm{poly}}

\newcommand{\problem}[1]{\ensuremath{\textsc{#1}}\xspace}

\newcommand{\class}[1]{\ensuremath{\mathsf{#1}}\xspace}

\renewcommand{\P}{\class{P}}
\newcommand{\FP}{\class{FP}}

\newcommand{\NP}{\class{NP}}
\newcommand{\FNP}{\class{FNP}}
\newcommand{\TFNP}{\class{TFNP}}

\newcommand{\PPP}{\class{PPP}}

\newcommand{\PPA}{\class{PPA}}
\newcommand{\PPAp}{\class{PPA}_{{p}}}

\newcommand{\PLS}{\class{PLS}}
\newcommand{\PPAD}{\class{PPAD}}
\newcommand{\PPADS}{\class{PPADS}}
\newcommand{\CLS}{\class{CLS}}

\newcommand{\PMOD}{\class{PMOD}}

\newcommand{\sis}{\problem{SIS}}

\newif\ifnotes\notestrue

\ifnotes
\usepackage{color}
\definecolor{mygrey}{gray}{0.50}
\newcommand{\notename}[2]{{\textcolor{red}{\footnotesize{\bf (#1:} {#2}{\bf
				) }}}}

\else

\newcommand{\notename}[2]{{}}

\fi

\mathchardef\mdash="2D

\renewcommand{\epsilon}{\varepsilon}

\newcommand{\pnum}{p}
\newcommand{\qnum}{q}

\newcommand{\amp}{\textsf{\upshape \&}}
\newcommand{\bamp}{~\textsf{\upshape \&}~}

\newcommand{\leafp}{\ensuremath\textsc{Leaf}_{\pnum}}
\newcommand{\lonelyp}{\ensuremath\textsc{Lonely}_{\pnum}}

\newcommand{\bipartitep}{\ensuremath{\textsc{Bipartite}_{\pnum}}}

\newcommand{\twoMatchingsp}{\ensuremath{\textsc{TwoMatchings}_{\pnum}}}
\newcommand{\sucBipartitep}{\ensuremath{\textsc{SuccinctBipartite}_{\pnum}}}

\newcommand{\leafq}{\ensuremath\textsc{Leaf}_{\qnum}}
\newcommand{\lonelyq}{\ensuremath\textsc{Lonely}_{\qnum}}

\newcommand{\bipartiteq}{\ensuremath{\textsc{Bipartite}_\qnum}}

\newcommand{\leaf}{\ensuremath\textsc{Leaf}}
\newcommand{\lonely}{\ensuremath\textsc{Lonely}}

\newcommand{\bipartite}{\ensuremath{\textsc{Bipartite}}}
\newcommand{\modn}{\ensuremath{\textsc{Mod}}}
\newcommand{\twoMatchings}{\ensuremath{\textsc{TwoMatchings}}}
\newcommand{\sucBipartite}{\ensuremath{\textsc{SuccinctBipartite}}}

\newcommand{\reducible}{\preceq}

\newcommand{\Chevalley}{\ensuremath\textsc{Chevalley}}

\newcommand{\CWT}{CWT\xspace}

\newcommand{\newclass}[1]{{\text{\upshape\sffamily #1}}\xspace}
\renewcommand{\P}{\newclass{P}}
\newcommand{\M}{\newclass{M}}

\newcommand{\secref}[1]{\hyperref[#1]{\S\ref{#1}}}
\newcommand{\eqnref}[1]{\hyperref[#1]{Eq. \ref{#1}}}

\newcommand{\lChevalleyp}{\textsc{GeneralChevalley}_{\pnum}}
\newcommand{\cChevalleyp}{\textsc{ChevalleyWithSymmetry}_{\pnum}}
\newcommand{\eChevalleyp}{\textsc{Chevalley}_{\pnum}}

\newcommand{\lChevalley}{\textsc{GeneralChevalley}}
\newcommand{\cChevalley}{\textsc{ChevalleyWithSymmetry}}
\newcommand{\eChevalley}{\textsc{Chevalley}}

\newcommand{\CW}{\mathsf{CW}}

\newcommand{\bis}{\ensuremath\textsc{BIS}}

\newcommand{\ACz}{\mathsf{AC}^{0}}
\newcommand{\ACzp}{\mathsf{AC}^{0}_{\F_p}} 
\usepackage[toc]{multitoc}

\begin{document}

\title{On the Complexity of Modulo-$q$ Arguments\\
	and the Chevalley--Warning Theorem}

\author{\setlength\tabcolsep{.7em}
\begin{tabular}{cccc}
Mika G\"o\"os &
Pritish Kamath&
Katerina Sotiraki&
Manolis Zampetakis\\[-1mm]
\small\slshape Stanford&
\small\slshape TTIC&
\small\slshape MIT&
\small\slshape MIT
\end{tabular}}

\date{\vspace{2mm}\large \today}

\maketitle
\thispagestyle{empty}

\begin{abstract}\noindent
We study the search problem class $\PPA_q$ defined as a
modulo-$q$ analog of the well-known \emph{polynomial parity argument} class
$\PPA$ introduced by Papadimitriou ({\small JCSS 1994}). Our first result
shows that this class can be characterized in terms of $\PPA_p$ for prime $p$.

Our main result is to establish that an \textit{explicit} version of a
search problem associated to the Chevalley--Warning theorem is complete for
$\PPA_p$ for prime $p$. This problem is \emph{natural} in that it does not
explicitly involve circuits as part of the input. It is the first such
complete problem for $\PPA_p$ when $p \ge 3$.

Finally we discuss connections between Chevalley-Warning theorem and the
well-studied \emph{short integer solution} problem and survey the structural
properties of $\PPA_q$.
\end{abstract}

\setcounter{page}{0}

\setlength{\cftbeforesecskip}{6pt}
\setcounter{tocdepth}{2}
\addtocontents{toc}{\protect\thispagestyle{empty}}

{\small \tableofcontents}

\clearpage

\section{Introduction} \label{sec:introduction}

The study of \emph{total $\NP\!$ search problems} ($\TFNP$) was initiated by
Megiddo and Papadimitriou \cite{megiddo91total} and Papadimitriou
\cite{papadimitriou94parity} to characterize the complexity of search problems
that have a solution for every input and where a given solution can be
efficiently checked for validity. Meggido and Papadimitriou
\cite{megiddo91total} showed that the notion of $\NP$-hardness is inadequate to
capture the complexity of total $\NP$ search problems. By now, this theory has
flowered into a sprawling jungle of widely-studied syntactic complexity classes
(such as $\PLS$~\cite{johnson88local},
$\PPA/\PPAD/\PPP$~\cite{papadimitriou94parity},
$\CLS$~\cite{daskalakis11continuous}) that serve to classify the
complexities of many relevant search problems.

  The goal of identifying \textit{natural}\footnote{Following the terminology
of many $\TFNP$ papers, including
\cite{grigni2001sperner, filosratsikas18consensus, filosratsikas19splitting, sotiraki18ppp},
a natural problem is one that does not have explicitly a circuit or a Turing
machine as part of the input.} complete problems for these complexity classes
lies in the foundation of this sub-field of complexity theory and not only gives
a complete picture of the computational complexity of the corresponding search
problems, but also provides a better understanding of the complexity
classes. Such natural complete problems have also been an essential middle-step
for proving the completeness of other important search problems, the same way
that the $\NP$-completeness of $\textsc{Sat}$ is an essential middle step in
showing the $\NP$-completeness of many other natural problems. Some known
natural complete problems for $\TFNP$ subclasses are: the $\PPAD$-completeness
of $\textsc{NashEquilibrium}$ \cite{daskalakis09nash}, the $\PPA$-completeness
of $\textsc{ConsensusHalving}$, $\textsc{NecklaceSplitting}$ and
$\textsc{HamSandwich}$ problems
\cite{filosratsikas18consensus, filosratsikas19splitting} and the
$\PPP$-completeness of natural problems related to lattice-based cryptography
\cite{sotiraki18ppp}. Finally, the theory of total search problems has found
connections beyond its original scope to areas like communication complexity and
circuit lower bounds \cite{goos19adventures}, cryptography
\cite{bitansky2015cryptographic, komargodski2019white, choudhuri2019finding}
and the Sum-of-Squares hierarchy \cite{kothari2018sum}.

  Our main result is to identify the first natural complete problem for
the classes $\PPA_q$, a variant of the class $\PPA$. We also illustrate the
relevance of these classes through connections with important search problems
from combinatorics and cryptography.
\medskip

\noindent \textbf{Class \boldmath $\PPA_q$.} The class $\PPA_q$ was defined, in
passing, by Papadimitriou~\cite[p.~520]{papadimitriou94parity}. It is a
modulo-$q$ analog of the well-studied \emph{polynomial parity argument}
class~$\PPA$ (which corresponds to $q = 2$). The class embodies the following
combinatorial principle:
\begin{quote}\itshape
	\centering
If a bipartite graph has a node of degree not a multiple of $q$,\\
then there is another such node.
\end{quote}
In more detail, $\PPA_q$ consists of all total $\NP$ search problems
reducible\footnote{Here, we consider a {\em many-one reduction}, which
is a polynomial time algorithm with one oracle query to the said
problem. In contrast, a {\em Turing reduction} allows polynomially many
oracle queries. See \autoref{sec:intro-struc} for a comparison.} to the
problem $\bipartite_q$ defined as follows. An instance of this problem
is a balanced bipartite graph $G = (V \cup U, E)$,
where $V \cup U = \{0, 1\}^n$ together with a designated vertex
$v^{\star} \in V \cup U$. The graph $G$ is implicitly given via a circuit $C$
that computes the neighborhood of every node in $G$. Let $\deg(v)$ be the
degree of the node $v$ in $G$. A valid solution is a node $v \in \{0, 1\}^n$
such that, either
\begin{itemize}
	\item[$\triangleright$] $v = v^{\star}$ satisfying 		$\deg(v) \equiv 0 \mod{q}$ [{\em Trivial Solution}] ;
	 or \smallskip
	\item[$\triangleright$] $v \ne v^{\star}$ 	satisfying $\deg(v) \not\equiv 0 \mod{q}$.
\end{itemize}

\noindent In \autoref{sec:definitions} we provide some other total search problems
($\lonely_q$, $\leaf_q$) that are reducible to and from $\bipartite_q$. Any one
of these problems could be used to define $\PPA_q$. In fact, $\lonely_q$ and
$\leaf_q$ are natural variants of the standard problems $\lonely$ and $\leaf$
which are used to define the class $\PPA$.
\bigskip
\clearpage

\noindent \textbf{Our contributions.} We illustrate the importance of the
complexity classes $\PPA_q$ by showing that many important search problems whose
computational complexity is not well understood belong to $\PPA_q$
(see \secref{sec:intro-open} for details). These problems span a wide range of
scientific areas, from algebraic topology to cryptography. For some of these
problems we conjecture that $\PPA_q$-completeness is the right notion to
characterize their computational complexity. The study of $\PPA_q$ is also
motivated from the connections to other important and well-studied classes like
$\PPAD$.

In this paper, we provide a systematic study of the complexity classes $\PPA_q$. Our main result is the identification of the first natural complete problem for
$\PPA_q$ together with some structural results. Below we give a more precise
overview of our results.
\begin{itemize}[leftmargin=1.5cm,itemsep=5pt]
	\item[$\underset{\text{\autoref{sec:characterization}}}{\text{\secref{sec:intro-char}}}$:]
		We characterize $\PPA_q$ in terms of $\PPA_p$ for prime $p$.
	\item[$\underset{\text{\autoref{sec:chevalley}}}{\text{\secref{sec:intro-complete}}}$:]
		Our main result is that an \textit{explicit}\footnote{Following the
		terminology in \cite{belovs17nullstellensatz}, by {\em explicit} we mean that
		the system of polynomials, which is the input of the computational problems we define, are given as a sum of monic monomials.} version of the Chevalley-Warning theorem is complete

		\vspace{-5.5pt} for $\PPA_p$ for prime $p$. This problem is \emph{natural}
		in that it does not involve circuits as part of the input and is the first
		known natural complete problem for $\PPA_p$ when $p \ge 3$.
	\item[$\underset{\text{\autoref{sec:simplification}}}{\text{\secref{sec:intro-depth}}}$:]
	  As a consequence of the $\PPA_p$-completeness of our natural problem, we show that
		restricting

		\vspace{-5.5pt} the input circuits in the definition of $\PPA_p$ to just constant
		depth arithmetic formulas doesn't change the power of the class.
	\item[$\underset{\text{\autoref{sec:applications}}}{\text{\secref{sec:intro-sis}}}$:]
	  We show a connection between $\PPA_q$ and the Short Integer Solution
	  ($\sis$) problem from

		\vspace{-5.5pt} the theory of lattices. This connection implies that $\sis$
		with constant modulus $q$ belongs to $\PPA_q \cap \PPP$, but also
		provides a polynomial time algorithm for solving $\sis$ when the modulus $q$
		is constant and has only $2$ and $3$ as prime factors.
	\item[$\underset{\text{\autoref{sec:structural}}}{\text{\secref{sec:intro-struc}}}$:]
	  We sketch how existing results already paint a near-complete picture of the relative power

	  \vspace{-5.5pt}  of $\PPA_p$ relative to other $\TFNP$ subclasses (via inclusions and oracle separations). We also show that $\PPA_q$ is closed under Turing
		reductions.
\end{itemize}

\noindent In \autoref{sec:intro-open}, we include a list of open problems that
illustrate the broader relevance of $\PPA_q$. We note that  a concurrent and
independent work by Hollender~\cite{hollender19ppak} also establishes the
structural properties of $\PPA_q$ corresponding to \secref{sec:intro-char} and
\secref{sec:intro-struc}.

\subsection{Characterization via Prime Modulus} \label{sec:intro-char}

  We show, in \autoref{sec:characterization}, that every class $\PPA_q$ is built
out of the classes $\PPA_p$ for $p$ a prime. To formalize this result, we
recall the operator `\amp' defined by Buss and
Johnson \cite[\S6]{buss12propositional}. For any two syntactic complexity
classes $\M_0$, $\M_1$ with complete problems $S_0$, $S_1$, the class
$\M_0 \bamp \M_1$ is defined via its complete problem $S_0 \bamp S_1$ where, on
input $(x, b) \in \{0, 1\}^* \times \bit$, the goal is to find a solution for
$x$ interpreted as an instance of problem $S_b$. Namely, if $b = 0$ then the output has to be a solution of $S_0$ with
input $x$, and otherwise it has to be a solution of $S_1$ with input $x$.
Intuitively speaking, $\M_1 \bamp \M_2$ combines
the powers of both $\M_1$ and $\M_2$. Note that
$\M_1 \cup \M_2 \subseteq \M_1 \bamp \M_2$. We can now formally
express our characterization result (where $p|q$ is the set of primes $p$
dividing $q$).

\begin{theorem}\label{thm:prime-characterization}
$\PPA_q = \amp_{p|q}\, \PPA_p$.
\end{theorem}

\noindent A special case of \autoref{thm:prime-characterization} is that
$\PPA_{p^k}=\PPA_p$ for every prime power $p^k$. Showing the inclusion
$\PPA_{p^k}\subseteq \PPA_p$ is the crux of our proof. This part of the theorem can be viewed as a total search problem
analog of the counting class result of Beigel and Gill~\cite{beigel92counting}
stating that $\mathsf{Mod}_{p^k}\P = \mathsf{Mod}_{p}\P$; ``an unexpected
result'', they wrote at the time. Throughout this paper, we use $q$ to denote any integer $\ge 2$ and $p$ to denote a prime integer.

\subsection{A Natural Complete Problem via Chevalley-Warning Theorem}
\label{sec:intro-complete}

There have been several works focusing on completeness results for the class $\PPA$ (i.e. $\PPA_2$). Initial works showed the
$\PPA$-completeness of (non-natural) total search problems corresponding to topological fixed point theorems
\cite{grigni2001sperner,aisenberg15tucker,deng16ppa}. Closer to our paper,
Belovs et al. \cite{belovs17nullstellensatz} show the $\PPA$-completeness of computational analogs of Combinatorial Nullstellensatz and the Chevalley--Warning Theorem, but which explicitly
involve a circuit as part of the input. More recently, breakthrough results
showed $\PPA$-completeness of problems without a circuit or a Turing Machine in
the input such as $\textsc{Consensus-Halving}$, $\textsc{Necklace-Splitting}$
and $\textsc{Ham-Sandwich}$
\cite{filosratsikas18consensus, filosratsikas19splitting} resolving an open
problem since the definition of $\PPA$ in \cite{papadimitriou94parity}.

  Our main contribution is to provide a natural complete problem for $\PPA_p$,
for every prime $p$; thereby also yielding a new complete problem for $\PPA$.
Our complete problem is an extension of the problem $\eChevalleyp$, defined by
Papadimitriou \cite{papadimitriou94parity}, which is a search problem
associated to the celebrated Chevalley-Warning Theorem. We first
present an abstract way to understand the proof of the Chevalley-Warning Theorem
that motivates the definition of our natural complete problem for
$\PPA_p$.

\subsubsection{Max-Degree Monic Monomials and Proof of Chevalley-Warning Theorem} \label{sec:maximum-degree}

  In 1935, Claude Chevalley \cite{chevalley35demonstration} resolved a
hypothesis stated by Emil Artin, that all finite fields are quasi-algebraically
closed. Later, Ewald Warning~\cite{warning36bemerkung}
proved a slight generalization of Chevalley's theorem. This
generalized statement is usually referred to as the Chevalley-Warning Theorem
(\CWT, for short). Despite its initial algebraic motivation, \CWT has found
profound applications in combinatorics and number theory as we discuss in
\secref{sec:intro-sis} (and \autoref{sec:applications}).

We now explain the statement of the Chevalley-Warning Theorem, starting with some notations. For any field $\bbF$ and any polynomial $f$ in a polynomial ring $\bbF[x_1, \ldots, x_n]$ we use $\deg(f)$ to represent the degree of $f$. We use $\vecx$ to succinctly denote the set of all variables $(x_1, \ldots, x_n)$ (the number of variables will always be $n$) and $\vecf$ to succinctly denote a system of polynomials $\vecf = (f_1, \ldots, f_m) \in \bbF[\vecx]^m$. We will often abuse notations to use $\bx$ to also denote assignments over $\bbF_p^n$. For instance, let $\calV_{\vecf} \coloneqq \set{\bx \in \bbF_p^n : f_i(\bx) = 0 \text{ for all } i \in [m]}$ be the set of all common roots of $\vecf$.

\theoremstyle{plain}
\newtheorem*{cwthm}{Chevalley-Warning Theorem}
\begin{cwthm}[\cite{chevalley35demonstration,warning36bemerkung}]
For any prime\footnote{While most of the results in this section
generalize to prime powers, we only consider prime fields for simplicity.} $p$ and polynomial system
$\vecf \in \bbF_p[\vecx]^m$ satisfying\vspace{-3.5mm}
\begin{equation}
\tag{CW Condition}
\sum_{i = 1}^m \deg(f_i) < n,\vspace*{-3mm}
\label{eq:CWcondition}
\end{equation}
it holds that $|\calV_{\vecf}| \equiv 0 \mod{p}$.
\end{cwthm}

\noindent Given a polynomial system $\vecf \in \bbF_p[\vecx]^m$,
the key idea in the proof of the Chevalley-Warning Theorem is the polynomial
\[ \CW_{\vecf}(\bx) \coloneqq \prod_{i = 1}^m \inparen{1 - f_i(\vecx)^{p - 1}} \mod{\set{x_i^p - x_i}_i}\,. \] Observe that $\CW_{\vecf}(\vecx) = 1$ if $\vecx \in \calV_{\vecf}$ and is $0$ otherwise. Thus, $|\calV_{\vecf}| \equiv \sum_{\vecx \in \bbF_p^n} \CW_{\vecf}(\bx) \mod{p}$.
The following definition informally describes a special type of monomial of
$\CW_{\vecf}$ that is of particular interest in the proof. For the precise definition, we
refer to \autoref{sec:chevalley}.

\begin{definition}[\textsc{Max-Degree Monic Monomials} (Informal)]
	\label{def:max-degreeMonomialInformal}
    Let $\vecf \in \bbF_p[\vecx]^m$.
    			A \textit{monic monomial} of $\CW_{\vecf}$ refers to a monic monomial obtained when symbolically expanding $\CW_{\vecf}$ as a sum of monic monomials. A monic monomial is said to be of \textit{max-degree} if it is
	$\prod_{j = 1}^n x_j^{p - 1}$.

\end{definition}

   In the above definition, it is important to consider the
\emph{symbolic expansion} of $\CW_{\vecf}$ and ignore any cancellation of coefficients that might occur. Observe that, although the
expansion of $\CW_{\vecf}$ is exponentially large in the description size of
$\vecf$, each monic monomial of $\CW_{\vecf}$ can be succinctly described as a combination of monic monomials of the polynomials $f_1, \ldots, f_m$. We formally
discuss this in \autoref{sec:chevalley}.

Using the definition of max-degree monic monomials, we state the main technical
lemma underlying the proof of \CWT (with proof in \autoref{sec:chevalley}).
\theoremstyle{plain}
\newtheorem*{cwthmprinciple}{Chevalley--Warning Lemma}
\begin{cwthmprinciple}
For any prime $p$ and $\vecf \in \bbF_p[\vecx]^m$,
\begin{equation}
	\tag{CW Lemma}
	|\calV_{\vecf}| ~\equiv~ (-1)^n \cdot  |\set{\mathrm{max\text{-}degree~monic~monomials~of~} \CW_{\vecf}}| \mod{p}
	\label{eq:mainPrinciple}
\end{equation}
\end{cwthmprinciple}

\noindent The Chevalley-Warning Theorem now follows by observing that if $\sum_{i = 1}^m \deg(f_i) < n$ then the number of max-degree monic monomials of $\CW_{\vecf}$ is zero. Hence, we get that $|\calV_{\vecf}| \equiv 0 \mod{p}$.

\subsubsection{Proofs of Cancellation}
\label{sec:proofsOfCancellation}

  From the proof sketch of CWT in the previous section, a slight generalization of CWT follows. In particular, $|\calV_{\vecf}| \equiv 0 \mod{p}$ if and only if
\begin{equation}
	\tag{Extended CW Condition}
	\inabs{\set{\mathrm{max\text{-}degree~monic~monomials~of~} \CW_{\vecf}}} ~\equiv~ 0 \mod{p},
		\label{eq:generalizedCondition}
\end{equation}
Thus, any condition on $\vecf$ that implies the \eqref{eq:generalizedCondition} can
 replace \eqref{eq:CWcondition} in the Chevalley-Warning
Theorem. Note that the \eqref{eq:generalizedCondition} is equivalent to all the max-degree monic monomials in $\CW_{\vecf}$ cancelling out.
Thus, we call any such condition on $\vecf$ that implies \eqref{eq:generalizedCondition} to be a
``\textit{proof of cancellation}'' for the system $\vecf$.

We can now reinterpret the result of Belovs et
al. \cite{belovs17nullstellensatz} in this framework of ``proof of cancellation'' conditions. In particular, \cite{belovs17nullstellensatz} considers the
case $p = 2$ and defines the problem $\textsc{PPA-Circuit-Chevalley}$, in which a ``proof of
cancellation'' is given in a specific form of circuits. These circuits describe the
system $(f_1, \ldots, f_m)$ in the $\textsc{PPA-Circuit-Chevalley}$ problem. It is then shown that $\textsc{PPA-Circuit-Chevalley}$ is $\PPA_2$-complete.

\subsubsection{Computational Problems Based on Chevalley-Warning Theorem}
\label{sec:computationalChevalleyWarning}

  Every ``proof of cancellation'' that is \textit{syntactically refutable} can be used to define
  a total search problem that lies in $\PPA_p$. By {\em syntactically refutable} we mean that whenever the
  ``proof of cancellation'' is false, there exists a small witness that certifies so. In this
  section, we define three computational problems with their corresponding ``proof of
  cancellation'': (1) the $\eChevalleyp$ problem defined by \cite{papadimitriou94parity}, (2) the
  $\lChevalleyp$ problem that is a generalization of $\eChevalleyp$, and (3) the problem
  $\cChevalleyp$ that we show to be $\PPA_p$-complete. All these problems are
  defined for every prime modulus $p$ and are natural in the sense that they do not explicitly
  involve a circuit or a Turing Machine in their input. In particular, the polynomial systems in the input are {\em explicit} in that they are given as a sum of monic monomials. \bigskip

\paragraph{Chevalley.}
  This is the direct computational analog of the Chevalley-Warning Theorem and was
defined by Papadimitriou \cite{papadimitriou94parity} as the following total
search problem:
\medskip

\noindent $\underbar{\eChevalley}_{\pnum}$

\noindent Given an explicit polynomial system $\vecf \in \bbF_p[\vecx]^m$, and an
$\vecx^{\star} \in \calV_{\vecf}$, output one of the following:\vspace*{-0.8mm}
\begin{itemize}
	\item[$\triangleright$] [{\em Refuting witness}] \eqref{eq:CWcondition} is not satisfied.
	\item[$\triangleright$] $\vecx \in \calV_{\vecf} \smallsetminus \set{\vecx^{\star}}$.
\end{itemize}

\noindent We will particularly consider a special case where all the $f_i$'s have zero constant term
({\em zecote}, for short). In this case, $\vecx^{\star} = \mathbf{0} \in \calV_{\vecf}$, so there is no need to explicitly include $\vecx^*$ in the input.

\paragraph{General Chevalley.} As mentioned already, we can define a search problem corresponding to any syntactically refutable condition that implies the
\eqref{eq:generalizedCondition}. One such condition is to directly assert that\vspace{-2.5mm}
\begin{equation}
	\tag{General CW Condition}
	\{\mathrm{max\text{-}degree~monic~monomials~of~} \CW_{\vecf}\} ~=~ \emptyset. \label{eq:generalCWCondition}
\end{equation}
In particular, note that \eqref{eq:CWcondition} implies this condition. Moreover, this condition is syntactically refutable by a max-degree monic monomial, which is efficiently representable as a combination of at most $m(p-1)$ monomials of the $f_i$'s. Thus, we can define the following total search problem generalizing $\eChevalleyp$.\medskip

\noindent $\underbar{\lChevalley}_{\pnum}$

\noindent Given an explicit polynomial system $\vecf \in \bbF_p[\vecx]^m$ and an
$\vecx^{\star} \in \calV_{\vecf}$, output one of the following:\vspace{-0.5mm}
\begin{itemize}
	\item[$\triangleright$] [{\em Refuting Witness}] A max-degree monic monomial of $\CW_{\vecf}$.
	\item[$\triangleright$] $\vecx \in \calV_{\vecf} \smallsetminus \set{\vecx^{\star}}$.
\end{itemize}
\medskip

\noindent While $\lChevalleyp$ generalizes $\eChevalleyp$, it does not
capture the full generality of \eqref{eq:generalizedCondition}. However \eqref{eq:generalizedCondition} is not syntactically refutable (in fact, it is $\mathsf{Mod}_p\mathsf{P}$--complete to decide\footnote{Circuit-SAT can be encoded as satisfiability of a polynomial system $\vecf \in \bbF_p[\vecx]^m$ by including a polynomial for each gate along with $\set{x_i^2 - x_i = 0}$ to ensure Booleanity. Thus, number of satisfiable assignments to the Circuit-SAT is $\equiv |\calV_{\vecf}| \mod{p}$, which is $0 \mod{p}$ iff the final coefficient of the max-degree monomial is $0$.} if the final coefficient of the max-degree monomial is $0$).

A natural question then is whether $\lChevalley_p$, or even $\Chevalley_p$, could already be $\PPA_p$--complete. We believe this to be unlikely because \eqref{eq:generalCWCondition} seems to fail in capturing other simple conditions that are syntactically refutable and yet imply \eqref{eq:generalizedCondition}. Namely, consider a permutation $\sigma \in S_n$ of the variables $x_1, \ldots, x_n$ of order $p$ (i.e. $\sigma^p$ is the identity permutation). Suppose that for every $\vecx \in \overline{\calV_{\vecf}}$, it holds that $\sigma(\vecx) \in \overline{\calV_{\vecf}} \smallsetminus \set{\vecx}$; in other words $\vecx, \sigma(\vecx), \sigma^2(\vecx), \ldots, \sigma^{p-1}(\vecx)$ are all distinct and in $\overline{\calV_{\vecf}}$ (where, $\sigma(\vecx)$ denotes the assignment obtained by permutating the variables of the assignment $\vecx$ according to $\sigma$); observe that this condition is syntactically refutable. This implies that the elements of $\overline{\calV_{\vecf}}$ can be partitioned into groups of size $p$ (given by the orbits of the action $\sigma$) and hence $|\overline{\calV_{\vecf}}| \equiv 0 \mod{p}$. Thus, such a $\sigma$ provides a syntactically refutable proof that $|\calV_{\vecf}| \equiv 0 \mod{p}$ and hence that \eqref{eq:generalizedCondition} hold.

Hence, we further generalize 
$\lChevalleyp$ into a problem that incorporates this additional ``proof of cancellation'' in the form of a permutation $\sigma \in S_n$.

\paragraph{Chevalley with Symmetry.}
We consider a union of two polynomial systems $\vecg  \in \bbF_p[\vecx]^{m_g}$ and $\vech \in \bbF_p[\vecx]^{m_h}$. Even if both $\vecg$ and
$\vech$ satisfy \eqref{eq:CWcondition}, the combined system $\vecf :=
(g_1, \ldots, g_{m_g}, h_1, \ldots, h_{m_h})$ might not satisfy \eqref{eq:CWcondition}
and it might even be the case that $|\calV_{\vecf}|$ is not a multiple of $p$. Thus, we need to bring in some additional conditions.

We start by observing that since $|\calV_{\vecf}| + |\overline{\calV_{\vecf}}| = p^n$, it holds that $|\calV_{\vecf}| \equiv 0 \mod{p}$ if and only if $|\overline{\calV_{\vecf}}| \equiv 0 \mod{p}$. Also note that, $|\overline{\calV_{\vecf}}| = |\overline{\calV_{\vecg}}| + |(\calV_{\vecg} \cap \overline{\calV_{\vech}})|$.

If $\vecg$ satisfies the \eqref{eq:generalCWCondition} then we have that $|\calV_{\vecg}| \equiv |\overline{\calV_{\vecg}}| \equiv 0 \mod{p}$. A simple way to enforce that $|\calV_{\vecg} \cap \overline{\calV_{\vech}}| \equiv 0 \mod{p}$ is to enforce a ``symmetry'', namely that its elements can be grouped into groups of size $p$ each. We impose this grouping with a permutation $\sigma \in S_n$ of the variables $x_1, \ldots, x_n$ of order $p$ such that for any $\vecx \in \calV_{\vecg} \cap \overline{\calV_{\vech}}$, it holds that $\sigma(\vecx) \in (\calV_{\vecg} \cap \overline{\calV_{\vech}}) \smallsetminus \set{\vecx}$; or in other words that $\vecx, \sigma(\vecx), \sigma^2(\vecx), \ldots, \sigma^{p-1}(\vecx)$ are all distinct and contained in $\calV_{\vecg} \cap \overline{\calV_{\vech}}$.

We now define the
following natural total search problem.

\bigskip

\noindent $\underbar{\cChevalley}_{\pnum}$

\noindent Given two explicit polynomial systems $\vecg \in \bbF_p[\vecx]^{m_g}$ and $\vech \in \bbF_p[\vecx]^{m_h}$, and an $\vecx^{\star} \in \calV_{\vecf}$ (where $\vecf \coloneqq (\vecg, \vech)$) and a
permutation $\sigma \in S_n$ of order $p$, output one of the following:\vspace{-0.5mm}
\begin{itemize}
	\item[$\triangleright$] [{\em Refuting Witness} -- 1] A max-degree monic monomial of $\CW_{\vecg}$.
  \item[$\triangleright$] [{\em Refuting Witness} -- 2] $\vecx \in \calV_{\vecg} \cap \overline{\calV_{\vech}}$ such that $\sigma(\vecx) \notin (\calV_{\vecg} \cap \overline{\calV_{\vech}}) \smallsetminus \set{\vecx}$.
	\item[$\triangleright$] $\vecx \in \calV_{\vecf} \smallsetminus \set{\vecx^{\star}}$.
\end{itemize}
\medskip

\noindent The above problem is natural, because the input consists of a system of polynomial in an
explicit form, i.e. as a sum of monic monomials, together with a permutation in $S_n$ given say in
one-line notation. Also, observe that when $\vech$ is empty, the above problem coincides with $\lChevalley_p$ (since $\overline{\calV_{\vech}} = \emptyset$ when $\vech$ is empty). Our main result is the following (proved in \autoref{sec:chevalley}).

\begin{restatable}{theorem}{completeness}
\label{thm:combinedChevalleyPPApCompleteness}
For any prime $p$, $\cChevalleyp$ is
$\PPA_p$-complete.
\end{restatable}

\subsection{Complete Problems via Small Depth Arithmetic Formulas} \label{sec:intro-depth}

While the $\cChevalleyp$ problem may seem somewhat contrived, the importance of its $\PPAp$-completeness is illustrated by our next result (proved in \autoref{sec:simplification}) showing that we can reformulate
any of the proposed definitions of $\PPAp$, by restricting the circuit in
the input to be just constant depth arithmetic formulas with gates
$\times \mod{p}$ and $+ \mod{p}$ (we call this class $\ACz_{\F_p}$). This result is analogous to the
$\NP$-completeness of $\textsc{SAT}$ which basically shows that
$\textsc{CircuitSAT}$ remains $\NP$-complete even if we restrict the input
circuit to be a (CNF) formula of depth $2$.

\begin{theorem}\label{thm:simplification}
	$\lonely_p$/$\bipartite_p$/$\leaf_p$ with $\ACz_{\F_p}$ input circuits are $\PPAp$--complete.
\end{theorem}
\noindent We hope that this theorem will be helpful in the context of proving $\PPAp$-hardness of other problems.
There it would be enough to consider only constant depth arithmetic formulas (and hence $\mathsf{NC}^{1}$ Boolean formulas) in the
definitions of $\PPAp$ as opposed to unbounded depth circuits. Such a
simplification has been a key-step for proving hardness results for other
$\TFNP$ subclasses, e.g. in the $\PPAD$-hardness proofs of $\textsc{approximate}$-$\textsc{Nash}$
(cf. \cite{rubinstein16}).

\subsection{Applications of Chevalley-Warning}
\label{sec:intro-sis}

  Apart from its initial algebraic motivation, the Chevalley-Warning
theorem has been used to derive several non-trivial combinatorial results. Alon
et al. \cite{alon1984regular} show that adding an extra edge to any $4$-regular graph forces it to contain a 3-regular subgraph. More generally, they prove that
certain types of ``almost'' regular graphs contain regular subgraphs. Another
application of \CWT is in proving \textit{zero-sum theorems}
similar to the Erd\"os-Ginzburg-Ziv Theorem. A famous such application
is the proof of Kemnitz's conjecture by Reiher \cite{Reiher07}.

We define two computational problems that we show are reducible to $\Chevalley_p$ and suffice for
proving most of the combinatorial applications of the Chevalley-Warning Theorem mentioned above (for a certain range of parameters $n$ and $m$). Both involve finding solutions to a system of linear equations modulo $q$, given as $\bA \vecx \equiv \vec{0} \mod{q}$ for $\bA \in \bbZ^{m \times n}$.
  \begin{itemize}
  	\item[$\triangleright$] $\bis_q$: Find $\vecx \in \bit^n$ satisfying $\vecx \ne \vec{0}$ and $\bA \vecx \equiv \vec{0} \mod{q}$.
    \smallskip
    \item[$\triangleright$] $\sis_q$: Find $\vecx \in \set{-1, 0, 1}^n$ satisfying $\vecx \ne \vec{0}$ and $\bA \vecx \equiv \vec{0} \mod{q}$.
  \end{itemize}
  \noindent The second problem is a special case of the well-known {\em short integer solution} problem in
  $\ell_\infty$ norm. Note that, when $n > m \cdot \log_2 q$, the totality of $\sis_q$ is guaranteed by pigeonhole principle;
  that is, $\sis_q$ is in $\PPP$ in this range of parameters. We are interested in identifying the range
  of parameters that places this problem in $\PPA_q$ --- see \hyperref[def:bis]{Definitions \ref{def:bis}} and \ref{def:sis} for the precise range of parameters $n$ and $m$ that we consider. In \autoref{thm:sis/bisPPAq}, we prove a formal version of the following:

\begin{ntheorem}[Informal] ~~~
\label{nthm:sis}
  For a certain range of parameters $n, m$, it holds that
  \begin{enumerate}
    \item For all primes $p$ : $\bis_p$ and $\sis_p$ are Karp-reducible to $\eChevalleyp$, hence are in $\PPA_p$. \\[-10pt]
    \item For all $q$ :  $\bis_q$ and $\sis_q$ are Turing-reducible to any $\PPA_q$--complete problem. \\[-10pt]
    \item For all $k$ : $\bis_{2^k}$ is solvable in  polynomial time. \\[-10pt]
    \item For $k$ and $\ell$ : $\sis_{2^k3^{\ell}}$ is solvable in  polynomial time. \\[-10pt]
  \end{enumerate}
\end{ntheorem}

\noindent Even though the $\sis_q$ problem is well-studied in lattice theory, not many
results are known in the regime where $q$ is a constant and the number of
variables depends linearly on the number of equations. Part (1) of the
above theorem establishes a reduction from $\sis_\pnum$ to
$\eChevalley_\pnum$ for prime $p$.
Part (2) follows by a bootstrapping method that allows us to combine algorithms for $\sis_{q_1}$ and $\sis_{q_2}$ to give an algorithm for $\sis_{q_1q_2}$ (for a certain regime for parameters $n$ and $m$).
Finally Parts (3) and (4) results follow by using this bootstrapping method along with the observation that Gaussian elimination provides valid solutions for $\bis_2$ (hence also $\sis_2$) and for $\sis_3$.

\subsection{Structural properties} \label{sec:intro-struc}

\paragraph{Relation to other classes.}
Buss and Johnson \cite{buss12propositional,johnson11reductions} had defined a class $\PMOD_q$ which turns out to be slightly weaker than $\PPA_q$ (refer to \autoref{sec:structural}). Despite this slight difference between the definitions of $\PPA_q$ and $\PMOD_q$, we can still deduce statements about $\PPA_q$ from the work of \cite{johnson11reductions}. In particular, it follows that $\PPAD \subseteq \PPA_q$ (refer to \autoref{sec:ppad-inclusion}).

More broadly, a near-complete picture of the power of $\PPA_q$ relative to other subclasses of $\TFNP$ is summarized in \autoref{fig:classes}. These relationships (inclusions and oracle separations) mostly follow from prior work in proof complexity~\cite{beame98more,buss01linear,johnson11reductions,goos19adventures} (refer to \autoref{sec:oracle-sep}).

\begin{figure}[tp]
\centering
\begin{tikzpicture}
\tikzset{
	inner sep=0,outer sep=3,
	mydash/.style={dashed,color=white!30!black},
	new/.style={line width=1pt,color=myGold},
	johnson/.style={line width=1pt,color=myPurple}}

\begin{scope}[scale=0.7]
\node (FP) at (0,0) {$\FP$};
\node (CLS) at (0,1.5) {$\CLS$};
\node (PPAD) at (-2.25,3) {$\PPAD$};
\node (PPADS) at (-4.5,5) {$\PPADS$};
\node (PLS) at (4.5,7.5) {$\PLS$};
\node (PPP) at (-4.5,7.5) {$\PPP$};
\node (PPA) at (-1.5,7.5) {$\PPA$};
\node at (0,7.5) {$\cdots$};
\node (PPAp) at (1.5,7.5) {$\PPA_p$};
\node (capPPAp) at (-0.25, 5) {$\bigcap_p \PPA_p$};
\node (TFNP) at (0,10) {$\TFNP$};
\end{scope}

\path[-{Stealth[length=5pt]},line width=.4pt,black]
(FP) edge (CLS)
(CLS) edge (PPAD)
(CLS) edge[bend right=15] (PLS)
(capPPAp) edge (PPA)
(capPPAp) edge (PPAp)
(PPAD) edge (PPADS)
(PPADS) edge (PPP)
(PLS) edge (TFNP)
(PPP) edge (TFNP)
(PPA) edge (TFNP)
(PPAp) edge (TFNP)
(CLS) edge[mydash,bend right=28] (FP)
(PPAD) edge[mydash,bend right=15] (PLS)
(PPADS) edge[mydash,bend right=10] (PPA)
(PPP) edge[mydash,bend right=28] (PPADS)
(PLS) edge[mydash,bend left=27] (PPA)
(PPA) edge[mydash,bend left=0] (PPP)
(PPAD) edge[new] (capPPAp)
(capPPAp) edge[mydash,new,bend left=16] (PPAD)
(PPADS) edge[mydash,new,bend right=10] (PPAp)
(PLS) edge[mydash,new,bend left=10] (PPAp)
(PPAp) edge[mydash,new,bend left=23] (PPP);

\path[{Stealth[length=5pt]}-{Stealth[length=5pt]},line width=.4pt,black]
(PPA) edge[mydash,new,bend right=15] (PPAp);
\end{tikzpicture}

\vspace{2mm}
\caption{The landscape of $\TFNP$ subclasses. A solid arrow $\M_1\rightarrow\M_2$ denotes $\M_1\subseteq\M_2$, and a dashed arrow $\M_1\dashrightarrow\M_2$ denotes an oracle separation: $\M_1^{\calO} \nsubseteq\M_2^{\calO}$ relative to some oracle $\calO$. The relationships involving $\PPA_p$ are highlighted in \textcolor{myGold}{yellow}. See \autoref{sec:structural} for details.
}
\label{fig:classes}
\vspace{2mm}
\end{figure}

\paragraph{Closure under Turing reductions.}
Recall that $\TFNP$ subclasses are defined as the set of all total search problems that are {\em many-one} reducible (aka Karp--reducible) to the corresponding complete problems. One can ask whether more power is gained by allowing {\em Turing reductions}, that is, polynomially many oracle queries to the corresponding complete problem. Buss and Johnson \cite{buss12propositional} showed that $\PLS$, $\PPAD$, $\PPADS$, $\PPA$ are closed under Turing reductions (with a notable exception of $\PPP$, which remains open). We show this for $\PPA_p$ when $p$ is a prime.

\begin{theorem}\label{thm:turing-closure}
$\FP^{\PPA_p} = \PPA_p$ for every prime $p$.
\end{theorem}

\noindent By contrast, it follows from \cite[\S6]{buss12propositional} that $\PPA_q$ is not closed under {\em black-box} Turing reductions for non-prime powers $q$. See \autoref{sec:turing-closure} for details.

\subsection{Open questions} \label{sec:intro-open}

\paragraph{Factoring.} It has been shown that $\textsc{Factoring}$ reduces to $\PPP$-complete problems as well as to $\PPA$-complete problems \cite{bureshoppenheim06factoring,jerabek16integer}, albeit under randomized reductions (which can be derandomized assuming the Generalized Reimann Hypothesis). It has been asked whether in fact $\textsc{Factoring}$ could be reduced to $\PPAD$-complete problems \cite{jerabek16integer}. As a step towards this problem, we propose the following question.

\begin{question}
Is $\textsc{Factoring}$ in $\PPA_p$ for all primes $p$
(perhaps under randomized reductions)?
\end{question}

\noindent This is clearly an easier problem since $\PPAD\subseteq \PPA_p$. Interestingly, note that there exists an oracle $\calO$ relative to which $\bigcap_p \PPA_p^{\calO} \nsubseteq \PPAD^{\calO}$. Thus, the above problem, even if established for all prime $p$, is still weaker than showing that $\textsc{Factoring}$ reduces to $\PPAD$-complete problems.

\paragraph{Necklace Splitting.} The $q\textsc{-Necklace-Splitting}$ problem is defined as follows: There is an open necklace\footnote{an ``open necklace'' means that the beads form a string, not a cycle} with $q \cdot a_i$ beads of color $i$, for $i \in [n]$. The goal is to cut the necklace in $(q-1) \cdot n$ places and partition the resulting substrings into $k$ collections, each containing precisely $a_i$ beads of color $i$ for each $i \in [n]$.

The fact that such a partition exists was first shown in the case of $q=2$ by Goldberg and West \cite{goldberg85bisection} and by Alon and West \cite{alon86borsukulam}. Later, Alon \cite{alon87splitting} proved it for all $q \ge 2$. As mentioned before, Filos-Ratsikas and Goldberg \cite{filosratsikas19splitting}
showed that the $2\textsc{-Necklace-Splitting}$ problem is $\PPA$-complete. Moreover, they put forth the following question (which we strengthen further).

\begin{question}
Is $q\textsc{-Necklace-Splitting}$ in $\PPA_q$? More strongly, is it $\PPA_q$-complete?
\end{question}

\noindent While we do not know how to prove/disprove this yet, we point out that it was also shown in \cite{filosratsikas19splitting} that $2^k\textsc{-Necklace-Splitting}$ is in fact in $\PPA_2$. This is actually well aligned with this conjecture since we showed that $\PPA_{2^k} = \PPA_2$ (\autoref{thm:prime-characterization}).

\paragraph{B\'ar\'any-Shlosman-Sz\"ucs theorem.} Alon's proof of the $q$-Necklace-Splitting theorem \cite{alon87splitting} was topological and used a certain generalization of the Borsuk-Ulam theorem due to B\'ar\'any, Shlosman and Sz\"ucs \cite{barany81topological}. Since the computational $\textsc{Borsuk-Ulam}$ problem is $\PPA$-complete, we could ask a similar question about this generalization.

\begin{question}
Is $\textsc{B\'ar\'any-Shlosman-Sz\"ucs}_p$ problem in $\PPA_p$ (perhaps even $\PPA_p$-complete)?
\end{question}

\paragraph{Applications of Chevalley-Warning Theorem.} We conclude with some interesting directions for further exploring the
connections of $\eChevalley$ with other computational problems.

\begin{question}
	  Does $\sis_q$ admit worst-to-average case reductions
	to other lattice problems in our range of parameters? Or is it
	average-case hard assuming standard cryptographic
	assumptions, e.g. the ``learning with errors'' assumption?
\end{question}

\noindent If resolved positively, the above would serve as evidence of the average-case hardness for the class $\PPA_p$, similar to the evidence that we have for $\PPA$ by reduction from $\textsc{Factoring}$.

\begin{question}
	  For all primes $p$, is $\eChevalley_p$ reducible to $\bis_{p}$?
\end{question}

\begin{question}
	  For all $q$, is there a non-trivial regime of parameters $n$, $m$ where $\bis_{q}$ is solvable in polynomial time?
\end{question} 

\section{\boldmath The class $\PPA_q$}
\label{sec:definitions}

\paragraph{\boldmath Search Problems in $\FNP$ and $\TFNP$.} A search problem
in $\FNP$ is defined by a polynomial time computable relation
$\calR\subseteq \bit^* \times \bit^*$, that is, for every $(x,y)$, it is
possible to decide whether $(x,y) \in \calR$ in $\poly(|x|, |y|)$ time. A
solution to the search problem on input $x$ is a $y$ such that
$|y| = \poly(|x|)$ and $(x,y) \in \calR$. For convenience, define
$\calR(x) \coloneqq \set{y : (x,y) \in \calR}$. A search problem is {\em total}
if for every input $x \in \bit^*$, there exists $y \in \calR(x)$ such that $|y|
\le \poly(|x|)$. $\TFNP$ is the class of all total search problems in $\FNP$.

\paragraph{Reducibility among search problems.} A search problem $\calR_1$ is
{\em Karp-reducible} (or {\em many-one reducible}) to a search problem
$\calR_2$, or $\calR_1 \preceq \calR_2$ for short, if there exist polynomial-time computable functions $f$ and $g$ such that given any instance $x$ of $\calR_1$, $f(x)$ is an instance of $\calR_2$ such that for any $y \in \calR_2(f(x))$, it holds that $g(x, f(x), y) \in \calR_1(x)$.

On the other hand, we say that $\calR_1$ is {\em Turing-reducible} to $\calR_2$, or $\calR_1 \preceq_T \calR_2$ for short, if there exists a polynomial-time oracle Turing machine that on input $x$ to $\calR_1$, makes oracle queries to $\calR_2$, and outputs a $y \in \calR_1(x)$. In this paper, we primarly deal with Karp-reductions, except in \autoref{sec:turing-closure}, where we compare the two different notions of reductions in the context of $\PPA_q$.

\paragraph{\boldmath $\PPA_q$ via complete problems.} We describe several total search problems (parameterized by $q$) that we show to be inter-reducible. $\PPA_q$ is then defined as the set of all search problems reducible to either one of the search problems defined below.

Recall that Boolean circuits take inputs of the form $\bit^n$ and operate using ($\wedge$, $\vee$, $\lnot$) gates. In addition, we'll also consider circuits acting on inputs in $[q]^n$. We interpret the input to be of the form $(\bit^{\ceil{\log q}})^n$, where the circuit will be evaluated only on inputs where each block of $\ceil{\log q}$ bits represents a element in $[q]$. In the case where $q$ is a prime, we could also represent the circuit as $C:\bbF_q^n \to \bbF_q^n$ with arbitrary gates of the form $g: \bbF_q^2 \to \bbF_q$. However, we can simulate any such gate with $\poly(q)$ many $+$ and $\times$ operations (over $\bbF_q$) along with a constant $(1)$ gate. Hence, in the case of prime $q$, we'll assume that such circuits are composed of only $(+, \times, 1)$ gates.\\[1mm]

\newcommand{\mydef}[7]{
\noindent
\begin{minipage}{\linewidth}
\begin{definition}
	\label{def:#1} \contour{black}{(#2)}\vspace{-5pt}
	\begin{itemize}[leftmargin=2.15cm,itemsep=4pt]
		\item[\slshape Principle:] #3
		\item[\slshape Object:] #4
		\item[\slshape Inputs:] #5
		\item[\slshape Encoding:] #6
		\item[\slshape Solutions:] #7
	\end{itemize}
\end{definition}
\end{minipage}\vspace{5mm}
}

\mydef{bipartite}{$\bipartite_q$}
{A bipartite graph with a non-multiple-of-$q$ degree node has another such node.}
{Bipartite graph $G = (V \cup U, E)$. Designated vertex $v^* \in V$}
{$\triangleright$ $C:\bit^n \to (\bit^n)^k$, with $(\bit^n)^k$ interpreted as a $k$-subset of $\bit^n$\\
	$\triangleright$ $v^* \in \set{0}\times\bit^{n-1}$ (usually $0^n$)}
{$V\coloneqq\set{0}\times\bit^{n-1}$, $U\coloneqq\set{1}\times\bit^{n-1}$,\\
	$E \coloneqq \set{(v,u) : v \in V \cap C(u) \text{ and } u \in U \cap C(v)}$}
{$v^*$ if $\deg(v^*) \equiv 0 \mod{q}$ and\\
	$v \ne v^*$ if $\deg(v) \not\equiv 0 \mod{q}$}

\mydef{lonely}{$\lonely_q$}
{A $q$-dimensional matching on a non-multiple-of-$q$ many vertices has an isolated node.}
{$q$-dimensional matching $G = (V, E)$. Designated vertices $V^* \subseteq V$ with $|V^*| \le q-1$}
{$\triangleright$ $C:[q]^n \to [q]^n$\\
	$\triangleright$ $V^* \subseteq [q]^n$ with $|V^*| \le q-1$}
{$V\coloneqq[q]^n$. For distinct $v_1, \ldots, v_q$, edge $e\coloneqq\set{v_1, \ldots, v_q}\in E$ if $C(v_i) = v_{i+1}$, $C(v_q) = v_1$}
{$v \in V^*$ if $\deg(v) = 1$ and\\
	$v \notin V^*$ if $\deg(v) = 0$}

\mydef{leaf}{$\leaf_q$}
{A $q$-uniform hypergraph with a non-multiple-of-$q$ degree node has another such node.}
{$q$-uniform hypergraph $G = (V, E)$. Designated vertex $v^* \in V$}
{$\triangleright$ $C:\bit^n \to (\bit^{nq})^q$, with $(\bit^{nq})^q$ interpreted as $q$ many $q$-subsets of $\bit^n$\\
	$\triangleright$ $v^* \in \bit^n$ (usually $0^n$)}
{$V\coloneqq\bit^n$. For distinct $v_1, \ldots, v_q$, edge $e\coloneqq\set{v_1, \ldots, v_q}\in E$ if $e \in C(v)$ for all $v \in e$}
{$v^*$ if $\deg(v) \equiv 0 \mod{q}$ and\\
	$v \ne v^*$ if $\deg(v) \not\equiv 0 \mod{q}$}

\noindent We remark that $\lonely_q$ and $\leaf_q$ are modulo-$q$ analogs of the $\PPA$-complete problems $\lonely$ and $\leaf$ \cite{papadimitriou94parity,beame98relative}. We prove the following theorem in \autoref{sec:equivalences}.

\begin{theorem}\label{thm:PPAqEquivalentProblems}
The problems $\bipartiteq$, $\lonelyq$ and $\leafq$ are inter-reducible.
\end{theorem}

\begin{remark}[Simplifications in describing reductions.]\label{rem:simplifications}
	We will use the following simple conventions repeatedly, in order to simplify the descriptions of reductions between different search problems.
	\begin{enumerate}
						\item We will often use ``algorithms'', instead of ``circuits'' to encode our hypergraphs. It is standard to simulate polynomial-time algorithms by polynomial sized circuits.
		\item While our definitions require vertex sets to be of a very special form, e.g. $\bit^n$ or $[q]^n$, it will hugely simplify the description of our reductions to let vertex sets be of arbitrary sizes. This is not a problem as long as the vertex set is efficiently indexable, that is, elements of $V$ must have a $\poly(n)$ length representation and we must have a poly-time computable bijective map $\varphi : V \to [|V|]$, whose inverse is also poly-time computable. We could then use $\varphi$ to interpret the first $|V|$ elements of $\bit^n$ (or $[q]^n$) as vertices in $V$.

		Note that, we need to ensure that no new solutions are introduced in this process. In the case of $\bipartite_q$ or $\leaf_q$, we simply leave the additional vertices isolated and they don't contribute any new solutions. In the case of $\lonely_q$ we need to additionally ensure that $|V| \equiv 0 \mod{q}$, so that we can easily partition the remaining vertices into $q$-uniform hyperedges thereby not introducing any new solutions.
		\item The above simplification gives us that all our problems have an {\em instance-extension property} (cf. \cite{bureshoppenheim04relativized}) -- this will be helpful in proving \autoref{thm:turing-closure}.
		\item To simplify our reductions even further, we'll often describe the edges/hyperdges directly instead of specifying how to compute the neighbors of a given vertex. This is only for simplicity and it will be easy to see how to compute the neighbors of any vertex locally.
	\end{enumerate}
\end{remark}

\begin{figure}[t]

	\centering
	\begin{tikzpicture}[scale=0.8, transform shape]
	\def\hgap{2}
	\def\vgap{2.25}
	\draw[fill=Gblue!15, draw=none] (-7,2.4*\vgap) rectangle (1.8,-0.7);
	\draw[fill=Ggreen!15, draw=none] (9,2.4*\vgap) rectangle (1.8,-0.7);
	\draw[fill=orange!15, draw=none] (9,0.3*\vgap) rectangle (1.8,-0.7);
	\node (bip) at (-5.5,2*\vgap) {$\bipartite_p$};
	\node (leaf) at (-5.5,0) {$\leaf_p$};
	\node (leaf-p) at (-5.5,\vgap) {$\leaf'_p$};
	\node (lone) at (-1,0) {$\lonely_p$};
		\node (sbip) at (-1,2*\vgap) {$\sucBipartite_p$};
	\node (twomat) at (-1,\vgap) {$\twoMatchings_p$};

	\def\vgap{1.5}
	\node (combCh) at (5.5,4.5) {$\cChevalley_p$};
	\node (labelCh) at (6,4.5-\vgap) {$\lChevalley_p$};
	\node (Ch) at (6, 4.5-2*\vgap) {$\Chevalley_p$};
	\node (SIS) at (6, 4.5-3*\vgap) {$\sis_p$};

	\path[{Stealth[length=5pt]}-{Stealth[length=5pt]},line width=.4pt,black]
	(bip) edge (leaf-p)
	(leaf-p) edge (leaf)
	(leaf) edge (lone);

	\path[-{Stealth[length=5pt]},line width=.4pt,black]
	(bip) edge (sbip)
	(sbip) edge (twomat)
	(twomat) edge (lone)
	(lone) edge[bend right=0] ([shift={(-1.5,0)}]combCh.south)
	(combCh) edge (sbip)
	(labelCh) edge ([shift={(0.5,0)}]combCh.south)
	(Ch) edge (labelCh)
	(SIS) edge (Ch);
	\end{tikzpicture}
	\caption{Total search problems studied in this work. An arrow
	$A \rightarrow B$ denotes a reduction $A \reducible B$ that we establish.
	Problems in the \textcolor{Gblue}{blue} region are non-natural problems,
	which are all complete for $\PPA_p$. Problems in the
	\textcolor{Ggreen}{green} region are natural problems of which
	$\cChevalley_p$ is the one we show to be $\PPA_p$--complete. The problem
	in the \textcolor{orange}{orange} region is a cryptographically relevant problem.
				}
	\label{fig:ppa-landscape}
\end{figure}
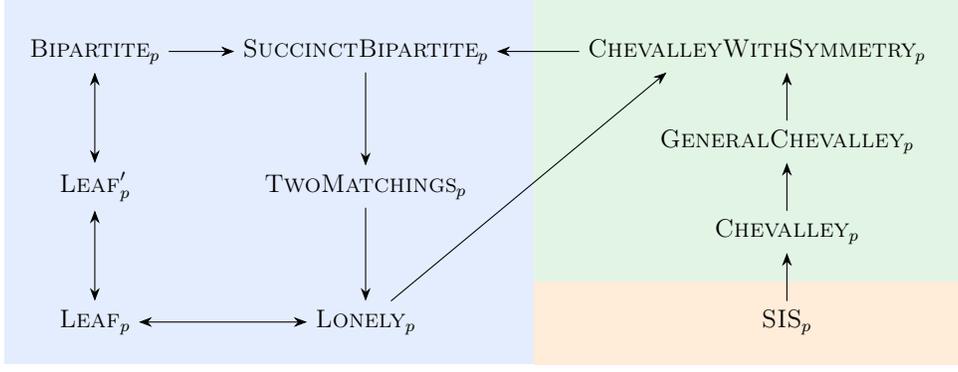
 
\section{Characterization via Primes} \label{sec:characterization}

In this section we prove \autoref{thm:prime-characterization}, namely $\PPA_q = \amp_{p|q}\, \PPA_p$. The theorem follows by combining the following two ingredients.
\begin{itemize}[leftmargin=1.5cm,itemsep=3pt]
\item[\secref{sec:coprime}:]
$\PPA_{q r} = \PPA_q\bamp\PPA_r$ for any coprime $q$ and $r$.
\item[\secref{sec:prime-power}:]
$\PPA_{p^k} = \PPA_p$ for any prime power $p^k$.
\end{itemize}

\subsection{Coprime case} \label{sec:coprime}
\paragraph{\boldmath$\PPA_{qr} \supseteq \PPA_q \bamp\PPA_r$.}
We show that $\lonely_q \bamp\lonely_r$ reduces to $\lonely_{qr}$. Recall that an instance of $\lonely_q \bamp\lonely_r$ is a tuple $(C,V^*, b)$ where $(C,V^*)$ describes an instance of either $\lonely_q$ or $\lonely_r$ as chosen by $b\in\{0,1\}$. Suppose wlog that $b=0$, so the input encodes a $q$-dimensional matching $G = (V,E)$ over $V = [q]^n$ with designated vertices $V^*\subseteq V$, $|V^*| \not\equiv 0 \mod{q}$. We can construct a $qr$-dimensional matching $\overline{G} = (\overline{V},\overline{E})$ on vertices $\overline{V}\coloneqq V \times [r]$ as follows: For every hyperedge $e \coloneqq \set{v_1, \ldots, v_q} \in E$, we include the hyperedge $e \times [r]$ in $\overline{E}$. We let the designated vertices of $\overline{G}$ be $\overline{V}^* \coloneqq V^* \times [r]$. Note that $|\overline{V}^*| \not\equiv 0 \mod{qr}$. It is easy to see that a vertex $(v,i)$ is isolated in $G'$ iff $v$ is isolated in $G$. This completes the reduction since $\overline{V}$ is efficiently indexable, and the neighbors of any vertex in $\overline{V}$ are locally computable using black-box access to $C$. 

\paragraph{\boldmath $\PPA_{qr} \subseteq \PPA_q \bamp\PPA_r$.}
We show that $\bipartite_{qr}$ reduces to $\bipartite_q \bamp\bipartite_r$. Our input instance of $\bipartite_{qr}$ is a circuit $C : \bit^n \to (\bit^n)^k$ that encodes a bipartite graph $G = (V\cup U, E)$ with a designated node $v^* \in V$. If $\deg(v^*) \equiv 0 \mod{qr}$, then we already have solved the problem and no further reduction is necessary. Otherwise, if $\deg(v^*) \not\equiv 0 \mod{qr}$, we have, by the coprime-ness of $q$ and $r$, that either $\deg(v^*) \not\equiv 0 \mod{q}$ or $\deg(v^*) \not\equiv 0 \mod{r}$. In the first case (the second case is analogous), we can simply view $(G,v^*)$ as an instance of $\bipartite_q$, since vertices with degree $\not\equiv 0 \mod{q}$ in $G$ are also solutions to $\bipartite_{qr}$.

\subsection{Prime power case} \label{sec:prime-power}

$\PPA_{p^k} \supseteq \PPA_{p}$ follows immediately from our proof of $\PPA_{qr} \supseteq \PPA_q \bamp\PPA_r$, which didn't require that $q$ and $r$ be coprime. It remains to show $\PPA_{p^k} \subseteq \PPA_{p}$. We exploit the following easy fact.

\begin{fact}\label{fact:primes}
For all primes $p$, it holds that,\vspace*{-1mm}
\begin{align}
\text{for integers } t, c > 0: && \binom{c \cdot p^{t}}{p^t} \equiv 0 \mod{p} &\quad \text{if and only if }\ c \equiv 0 \mod{p}\label{eq:lucas-1}\\
\text{for integer } k > 0: && \binom{p^k}{i} \equiv 0 \mod{p} &\quad \text{for all }\ 0 < i < p^k\label{eq:lucas-2}
\end{align}

\end{fact}

\begin{figure}[t]
\centering
\begin{tikzpicture}
\tikzset{
	vert/.style = {circle, fill=black, minimum size=5pt, inner sep=0pt, outer sep=0pt},
	hedge/.style = {rectangle, rounded corners=4pt, above right, draw, minimum height=1.8cm, minimum width=1.8cm},
	medge/.style = {rectangle, rounded corners=4pt, draw, minimum height=0.4cm, minimum width=1.45cm, thick, right}
}
\node[vert] (A0) at (0,0) {};
\node[vert] (A1) at ([shift={(1,0)}]A0) {};
\node[vert] (A2) at ([shift={(1,1)}]A0) {};
\node[vert] (A3) at ([shift={(0,1)}]A0) {};
\node[hedge] at ([shift={(-0.4,-0.4)}]A0) {};
\node[medge, minimum width=1.45cm, draw=Gred] (A01) at ([shift={(-0.25,0)}]A0) {};
\node[medge, minimum width=1.45cm, draw=Gred] (A23)at ([shift={(-0.25,0)}]A3) {};
\draw[ultra thick, Gred] (A01) -- (A23);

\node[vert] (B0) at ([shift={(3.5,0)}]A0) {};
\node[vert] (B1) at ([shift={(1,0)}]B0) {};
\node[vert] (B2) at ([shift={(1,1)}]B0) {};
\node[vert] (B3) at ([shift={(0,1)}]B0) {};
\node[hedge] at ([shift={(-0.4,-0.4)}]B0) {};
\node[medge, minimum width=2.95cm, draw=Gblue] (B1C0) at ([shift={(-0.25,0)}]B1) {};
\node[medge, minimum width=2.95cm, draw=Gblue] (B2C3)at ([shift={(-0.25,0)}]B2) {};
\draw[ultra thick, Gblue] (B1C0) -- (B2C3);

\node[vert] (C0) at ([shift={(3.5,0)}]B0) {};
\node[vert] (C1) at ([shift={(1,0)}]C0) {};
\node[vert] (C2) at ([shift={(1,1)}]C0) {};
\node[vert] (C3) at ([shift={(0,1)}]C0) {};
\node[hedge] at ([shift={(-0.4,-0.4)}]C0) {};
\node[medge, minimum width=2.45cm, draw=Ggreen] (C1D0)at ([shift={(-0.25,0)}]C1) {};
\node[medge, minimum width=2.65cm, draw=Ggreen, rotate=-27] (C2D0)at ([shift={(-0.2,0.1)}]C2) {};
\draw[ultra thick, Ggreen, out=130, in=-100] ([shift={(0,0.12)}]C1.north) -- ([shift={(0.05,-0.15)}]C2.south);

\node[vert] (D0) at ([shift={(2,0)}]C1) {};
\node[vert] (D1) at ([shift={(2,0)}]D0) {};
\node[vert] (D2) at ([shift={(2,1)}]D0) {};
\node[vert] (D3) at ([shift={(0,1)}]D0) {};
\node[medge, minimum width=2.45cm, draw=Gyellow] (D1D2)at ([shift={(-0.25,0)}]D3) {};

\node[rectangle, rounded corners=4pt, above, draw, minimum height=1.8cm, minimum width=0.8cm, dashed] at ([shift={(0,-0.4)}]D1) {};
\node at ([shift={(0.9,0.5)}]D1) {$V^*$};

\end{tikzpicture}
\caption{Illustration of the proof of $\PPA_{p^k} \subseteq \PPA_{p}$ for $p=2$, $k=2$, $n=2$, $t=1$. In black, we indicate the $4$-dimensional matching $G$. In color, we highlight some of the vertices of $\overline{G}$ and the edges between them. The vertices of $\overline{G}$ in \textcolor{Gred}{red}, \textcolor{Gblue}{blue} and \textcolor{Ggreen}{green} are paired up and hence are non-solutions; whereas the vertex in \textcolor{Gyellow}{yellow} is isolated and not in $\overline{V}^*$ and hence a solution.}
\label{fig:ppa-prime-pow}
\end{figure}
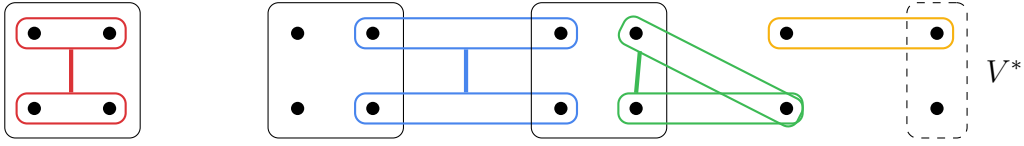

\noindent We reduce $\lonely_{p^k}$ to $\lonely_p$. Our instance of $\lonely_{p^k}$ is $(C,V^*)$ where $C$ implicitly encodes a ${p^k}$-dimensional matching $G = (V = [p^k]^n, E)$ and a designated vertex set $V^* \subseteq V$ such that $|V^*| \not\equiv 0 \mod{p^k}$.

Let $p^t$, $0 \le t < k$, be the largest power of $p$ that divides $|V^*|$. Through local operations we construct a $p$-dimensional matching hypergraph $\overline{G} = (\overline{V}, \overline{E})$ over vertices $\overline{V} \coloneqq \binom{V}{p^t}$ (set of all size-$p^t$ subsets of $V$) with designated vertices $\overline{V}^*\coloneqq\binom{V^*}{p^t}$. From \eqnref{eq:lucas-1}, we get that $|\overline{V}| \equiv 0 \mod{p}$ and $|\overline{V}^*| \not\equiv 0 \mod{p}$.

We will describe an algorithm that on vertex $\overline{v} \in \overline{V}$ outputs a hyperedge of $p$ vertices that contains $\overline{v}$ (if any). To this end, first fix an algorithm that for any set $e \coloneqq \set{u_1, \ldots, u_{p^k}} \subseteq V$ and for any $1 \le i \le p^t$, computes some ``canonical'' partition of the set $\binom{e}{i}$ into subsets of size $p$, and moreover assigns a canonical cyclic order within each such subset. This is indeed possible because of \eqnref{eq:lucas-2}, since $t < k$.\\[-0.2cm]

\noindent Given a vertex $\overline{v} \coloneqq \set{v_1, \ldots, v_{p^t}} \in \overline{V}$,\vspace{-0.2cm}
\begin{itemize}[label=$\triangleright$,leftmargin=0.6cm]
\item Compute all edges $e_1, \ldots, e_\ell \in E$ that include some $v \in \overline{v}$.
\item For edge $e_j$, define $S_j \coloneqq e_j \cap \overline{v}$ and let $S_j^1, \ldots, S_j^{p-1}$ be the remaining subsets in the same partition as $S_j$ in the canonical partition of $\binom{e_j}{|S_j|}$, listed in the canonical cyclic order starting at $S_j$. Also, let $S_0$ be the set of untouched vertices in $\overline{v}$. Observe that $\overline{v} = S_0 \cup S_1 \cup \ldots \cup S_\ell$.
\item Output neighbors of $\overline{v}$ as the vertices $\overline{v}_1, \ldots, \overline{v}_{p-1}$ where $\overline{v}_i \coloneqq S_0 \cup S_1^i \cup \ldots \cup S_{\ell}^i$.
\end{itemize}

\noindent It is easy to see that $\overline{v}$ is isolated in $\overline{G}$ iff all $v \in \overline{v}$ are isolated in $G$. Moreover, any isolated vertex in $\overline{V} \smallsetminus \overline{V}^*$ contains at least one isolated vertex in $V \smallsetminus V^*$; and a non-isolated vertex in $\overline{V}^*$ contains at least one non-isolated vertex in $V^*$ (in fact $p^t$ many).

The edges of $\overline{G}$ can indeed be computed efficiently with just black-box access to $C$. In order to complete the reduction, we only need that $\overline{V}$ is efficiently indexable. This is indeed standard; see \cite[\S2.3]{kreher98combinatorial} for a reference. See \autoref{fig:ppa-prime-pow} for an illustration of the proof.

\begin{remark}
Note that the size of the underlying graph blows up polynomially in our reduction. We do not know whether a reduction exists that avoids such a blow-up, although we suspect that the techniques of \cite{beame98more} can be used to show that some blow-up is necessary for black-box reductions.
\end{remark}

\section{A Natural Complete Problem}
\label{sec:chevalley}

\newcommand{\myshortdef}[6]{
	\noindent
	\begin{minipage}{\linewidth}
		\begin{definition}
			\label{def:#1} \contour{black}{(#2)}\vspace{-5pt}
			\begin{itemize}[leftmargin=2.15cm,itemsep=4pt]
				\item[\slshape Principle:] #3
				\item[\slshape Input:] #4
				\item[\slshape Condition:] #5
				\item[\slshape Output:] #6
			\end{itemize}
		\end{definition}
	\end{minipage}\vspace{3mm}
}

\newcommand{\myshortdefpr}[5]{
	\noindent
	\begin{minipage}{\linewidth}
		\begin{definition}
			\label{def:#1} \contour{black}{(#2)}\vspace{-5pt}
			\begin{itemize}[leftmargin=2.15cm,itemsep=4pt]
				\item[\slshape Principle:] #3
				\item[\slshape Input:] #4
				\item[\slshape Output:] #5
			\end{itemize}
		\end{definition}
	\end{minipage}\vspace{3mm}
}

We start with some notation that will be useful for the
presentation of our results.
\smallskip

\paragr{Notations.} For any polynomial $g \in \bbF_p[\vecx]$, we
define $\deg(g)$ to be the degree of $g$. We define
\textit{the expansion to monic monomials} of $g$ as $\sum_{\ell = 1}^L t_{\ell}(\vecx)$, 
where $t_{\ell}(\vecx)$ is a monic monomial in
$\bbF_p[\vecx]$, i.e. a monomial with coefficient $1$. For example,
the expansion of the polynomial
$g(x_1, x_2) = x_1 \cdot (2 x_1 + 3 x_2)$ is
given by $x_1^2 + x_1^2 + x_1 x_2 + x_1 x_2 + x_1 x_2$.

For a polynomial system $\vec{f} := (f_1, \ldots, f_m) \in \bbF_p[\vecx]^m$, its affine variety $\calV_{\vecf} \subseteq \bbF_p^n$ is defined as
$\calV_{\vecf} \coloneqq \set{\vecx \in \bbF_p^n \mid \vecf(\vecx) = \vec{0}}$.
Let $\overline{\calV_{\vecf}} \coloneqq \bbF_p^n \setminus \calV_{\vecf}$. If the constant term of each $f_i$ is $0$, we say that $\vecf$ is {\em zecote}, standing for ``Zero Constant Term'' (owing to lack of known terminology and creativity on our part).

\subsection{The Chevalley-Warning Theorem} \label{sec:ChevalleyWarningProof}

  We repeat the formal statement of Chevalley-Warning Theorem together with its
proof.

\begin{cwthm}[\cite{chevalley35demonstration,warning36bemerkung}]
    For any prime $p$ and a polynomial system
  $\vecf \in \bbF_p[\vecx]^m$ satisfying $\sum_{i = 1}^m \deg(f_i) < n$
  \emph{(CW Condition)}, $|\calV_{\vecf}| \equiv 0 \mod{p}$.
\end{cwthm}

  We describe the proof of \CWT through \autoref{lem:mainLemmaInProofOfCWT}.
Even though there are direct proofs, the following presentation helps motivate the generalizations we study in future sections. Given a polynomial system $\vecf \in
\bbF_p[\vecx]^m$, a key idea in the proof is the polynomial $\CW_{\vecf}(\vecx)
\coloneqq \prod_{i = 1}^m \CW_{f_i}(\vecx)$ where each $\CW_{f_i}(\vecx) := (1 - f_i(\vecx)^{p - 1})$. Observe that $\CW_{\vecf}(\vecx) = 1$ if
$\vecx \in \calV_{\vecf}$ and is $0$ otherwise. The following definition
describes the notion of a max-degree monomial of $\CW_{\vecf}$ that plays an important role in the proof.

\begin{definition}[\textsc{Max-Degree Monic Monomials}]
\label{def:max-degreeMonomial}
For any prime $p$, let $\vecf \in \bbF_p[\vecx]^m$ and let the expansion into monic monomials of $\CW_{f_i}(\vecx)$ be $\sum_{\ell = 1}^{r_i} t_{i, \ell}(\vecx)$. Let also $U_i = \{(i, \ell) \mid \ell \in [r_i]\}$ and $U = \bigtimes_{i = 1}^m U_i$, we define the following quantities.
  \begin{Enumerate}
    \item A \textit{monic monomial} of $\CW_{\vecf}$ is a product $t_S(\vecx) = \prod_{i = 1}^m t_{s_i}(\vecx)$ for $S = (s_1, \ldots, s_m) \in U$.
	    \item A \textit{max-degree monic monomial} of $\CW_{\vecf}$ is any monic monomial $t_S(\vecx)$, such that\\
    $t_S(\vecx) \equiv \prod_{j = 1}^n x_j^{p - 1} \mod{\set{x_i^p - x_i}_{i \in [n]}}$.
    \item We define $\calM_{\vecf}$ to be the set of max-degree monic monomials of $\CW_{\vecf}$, i.e. \\
	$\calM_{\vecf} \coloneqq \set{S \in U \mid t_S \text{ is a max-degree monic monomial of } \CW_{\vecf}}$.
  \end{Enumerate}
\end{definition}

\noindent In words, the monomials $t(S)$ are precisely the ones that arise when symbolically expanding $\CW_{\vecf}(\bx)$. We illustrate this with an example: Let $p = 3$ and $f_1(x_1, x_2) = x_1 + x_2$ and $f_2(x_1, x_2) = x_1^2$. Then modulo $\set{x_1^3 - x_1, x_2^3 - x_2}$, we have
\begin{align*}
\CW_{(f_1,f_2)}(x_1, x_2) &~=~ (1 - (x_1 + x_2)^2)(1 - (x_1^2)^2)\\
&~=~ (1 - x_1^2 - 2 x_1 x_2 - x_2^2) \cdot (1 - x_1^2)\\
&~=~ (1 + x_1^2 + x_1^2 + x_1x_2 +x_2^2 + x_2^2) \cdot (1 + x_1^2 + x_1^2)
\end{align*}

\noindent Thus there are $18$ $(= 6 \times 3)$ monic monomials in the system $(f_1, f_2)$. The monomial corresponding to $S = ((1,5),(2,2))$ is a maximal monomial since the $5$-th term in $\CW_{f_1}$ is $x_2^2$ and $2$-nd term in $\CW_{f_2}$ is $x_1^2$. Using the above definitions, we now state the main technical lemma of the proof of \CWT.

\begin{lemma}[Main Lemma in the proof of CWT] \label{lem:mainLemmaInProofOfCWT}
   For any prime $p$ and any system of polynomials
  $\vecf \in \bbF_p[\vecx]^m$, it holds that
  $\abs{\calV_{\vecf}} \equiv (-1)^n \abs{\calM_{\vecf}} \mod{p}$.
\end{lemma}

\begin{proof}
As noted earlier, $\CW_{\vecf}(\vecx) = 1$ if $\bx \in \calV_{\vecf}$ and is $0$ otherwise. Thus, it follows that $\abs{\calV_{\vecf}} \equiv \sum_{\bx \in \bbF_p^n} \CW_{\vecf}(\bx) \mod{p}$. For any monic monomial $m(\bx) = \prod_{j=1}^n x_j^{d_j}$, it holds that $\sum_{\bx \in \bbF_p^n} m(\bx) = 0$ if $d_j < p - 1$ for some $x_j$. On the other hand, for the monic max-degree monomial $m(\bx) = \prod_{j = 1}^n x_j^{p-1}$, it holds that $\sum_{\bx \in \bbF_p^n} m(\bx) = (p-1)^n$. Thus, we get that $|\calV_{\vecf}| \equiv \sum_{\bx \in \bbF_p^n} \CW_{\vecf}(\bx) \mod{p} \equiv \sum_{S\in U} \sum_{\bx \in \bbF_p^n} t_S(\vecx)  \mod{p} \equiv (-1)^n |\calM_{\vecf}| \mod{p}$.
\end{proof}

\noindent The proof of Chevalley-Warning Theorem follows easily from
\autoref{lem:mainLemmaInProofOfCWT}.

\begin{proof}[Proof of Chevalley-Warning Theorem.]
We have that $\deg(\CW_{\vecf}) \le (p - 1) \sum_{i = 1}^m \deg(f_i)$. Thus, if $\vecf$ satisfies \eqref{eq:CWcondition}, then $\deg(\CW_{\vecf}) < (p - 1) n$ and hence $\abs{\calM_{\vecf}} = 0$. \CWT now follows from \autoref{lem:mainLemmaInProofOfCWT}.
\end{proof}

\subsection{The Chevalley-Warning Theorem with Symmetry}
\label{sec:ChevalleyWarningWithSymmetry}

  In this section, we formalize the intuition that we built in
\hyperref[sec:proofsOfCancellation]{Sections~\ref{sec:proofsOfCancellation}} and \ref{sec:computationalChevalleyWarning} to prove the more general
statements to lead to the same conclusion as the Chevalley-Warning
Theorem.

  First, we prove a theorem that argues about the cardinality of $\calV_{\vecf}$
directly using some symmetry of the system of polynomials $\vecf$. Then,
combining this symmetry-based argument with the \eqref{eq:generalCWCondition} we get
the generalization of the Chevalley-Warning Theorem. Our natural
$\PPA_p$-complete problem is based on this generalization.

The theorem statements are simplified using the definition of
\textit{free action} of a group. For a permutation over $n$ elements $\sigma \in S_n$, we define
$\langle \sigma \rangle$ to be the sub-group generated by $\sigma$ and
$|\sigma|$ to be the order of $\langle\sigma \rangle$. For $\vecx \in \bbF_p^n$, $\sigma(\vecx)$ denotes the assignment obtained by permutating the variables of the assignment $\vecx$ according to $\sigma$.

\begin{definition}[\textsc{Free Group Action}]
    Let $\sigma \in S_n$ and $\calV \subseteq \bbF_p^n$, then we say that
  $\langle \sigma \rangle$ acts freely on $\calV$ if, for every
  $\vecx \in \calV$, it holds that $\sigma(\vecx) \in \calV$ and
  $\vecx \neq \sigma(\vecx)$.
\end{definition}

  Our first theorem highlights the
use of symmetry in arguing about the size of $\abs{\calV_{\vecf}}$.

\begin{theorem} \label{thm:symmetricChevalleyTheorem}
    Let $\vecf \in \F_p[\vecx]^m$ be a system of polynomials.
  If there exists a permutation $\sigma \in S_n$ with
  $\abs{\sigma} = p$ such that $\a{\sigma}$ acts freely on
  $\overline{\calV}_{\vecf}$, then $\abs{\calV_{\vecf}} \equiv 0 \mod{p}$.
\end{theorem}

\begin{proof}
    Since $\sigma$ acts freely on $\overline{\calV}_{\vecf}$, we can partition
  $\overline{\calV}_{\vecf}$ into orbits of any
  $\vecx \in \overline{\calV}_{\vecf}$ under actions of $\a{\sigma}$, namely
  sets of the type $\set{\sigma^i(\vecx)}_{i \in [p]}$ for
  $\vecx \in \overline{\calV}_{\vecf}$. Since $\a{\sigma}$ acts freely on
  $\overline{\calV}_{\vecf}$, each such orbit has size $p$. Thus, we can
  conclude that $\abs{\overline{\calV}_{\vecf}} \equiv 0 \pmod p$ from which
  the theorem follows.
\end{proof}

\begin{remark} \label{rem:syntacticVerificationOfSymmetry}
  For any polynomial system $\vecf$ and any
  permutation $\sigma$, we can check in linear time if $\abs{\sigma} = p$ and we
  can syntactically refute that $\langle \sigma \rangle$ acts freely on
  $\overline{\calV}_{\vecf}$ with an $\vecx \in \bbF_p^n \smallsetminus \set{\vec{0}}$ such that $\vecf(\sigma(\vecx)) = \vec{0}$ or $\sigma(\vecx) = \vecx$.
\end{remark}
\bigskip

We now state and prove an extension of \CWT that captures both
the argument from \autoref{lem:mainLemmaInProofOfCWT} and the symmetry
argument from \autoref{thm:symmetricChevalleyTheorem}.

\begin{theorem}[\textsc{Chevalley-Warning with Symmetry Theorem}]
\label{thm:combinedChevalleyTheorem}
	  Let $\vecg \in \F_p[\vecx]^{m_g}$ and
	$\vech \in \F_p[\vecx]^{m_h}$ be two systems of polynomials, and
	$\vecf \coloneqq (\vecg, \vech)$. If there exists a permutation $\sigma \in S_n$ with $\abs{\sigma} = p$ such that (1) $\calM_{\vecg} = \emptyset$ and (2) $\a{\sigma}$ acts freely on
	$\calV_{\vecg} \cap \overline{\calV_{\vech}}$, then $\abs{\calV_{\vecf}} = 0 \mod{p}$.
\end{theorem}

\begin{remark} \label{rem:conditionsOfChevalleyWithSymmetry}
    We point to the special form of Condition 2. By definition,
  $\calV_{\vecf} = \calV_{\vecg} \cap \calV_{\vech}$, hence if $\a{\sigma}$
  were to act freely on $\overline{\calV_{\vecg}} \cup \overline{\calV_{\vech}}$ (or even $\calV_{\vecg} \cap \calV_{\vech}$), then we could just
  use \autoref{thm:symmetricChevalleyTheorem} to get that
  $\abs{\calV_{\vecf}} \equiv 0 \mod{p}$. In the above theorem,
  we only require that $\a{\sigma}$ acts freely on
  $\calV_{\vecg} \cap \overline{\calV_{\vech}}$. Observe that
  \autoref{thm:symmetricChevalleyTheorem} follows as a special case of \CWT with
  Symmetry by setting $m_g = 0$. Additionally, by setting $m_h = 0$ we get the generalization of \CWT corresponding to the
  \eqref{eq:generalCWCondition} as presented in
  \autoref{sec:computationalChevalleyWarning}.
\end{remark}

\begin{proof}[Proof of \autoref{thm:combinedChevalleyTheorem}]
  	If $\CW_{\vecg}$ does not have any max-
	degree monic monomials, we have $\abs{\calV_{\vecg}} \equiv 0 \mod{p}$ (similar to proof of \CWT) and, since
	$\overline{\calV_{\vecg}} = \F_p^n \setminus \calV_{\vecg}$, we have $\abs{\overline{\calV_{\vecg}}} \equiv 0 \mod{p}$. Also,
	since $\a{\sigma}$ acts freely on $\calV_{\vecg} \cap
	\overline{\calV_{\vech}}$, we have
	$\abs{\calV_{\vecg} \cap \overline{\calV_{\vech}}} \equiv 0 \mod{p}$ (similar to the proof of
	\autoref{thm:symmetricChevalleyTheorem}). Hence,
	$\abs{\overline{\calV_{\vecf}}} = \abs{\overline{\calV_{\vecg} \cap
	\calV_{\vech}}}  = \abs{\overline{\calV_{\vecg}} \cup
	\overline{\calV_{\vech}}} = \abs{\overline{\calV_{\vecg}}} +
	\abs{\calV_{\vecg} \cap \overline{\calV_{\vech}}} \equiv  0 \mod{p}$. Thus,
	$\abs{\calV_{\vecf}} \equiv 0 \mod{p}$.
\end{proof}

\subsection{Computational Problems Related to Chevalley-Warning Theorem}
\label{sec:computationalChevalleyWarningFormal}

  We now follow the intuition developed in the previous section and in
\autoref{sec:intro-complete} to formally define the computational problems
$\eChevalleyp$, $\lChevalleyp$, and $\cChevalleyp$.
\medskip

\myshortdef{explicit-chevalley}{$\eChevalleyp$}
{Chevalley-Warning Theorem.}
{$\vecf \in \bbF_p[\vecx]^m$ : an explicit zecote polynomial system.}
{$\sum_{i = 1}^m \deg(f_i) < n$.}
{$\vecx \in \bbF_p^n$ such that $\vecx \neq \vec{0}$ and $\vecf(\vecx) = \vec{0}$.}
\bigskip

\myshortdefpr{general-chevalley}{$\lChevalleyp$}
{General Chevalley-Warning Theorem via \eqref{eq:generalCWCondition}.}
{$\vecf \in \bbF_p[\vecx]^m$ : an explicit zecote polynomial system.}
{\begin{Enumerate}[leftmargin=*]
  \item[0.] A max-degree monic monomial $t_S(\vecx)$ of $\CW_{\vecf}$, or
  \item $\vecx \in \bbF_p^n$ such that $\vecx \neq \vec{0}$ and $\vecf(\vecx) = \vec{0}$.
\end{Enumerate}}

\bigskip

\myshortdef{combined-chevalley}{$\cChevalleyp$}
{Chevalley-Warning Theorem with Symmetry (\autoref{thm:combinedChevalleyTheorem}).}
{$\triangleright$ $\vecg \in \bbF_p[\vecx]^{m_g}$ and $\vech \in \bbF_p[\vecx]^{m_h}$ : explicit zecote polynomial systems\\
 $\triangleright$ $\sigma \in S_n$ : a permutation over $[n]$.}
{$|\sigma| = p$.}
{\begin{Enumerate}[leftmargin=*]
  \item[0.] (a) A max-degree monic monomial $t_S(\vecx)$ of $\CW_{\vecg}$, or
  \item[] (b) $\vecx \in \calV_{\vecg} \cap \overline{\calV_{\vech}}$ such that
        $\sigma(\vecx) \not\in (\calV_{\vecg} \cap \overline{\calV_{\vech}}) \smallsetminus \set{\vecx}$, or
  \item $\vecx \in \bbF_p^n$ such that $\vecx \neq \vec{0}$ and $\vecf(\vecx) = \vec{0}$.
\end{Enumerate}}

\begin{remark} \label{rem:computationalDefinitions}
    Some observations about the above computational problems follow:
  \begin{Enumerate}
    \item In the problems $\lChevalleyp$ and $\cChevalleyp$,  we assume that, if the output is a max-degree monic monomial, this is given via the multiset of indices $S$ that describes the monomial as formalized in \autoref{def:max-degreeMonomial}.
    \item We have $\eChevalleyp \preceq \lChevalleyp \preceq \cChevalleyp$. Thus, inclusion of $\cChevalleyp$ in $\PPA_p$ implies that the problems $\eChevalleyp$ and $\lChevalleyp$ are in $\PPA_p$. Also, in \autoref{sec:applications} we prove that $\sis_p \preceq \eChevalleyp$, where $\sis_p$ is a cryptographically relevant problem. This shows that the $\lChevalleyp$ and the $\cChevalleyp$ problems are at least as hard as the $\sis_p$ problem.
  \end{Enumerate}
\end{remark}

\noindent We restate our main result.

\completeness*

\subsection{\boldmath $\cChevalleyp$ is $\PPAp$--complete}
\label{sec:proofOfCombinedChevalley}

  We fist prove that $\cChevalley$ is in $\PPAp$ and then
prove its $\PPAp$-hardness.

\subsubsection{\boldmath $\cChevalleyp$ is in $\PPA_p$}

Even though Papadimitriou~\cite{papadimitriou94parity} provided a rough proof sketch of $\Chevalley_p \in \PPAp$, a formal proof was not given. We show that $\cChevalleyp$ is in $\PPAp$ (and so are $\lChevalleyp$ and $\eChevalleyp$). In order to do so we extend the definition of $\bipartiteq$ to instances where the vertices might have exponential degree and edges appear with multiplicity. The key here is to define a $\bipartite_q$ instance with unbounded (even exponential) degree, but with additional information that allows us to verify solutions efficiently.\\[-1mm]

\mydef{succ-bipartite}{$\sucBipartite_q$}
{Similar to $\bipartite_q$, but degrees are allowed to be exponentially large, edges are allowed with multiplicities at most $q-1$.}
{Bipartite graph $G = (V \cup U, E)$ s.t. $E \subseteq V \times U \times \bbZ_q$. Designated edge $e^* \in E$.}
{Let $V \coloneqq \set{0} \times \bit^{n-1}$ and $U \coloneqq \set{1} \times \bit^{n-1}$:\\
	$\triangleright$ $\calC:V \times U \to [q]$, edge counting circuit\\
	$\triangleright$ $\phi_V: V \times U \times [q] \to (U \times [q])^{q}$, grouping pivoted at $V$\\
	$\triangleright$ $\phi_U: V \times U \times [q] \to (V \times [q])^{q}$, grouping pivoted at $U$\\
	$\triangleright$ $e^* = (v^*, u^*, k^*)$, designated edge}
{$V \coloneqq \set{0} \times \bit^{n-1}$, $U \coloneqq \set{1} \times \bit^{n-1}$,\\
	$E \coloneqq \set{(v,u,k) : 1 \le k \le C(v,u),\ (v,u) \in V \times U}$ (here $k$ distinguishes multiplicities)\\
	Edge $(v,u,k)$ is grouped with $\set{(v, u', k') : (u',k') \in \phi_V(v,u,k)}$ (pivoting at $v$),\\
	\hspace*{1cm} provided $|\phi_V(v,u,k)| = q$, all $(v,u',k') \in E$ and $\phi_V(v,u',k') = \phi_V(v,u,k)$.\\
	Edge $(v,u,k)$ is grouped with $\set{(v', u, k') : (v',k') \in \phi_U(v,u,k)}$ (pivoting at $u$),\\
	\hspace*{1cm} provided $|\phi_U(v,u,k)| = q$, all $(v,u',k') \in E$ and $\phi_U(v',u,k') = \phi_V(v,u,k)$.}
{$e^*$ if $e^*$ is grouped, pivoting at $v^*$, or if $e^*$ is not grouped pivoting at $u^*$, OR\\
	$e \ne e^*$ if $e$ is not grouped pivoting at one of its ends.}
In words, $\sucBipartite_p$ encodes a bipartite graph with arbitrary degree. Instead of listing the neighbors of a vertex using a circuit, we have a circuit that outputs the multiplicity of edges between any two given vertices. We are therefore unable to efficiently count the number of edges incident on any vertex. The grouping function $\phi_V$ aims to group edges incident on any vertex $v \in V$ into groups of size $q$. Similarly, $\phi_U$ aims to group edges incident on any vertex $u \in U$. The underlying principle is that if we have an edge $e^*$ that is not grouped pivoting at $v^*$ (one of its endpoints), then either $e^*$ is not pivoted at $u^*$ (its other endpoint) or there exists another edge that is also not grouped pivoting at one of its ends. Note that in contrast to the problems previously defined, $v^*$ might still be an endpoint of a valid solution.

\begin{lemma} \label{lem:combinedChevalleyPPAp}
 For all primes $p$, $\cChevalleyp \in \PPAp$.
\end{lemma}
\begin{proof}We reduce $\cChevalleyp$ to $\sucBipartite_p$, which we show to be $\PPA_p$--complete in \autoref{sec:sucBipartite}. Given an instance of $\cChevalleyp$, namely a zecote polynomial system $\vecf = (\vecg, \vech)$ and a permutation $\sigma$, we construct a bipartite graph $G = (U \cup V, E)$ encoded as an instance of $\sucBipartite_p$ as follows. 

 \paragraph{Description of vertices.} $U = \bbF_p^n$, namely all possible assignments
 of $\vecx$. The vertices of $V$ are divided into two parts $V_1 \cup V_2$.
 The part $V_1$ contains one vertex for each monomial in the expansion of $\CW_{\vecg} = \prod_{i = 1}^{m_g} (1 - g_i^{p-1})$. Since $p$ is constant, we can
 efficiently list out the monomials of $1 - g_{i}^{p-1}$. For a fixed lexicographic ordering of the
 monomials of each $\CW_{g_i} \coloneqq 1 - g_{i}^{p-1}$, a
 monomial of $\CW_{\vecg}$ is represented by a tuple
 $(a_1, a_{2}, \dots, a_{m_g})$ with $0 \leq a_{i} < L_{i}$, where
 $a_{i}$ represents the index of a monomial of $\CW_{g_i}$ and  $L_{i}$
 is the number of monomials of $\CW_{g_i}$, where $a_i = 0$ corresponds to the constant term $1$. The part $V_2 \coloneqq \binom{\bbF_p^n}{p}$, i.e. it contains a vertex for each subset of $p$ distinct elements in $\bbF_p^n$.

 \paragraph{Description of edges.} We first describe the edges between $U$ and $V_1$, namely include an
 edge between an assignment $\vecx$ and a monomial $t$ with
 multiplicity $t(\vecx)$. With these edges in place, the degree of vertices are as follows:\vspace{-2mm}
 \begin{itemize}[leftmargin=0.5cm]
	\item $\vecx = 0^n$ has a single edge corresponding to the constant monomial $1$, since $\vecf$ is zecote. We let this be the designated edge $e^*$ in the final $\sucBipartite_p$ instance.
	\item $\vecx \notin \calV_{\vecg}$ has $0 \pmod p$ edges
	(counting multiplicities). Since $\CW_{\vecg}(\vecx) = 0$, the sum over all monomials of $t(\vecx)$ must be $0 \mod{p}$.
	\item $\vecx \in \calV_{\vecg}$ has $1 \pmod p$ edges (counting multiplicities), since the sum over all $t(\vecx)$ monomials gives $\CW_{\vecg}(\bx) \equiv 1 \mod{p}$.
 \end{itemize}
 Thus with the edges so far, the vertices (excluding $0^n$), with degree $\not\equiv 0 \mod{p}$ are precisely vertices  $t \in V_1$ such that $\sum_{\vecx} t(\vecx) \not\equiv 0 \pmod p$ or $\vecx \in \calV_{\vecg} \smallsetminus \set{0^n}$. For the former case, if $t$ contained a variable with degree less than $p-1$, then $\sum_{\vecx} t(\vecx) \equiv 0 \pmod p$. Hence, it must be that $t = \prod\limits_{i = 1}^{ n} x_i^{p-1}$. In the later case, the degree of $\vecx $ is $1 \mod{p}$ and hence $\vecx \in \calV_{\vecg}$. 
 However, there is no guarantee that a vertex $\vecx$ with degree $1 \mod{p}$ is in $\calV_{\vech}$ as well. To argue
 about $\vech$, we add edges between $U$ and $V_2$ that exclude solutions $\vecx \in \calV_{\vecg} \cap \overline{\calV_{\vech}}$, on which $\sigma$ acts freely (that is, $\sigma(\vecx) = \vecx$). More specifically, for
 $\vecx \in \calV_{\vecg} \cap \overline{\calV_{\vech}}$, if $\sigma(\vecx) \ne \vecx$,
  we add an edge with multiplicity $p-1$ between $\vecx$ and $\Sigma_{\vecx} \in V_2$ where $\Sigma_{\vecx} \coloneqq \{\sigma^i(\vecx)\}_{i \in \bbZ_p}$ (note that, in this case $\abs{\Sigma_{\vecx}} = p$ since $\sigma(\vecx) \ne \vecx$ and $|\sigma| = p$ is prime).
 Observe that, if a vertex in $V_2$ corresponds
 to a $\Sigma_{\vecx}$, it has $p$ edges each with multiplicity $p-1$, one for each
 $\vecx' \in \Sigma_{\vecx}$ only if $\Sigma_{\vecx} \subseteq \calV_{\vecg} \cap \overline{\calV_{\vech}}$. If a vertex in $V_2$ does not correspond to a
 $\Sigma_{\vecx}$, then it has no edges. Thus, a vertex in $V_2$ has degree $\not\equiv 0 \mod{p}$ iff it contains an $\vecx \in \calV_{\vecg} \cap \overline{\calV_{\vech}}$ such that $\sigma(\vecx) \notin \calV_{\vecg} \cap \overline{\calV_{\vech}}$.

 Thus, with all the edges added, vertices with degree $\not\equiv 0 \mod{p}$ correspond to one of\vspace{-2mm}
 \begin{itemize}
 \item $\vecx \in \calV_{\vecg} \cap \calV_{\vech}$ such that $\vecx \ne \vec{0}$, or
 \item $t \in V_1$ such that $t(\vecx)$ is a max-degree monomial or
 \item $\vecx \in \calV_{\vecg} \cap \overline{\calV_{\vech}}$ such that $\sigma(\vecx) = \vecx$ or
 \item $v \in V_2$ such that $\exists\, \vecx \in v$ satisfying $\vecx \in \calV_{\vecg} \cap \overline{\calV_{\vech}}$ and $\sigma(\vecx) \notin \calV_{\vecg} \cap \overline{\calV_{\vech}}$.
 \end{itemize}
 These correspond precisely to the solutions of $\cChevalleyp$. To summarize, the edge
 counting circuit $C$ on input $(\vecx, t) \in U \times V_1$ outputs
 $t(\vecx)$ and on input $(\vecx, v) \in U \times V_2$ outputs $p-1$ if
 $\vecx \in \calV_{\vecg} \cap \overline{\calV_{\vech}}$, $\sigma(\vecx) \ne \vecx$ and $v = \Sigma_{\vecx}$ and 0
 otherwise.

 \paragraph{Grouping Functions.}
      The grouping functions $\phi_U$ and $\phi_V$ are defined as follows (analogous to the so-called ``chessplayer algorithm'' in \cite{papadimitriou94parity}):
 \begin{itemize}[label=$\triangleright$,leftmargin=0.5cm]
	\item Grouping $\phi_U$ (corresponding to endpoint in $U$):
\begin{itemize}[leftmargin=0.25cm]
		\item For $\vecx \in \overline{\calV_{\vecg}}$: there exists some $i$ such
		that $\CW_{g_i}(\vecx) = 0$. Consider an edge
		$(\vecx,(a_1,a_2, \dots, a_{m_g}),k)$. We can explicitly list out the multiset containing the monomials
		$t_j = (a_1,a_2, \dots, a_i \gets j, \dots, a_{m_g})$ with multiplicity
	        $t_j(\vecx)$, for each $1 \le j \le L_i$. Since $\CW_{g_i}(\vecx) = 0$, this multiset has size
		multiple of $p$. Hence, we can canonically divide its elements into groups of
		size $p$, counting multiplicities and $\phi_U$ returns the subset containing $(t,k)$.
		\item For $\vecx \in \calV_{\vecg} \cap \overline{\calV_{\vech}}$ such that $\sigma(\vecx) \ne \vecx$: Note that $g_{i}^{p-1}(\vecx) = 0$ for all $i \in \bracks{m_g}$.
		Let $v_1 \in V_1$ be the vertex corresponding to the
		constant monomial $1$. $\phi_U$ groups the edge $(\vecx, v_1, 1)$ (of multiplicity $1$) with the $p-1$ edges $(\vecx, \Sigma_{\vecx}, k)$ for $k \in [p-1]$.
		For any other $t \in  V_1 \setminus \set{v_1}$ and an edge $(\vecx, t,k)$, we have that
		$t = (a_1,\dots, a_{m_g})$ has $a_i \ne 0$ for some $i$. We define
		the multiset containing $t_j = (a_1, \dots, a_{i} \gets j, \dots a_{m_g})$
		with multiplicity $t_j(\vecx)$ for each $1 \le j < L_i$. Since
		$g_{i}^{p-1}(\vecx) = 0$, this multiset has size which is a multiple of $p$, which we can canonically partition into groups of size $p$. Thus, $\phi_U$ on input
		$(\vecx, t, k)$ returns the group containing $(t,k)$.
 	\end{itemize}
	\item Grouping $\phi_V$ (corresponding to endpoint in $V$):
		\begin{itemize}[leftmargin=0.25cm]
		\item For $t \in V_1$ such that $t \neq \prod\limits_{i = 1}^n x_i^{p-1}$:
		there exists a variable $x_i$ with degree less than $p-1$. For
		$\vecx_j = (x_1, \dots, x_{i-1}, x_i \gets j, \dots, x_n)$ with $j \in \bbF_p$
		we define the multiset $\set{(\vecx_j, t(\vecx_j))}_{j \in \bbF_p}$.
		Since $\sum\limits_{j = 0}^{p-1} t(\vecx_j) = 0$, this multiset has
		size multiple of $p$, so we can canonically partition it into groups of size $p$. Then,
		$\phi_V(\vecx, t, k)$ returns the group containing $(\vecx,k)$,
		\item For $v \in V_2$: if $\deg(v) = 0$, then there is no grouping to be done. Else if $\deg(v) \equiv 0 \mod{p}$ then $\phi_V(\vecx,t,k)$ returns
          		$\set{(\vecx,k)}_{\vecx \in v}$.
	  \end{itemize}
\end{itemize}
	Thus, for any vertex with degree $\equiv 0 \mod{p}$, we have provided a grouping function for all its edges. So, for any edge that is not grouped by grouping function at any of its endpoints, then such an endpoint must have degree $\not\equiv 0 \mod{p}$ and hence point to a valid solution of the $\cChevalleyp$ instance.
\end{proof}

\subsubsection{$\cChevalleyp$ is $\PPA_p$--hard}

We show that $\lonely_p \preceq \cChevalleyp$. In the $\cChevalleyp$ instance that we create, we will ensure that there are no solutions of type 0 (as in \autoref{def:combined-chevalley}) and thus, the only valid solutions will be of type 1. In order to do so, we introduce the notions of labeling and proper labeling and prove a generalization of \CWT that we call Labeled \CWT (\autoref{thm:labeledChevalleyWarningTheorem}).

As we will see, the Labeled \CWT, is just a
re-formulation of the original \CWT rather than a generalization. To
understand the Labeled \CWT we start with some
examples that do not seem to satisfy the Chevalley-Warning
condition, but where a solution exists.

\bigskip
\paragr{Example 1.} Consider the case where $p = 3$ and
$f(x_1, x_2) = x_2 - x_1^2$. In this case the Chevalley-Warning condition is not met, since we have $2$ variables and the total degree is also
$2$. But, let us consider a slightly different polynomial where we
replace the variable $x_2$ with the product of two variables
$x_{21}, x_{22}$ then we get the polynomial
$g(x_1, x_{21}, x_{22}) = x_{21} \cdot x_{22} - x_1^2$. Now,
$g$ satisfies \eqref{eq:CWcondition} and hence, we
conclude that the number of roots of $g$ is a multiple of $3$.
Interestingly, from this fact we can argue that there exists a non-trivial solution for $f(\vecx) = 0$.
In particular, the assignment $x_1 = 0$, $x_2 = 0$ corresponds to five
assignments of the variables $x_1$, $x_{21}$, $x_{22}$. Hence, since
$\abs{\calV_{g}} = 0 \mod{3}$, $g$ has another root, which corresponds to a
non-trivial root of $f$. In this example, we applied the
\CWT on a slightly different polynomial than $f$
to argue about the existence of non-trivial solutions of $f$, even though $f$ did not satisfy \eqref{eq:CWcondition} itself.
\bigskip

\paragr{Ignore Some Terms.} The Labeled \CWT formalizes the phenomenon
observed in Example 1 and shows that under certain conditions we can
\textit{ignore some terms} when defining the degree of each polynomial.
For instance, in Example 1, we can ignore the term
$x_1^2$ when computing the degree of $f$ and treat $f$ as a degree $1$
polynomial of $2$ variables, in which case the condition of \CWT is
satisfied.

We describe which terms can be ignored by defining a \textit{labeling}
of the terms of each polynomial in the system. The labels take values in
$\{-1, 0, +1\}$ and our final goal is to ignore the terms with label
$+1$. Of course, it should not be possible to define any labeling that
we want; for example we cannot ignore all the terms of a polynomial.
Next, we describe the rules of a \textit{proper labeling} that will
allow us to prove the Labeled \CWT. We start with
a definition of a labeling.

\begin{definition}[\textsc{Monomial Labeling}]
  \label{def:monomiallabeling}
     Let $\vecf \in \bbF[\vecx]^m$ and let $t_{i j}$ be the $j$-th monomial of the polynomial $f_i \in \bbF[\vecx]$ (written in some canonical sorted order). Let $\calT$ be the
  set of all pairs $(i, j)$ such that $t_{i j}$ is a monomial in
  $\vecf$. A \textit{labeling} of $\vecf$ is a function
  $\lambda : \calT \to \{-1, 0, +1\}$ and we say that $\lambda(i, j)$
  is the \textit{label} of $t_{i j}$ according to $\lambda$.
\end{definition}

\begin{definition}[\textsc{Labeled Degree}] \label{def:labeledDegree}
    For $\vecf \in \bbF[\vecx]^m$ with a labeling $\lambda$,     we define the \textit{labeled degree} of $f_i$ as,
      $ \deg^{\lambda}(f_i)\coloneqq \max_{j\,:\,\lambda(i, j) \ne +1} \deg(t_{i j})$, in words, maximum degree among monomials of $f_i$ labeled either $0$ or $-1$.
\end{definition}

\paragr{Example 1 (continued).} According to the lexicographic ordering,
$f(x_1, x_2) = - x_1^2 + x_2$ and we have the monomials
$t_{1 1} = - x_1^2$ and $t_{1 2} = x_2$. Hence, one possible labeling,
which as we will see later corresponds to the vanilla Chevalley-Warning
Theorem, is $\lambda(1, 1) = \lambda(1, 2) = 0$. According to this labeling,
$\deg^{\lambda}(f) = 2$. Another possible
labeling, that, as we will see, allows us to apply the Labeled
\CWT\!, is $\lambda(1, 1) = +1$ and $\lambda(1, 2) = -1$.
In this case, the labeled degree is $\deg^{\lambda}(f) = 1$.
\bigskip

  As we highlighted before, our goal is to prove the Chevalley-Warning
Theorem, but with the weaker condition that
$\sum_{i = 1}^m \deg^{\lambda}(f_{i}) < n$ instead of
$\sum_{i = 1}^m \deg(f_i) < n$. Of course, we first have to restrict the
space of all possible labelings by defining {\em proper labelings}.
In order to make the condition of proper labelings easier to interpret
we start by defining the notion of a labeling graph.

\begin{definition}[\textsc{Labeling Graph}] \label{def:labelingGraph}
  For $\vecf \in \bbF[\vecx]^m$ with a labeling $\lambda$, we define the
  \textit{labeling graph} $G_{\lambda} = \p{U \cup V, E}$ as a
  directed bipartite graph on vertices
  $U = \set{x_1, \dots, x_n}$ and $V=\set{f_1, \dots, f_m}$. The edge
  $(x_j \to f_i)$ belongs to $E$ if $x_j$ appears in a monomial
  $t_{i r}$ in $f_i$ with label $+1$, i.e. $\lambda(i, r) = +1$. Symmetrically, the edge
  $(f_i \to x_j)$ belongs to $E$ if the $x_j$ appears in a monomial
  $t_{i r}$ in $f_i$ with label $-1$, i.e. $\lambda(i, r) = -1$.
\end{definition}

\paragr{Example 2.} Let $p = 2$ and
$f_1(x_1, x_2, x_3, x_4) = x_1 x_2 - x_3$,
$f_2(x_1, x_2, x_3, x_4) = x_1 x_3 - x_4$. In this system, if we use the
lexicographic monomial ordering we have the monomials
$t_{1 1} = x_1 x_2$,
$t_{1 2} = - x_3$, $t_{2 1} = x_1 x_3$, $t_{2 2} = - x_4$. The following
figure shows the graph $G_{\lambda}$ for the labeling
$\lambda(1, 1) = +1$, $\lambda(1, 2) = -1$, $\lambda(2, 1) = +1$ and
$\lambda(2, 2) = -1$.

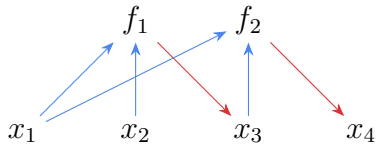
\begin{center}
\begin{tikzpicture}
	\tikzset{
		inner sep=0,outer sep=3,
		neg/.style={color=Gred},
		posi/.style={color=Gblue}}
	\def\hgap{1.5}
	\def\vgap{1.5}
	\node (f1) at (0,\vgap) {$f_1$};
	\node (f2) at (\hgap,\vgap) {$f_2$};

	\node (x1) at (-\hgap,0) {$x_1$};
	\node (x2)at (0,0) {$x_2$};
	\node (x3) at (\hgap,0) {$x_3$};
	\node (x4) at (2*\hgap,0) {$x_4$};

	\path[-{Stealth[length=5pt]},line width=.4pt,black]
	(f1) edge[neg] (x3)
	(f2) edge[neg] (x4)
	(x1) edge[posi] (f1)
	(x2) edge[posi] (f1)
	(x1) edge[posi] (f2)
	(x3) edge[posi] (f2);
\end{tikzpicture}
\end{center}

\begin{definition}[\textsc{Proper Labeling}] \label{def:properLabel}
  Let $\vecf \in \bbF[\vecx]^m$ with a labeling $\lambda$.
    We say that the labeling $\lambda$ is \textit{proper} if the following
  conditions hold.\\[-6mm]
  \begin{enumerate}
    \item[(1)] For all $i$, either $\lambda(i, j) \in \set{-1, 1}$ for all $j$, or $\lambda(i, j) = 0$ for all $j$.
    \item[(2)] If two monomials $t_{i j}$, $t_{i j'}$ contain the same
          variable $x_k$, then $\lambda(i, j) = \lambda(i, j')$.
    \item[(3)] If $\lambda(i, j) = -1$, then $t_{i j}$ is multilinear.
    \item[(4)] If $x_k$ is a variable in the monomials
          $t_{i j}$, $t_{i' j'}$,
          with $i \neq i'$ and $\lambda(i, j) = -1$, then
          $\lambda(i', j') = +1$.
    \item[(5)] If $\lambda(i, j) \neq 0$, then there exists a $j'$ such that
          $\lambda(i, j') = -1$.
    \item[(6)] The labeling graph $G_{\lambda}$ contains no directed cycles.
  \end{enumerate}
\end{definition}

\noindent We give an equivalent way to understand the definition of a proper labeling.
\begin{itemize}[label=$\triangleright$]
\item Condition (1) : there is a partition of the polynomial system $\vecf$ into polynomial systems $\vecg$ and $\vech$ such that all monomials in $\vecg$ are labeled in $\set{+1, -1}$ and all monomials in $\vech$ are labeled $0$.
\item Condition (2) : each polynomial $g_i$ in $\vecg$ can be written as $g_i = g_i^+ + g_i^{-}$, such that $g_i^+$ and $g_i^-$ are polynomials on a disjoint set of variables.
\item Condition (3) : Each $g_i^-$ is multilinear.
\item Condition (4) : Any variable $x_k$ can appear in at most one of the $g_i^-$. Moreover, if an $x_k$ appears in some $g_i^-$, it does not appear in any $h_j$ in $\vech$.
\item Condition (5) : Every $g_i^-$ involves at least one variable.
\item Condition (6) : The graph $G_{\lambda}$ is essentially between polynomials in $\vecg$ and the variables that appear in them, with an edge $(g_i \to x_k)$ if $x_k$ appears in $g_i^-$ or an edge $(x_k \to g_i)$ if $x_k$ appears in $g_i^+$.
\item Note that $\deg^\lambda(g_i) = \deg(g_i^-)$, whereas $\deg^\lambda(h_j) = \deg(h_j)$.
\end{itemize}

\noindent It is easy to see that the trivial labeling $\lambda(i, j) = 0$ is always proper.
As we will see this special case of the Labeled \CWT\! corresponds to the
original \CWT\!. Note that in this case the
labeling graph $G_{\lambda}$ is an empty graph. Also, given a
system of polynomials $\vecf$ and a labeling $\lambda$, it is possible to
check in polynomial time whether the labeling $\lambda$ is proper or not.

\bigskip
\paragr{Example 2 (continued).} It is an instructive exercise to verify that the labeling $\lambda$ specified was indeed a proper labeling of $\vecf$.
\bigskip

\begin{theorem}[\textsc{Labeled Chevalley-Warning Theorem}]
  \label{thm:labeledChevalleyWarningTheorem}
    Let $\F_q$ be a finite field with characteristic $p$ and
  $\vecf \in \F_q[\vecx]^m$. If $\lambda$ is a proper labeling of
  $\vecf$ with $\sum_{i = 1}^m \deg^{\lambda}(f_i) < n$, then $|\calM_{\vecf}| = 0$. In particular, $\abs{\calV_{\vecf}} \equiv 0 \mod{p}$.
\end{theorem}
\begin{proof}
We can re-write $\CW_{\vecf}(x)$ as
$\sum_{S \subseteq [m]} \prod_{i \in S} (-1)^{|S|} f_i^{p-1}$.
We'll show that every monomial appearing in the expansion of
$\prod_{i \in S} f_i^{p-1}$ will have at least one variable with degree at most
$p-1$. For simplicity, we focus on the case $S = [m]$ and the other cases of
$S$ follow similarly.

We index a monomial of $\prod_{i \in [m]} f_i^{p-1}$ with a tuple
$((j_{11},j_{12}, \dots, j_{1(p-1)}), \dots, (j_{m1}, \dots, j_{m(p-1)}))$ with
$1 \leq j_{i\ell} \leq L_{i}$ where $L_i$ is the number of monomials in the
explicit representation of $f_i$.
The coordinates $(j_{i1}, \dots, j_{i(p-1)})$ represent the indices of the
monomials chosen from each of the $p-1$ copies of $f_i^{p-1}$. More succinctly,
we have $t = \prod_{i=1}^m \prod_{\ell=1}^{p-1} t_{i,j_{i\ell}}$.\\[-0.8mm]

\noindent {\em Case 1. $\lambda(i,j_{i\ell}) \in \set{0, -1}$, for all $(i,\ell)$:}\\
Here, $\deg(t) \le (p-1) \sum_{i=1}^m \deg^{\lambda}(f_i)$ which, by our assumption, is strictly less than $(p-1)n$. Hence, there is a variable with degree less than $p-1$.\\[-0.8mm]

\noindent {\em Case 2. There is a unique $i$ with $\lambda(i,j_{i\ell}) = +1$ for some $\ell$:} (warmup for case 3)\\
That is, for all $i'\neq i$, $\lambda(i',j_{i'\ell})  \in \set{0, -1}$. By condition (5) of proper labeling there exists a $j' \neq j_{i\ell}$ such that  $\lambda(i,j') = -1$. Let $x_k$ be a variable in the monomial $t_{ij'}$. By condition (2), $x_k$ is not present in the monomial $t_{i,j_{i\ell}}$ and by condition (3), its degree in $(t_{i,j_{i,1}}, \dots, t_{i,j_{i,p-1}})$ is at most $p-2$. Additionally, by condition (4), any monomial of $f_{i'}$ for $i' \neq i$ containing $x_k$ must have label $+1$, but $\lambda(i',j_{i',\ell})$ are all in $\set{0, -1}$. Hence, $x_k$ does not appear in any other monomial of $t$ and its degree on $t$ is equal to its degree in $(t_{i,j_{i,1}} \cdots t_{i,j_{i,p-1}})$, which is strictly less than $p-1$.\\[-0.8mm]

\noindent {\em Case 3. $I = \set{i : \lambda(i,j_{i\ell}) = +1 \text{ for some } \ell}$:}\\
In the labeling graph $G_\lambda$, let $i \in I$ be such that there is no path from $f_i$ to any other $f_{i'}$ for $i' \in I$. Such an $i$ exists due to acyclicity of $G_{\lambda}$, i.e. condition (6). Let $\ell$ be such that $\lambda(i,j_{i\ell}) = +1$. Again, by condition (5) of proper labeling there exists a $j' \neq j_{i\ell}$ such that  $\lambda(i,j') = -1$. Let $x_k$ be a variable in the monomial $t_{ij'}$. By condition (2), $x_k$ is not present in the monomial $t_{i,j_{i\ell}}$ and by condition (3), its degree in $(t_{i,j_{i,1}}, \dots, t_{i,j_{i,p-1}})$ is at most $p-2$. Additionally, by condition (4), any monomial of $f_{i'}$ for $i' \neq i$ containing $x_k$ must have label $+1$. For $i' \notin I$, $\lambda(i',j_{i',\ell})$ are all in $\set{0, -1}$. And for $i' \in I$, variable $x_k$ cannot appear with $+1$ label in $f_{i'}$ by our choice of $f_i$. Hence, $x_k$ does not appear in any other monomial of $t$ and its degree on $t$ is equal to its degree on $(t_{i,j_{i,1}} \cdots t_{i,j_{i,p-1}})$, which is strictly less than $p-1$.

\end{proof}

\medskip
\noindent We are now ready to prove the $\PPAp$-hardness of $\cChevalleyp$.

\begin{lemma} \label{lem:combinedChevalleyPPApHardness}
    For all primes $p$, $\cChevalleyp$ is $\PPAp$-hard.
\end{lemma}

\begin{proof}    We prove that $\lonelyp \reducible \cChevalleyp$. Let us assume (without loss of
    generality from \autoref{lem:restricted-lonely}) that the $\lonelyp$
    instance has a single distinguished vertex represented by $0^n$. We'll assume that $0^n$ is isolated, otherwise, no further reduction is necessary.\medskip

    \paragr{Pre-processing.} We slightly
    modify the given circuit $\calC$ by defining
    $\calC': \bbF_p^n \to \bbF_p^n$ as follows:
    \[ \calC'(v) = \left\{
                \begin{array}{cl}
                  v         &     \text{, if } \calC^p(v) \neq v \\
                  \calC(v)  &     \text{, otherwise}
                \end{array}
                \right.
    \]
    Since $p$ is a prime, a vertex $v \in \bbF_p^n$ has $\deg(v) = 1$ if and only if
    $\calC^p(v) = v$ and $\calC(v) \neq v$. By modifying the circuit, we changed this
    condition to just $\calC'(v) \neq v$, which facilitates our reduction.

    Circuit
    $\calC'$ is composed of the $\bbF_p$-addition $(+)$, $\bbF_p$-multiplication $(\times)$
    and the constant $(1)$ gates. However, we require the input of $\cChevalleyp$ to be a
    zecote polynomial system, and so we further modify the circuit $\calC'$ to eliminate all the constant $(1)$ gates, without changing it's behavior -- this is possible because we assume $\calC'(0^n) = 0^n$.
                \begin{claim}
        \label{cl:circuitWithout1}
          Given circuit $\calC'$ with $(+,\times, 1)$ gates, there exists circuit
          $\bar{\calC}$ with $(+, \times)$ gates such that
          \[ \bar{\calC}(\vecv) = \infork{0^n & \text{, if } \vecv = 0^n \\ \calC'(\vecv) & \text{, otherwise}}
          \]
    \end{claim}
    \begin{proof}[Proof of \autoref{cl:circuitWithout1}]
    	We replace all instances of the $(1)$ gate by the function $\mathds{1}_{\set{\vecv \ne 0^n}}$, which we can compute using only $(+, \times)$ gates as follows: For any $x, y \in \bbF_p$, observe that $\mathds{1}_{\set{x \ne 0}} \vee \mathds{1}_{\set{y \ne 0}} = x^{p-1} + y^{p-1} - x^{p-1} y^{p-1}$. We can thus recursively compute $\bigvee_{i=1}^n \mathds{1}_{\set{v_i \ne 0}}$ using only $(+, \times)$ gates. Thus, $\bar{\calC}(\vecv) = \calC'(\vecv)$ for all $\vecv \ne 0^n$. And $\bar{\calC}(0^n) = 0^n$, since $\bar{\calC}$ is computed with only $(+, \times)$ gates.
    \end{proof}

        Thus, we can transform our original circuit $\calC$ into a circuit $\bar{\calC}$ with just $(+,\times)$ gates. For simplicity, we'll write $\bar{\calC}$ as simply $\calC$ from now on.\medskip

        As an intermiate step in the reduction we describe a system of
        polynomials $\vecf_{{\calC}}$ over $2n+s$ variables
        $(x_1, \dots, x_n, z_1, \dots, z_s, y_{1}, \dots, y_{n})$, where $s$ is the size
        of the circuit $\calC$. The variables
        $\vecx = (x_1, \dots,x_n)$ correspond to the input of $\calC$, the
        variables $\vecy = (y_1, \dots, y_n)$ correspond to the output and the variables
        $\vecz = (z_1, \dots, z_s)$  correspond to the gates of $\calC$. For an addition
        gate $(+)$ we include a polynomial of the form
        \[           f(a_1, a_2, a_3) =  a_2 + a_3 - a_1, \]
        where $a_1$ is the variable corresponding to the output of the $(+)$ gate and $a_2,a_3$
        are the variables corresponding to its two inputs. Similarly for a multiplication
        $(\times)$ gate, we include a polynomial of the form
        \[           f(a_1,a_2,a_3) = a_2 \cdot a_3 -  a_1 \]
        Finally, for the output of the circuit, we include the polynomial
        \[              f(a,y_i) =  a -  y_i, \]
        where $a$ is the variable corresponding to the $i$-th output gate of
        $\calC$.
        It holds that
        \[  \calC(\vecx) = \vecy \quad \Longleftrightarrow \quad
         \vecf_{\calC}(\vecx,\vecy,\vecz) = \boldsymbol{0}.\]

		\noindent We now describe the reduction from an instance of $\lonely_p$ to that of $\cChevalley_p$. In order to do this, we need to specify a system of polynomials $(\vecg, \vech)$ and a permutation $\sigma$ such that $|\sigma| = p$. In addition, we will provide a proper labeling $\lambda$ for $\vecg$ satisfying the degree condition. We will also ensure that $\a{\sigma}$ acts freely on $\calV_{\vecg} \cap \overline{\calV_{\vech}}$. And hence, the only valid solutions for the resulting $\cChevalley_p$ instance will be $\vecx \in \calV_{\vecg} \cap \calV_{\vech}$.

        \paragraph{Definition of $\vecg$.} The polynomial system $\vecg$ contains the
              following systems of polynomials.
             \[
              \begin{array}{c}
                   \vecf_{\calC}(\vecx_1, \vecx_2, \vecz_{1,2})\\
                   \vecx_2 - \vecx_3 \\
                   \vecf_{\calC}(\vecx_3, \vecx_4, \vecz_{3,4})\\
                   \vecx_4 - \vecx_5 \\
                   \vdots \\
                   \vecf_{\calC}(\vecx_{2p-1}, \vecx_{2p}, \vecz_{2p-1,2p})
              \end{array}
              \]
              Note that there are $N = (2n+s)p$ variables in total.

        \medskip
        \paragraph{Labeling $\lambda$ of $\vecg$.} For the polynomials belonging to a
        system  of the form $\vecf_{\calC}$, the labeling is equal to
        $-1$ for the monomials corresponding to the output of each gate
        and $+1$, otherwise. For instance, let $a_2 + a_3 - a_1$ be the $i$-th
        polynomial of $\vecg$ corresponding to a $(+)$ gate and let
        $a_1 \prec a_2 \prec a_3$, then
        $\lambda(i,1) = -1$ and $\lambda(i,2) = \lambda(i,3) = +1$.

        For the polynomials belonging to a system of the form $\vecx_i - \vecx_{i+1}$,
        the labeling is equal to $-1$ for the monomials with variables in $\vecx_{i+1}$
        and $+1$ for the monomials with variables in  $\vecx_{i}$.

        \begin{claim}
          \label{cl:properLabelingForf1}
          The labeling $\lambda$ for $\vecg$ is proper.
        \end{claim}

\begin{proof}[Proof of \autoref{cl:properLabelingForf1}]
        By \autoref{def:properLabel}, the labeling $\lambda$ is proper if the following conditions hold.
      \begin{description}
        \item[Condition 1.] \emph{For all $i$, either $\lambda(i, j) \in \set{-1, 1}$ for all $j$, or $\lambda(i, j) = 0$ for all $j$.} \\
              In the labeling $\lambda$, there are no labels equal to
              $0$, so this condition holds trivially.
        \item[Condition 2.] \emph{If two monomials $t_{i j}$, $t_{i j'}$ contain the same
              variable $x_k$, then $\lambda(i, j) = \lambda(i, j')$.}\\
              By construction of $\vecg$, no variable appears twice in the same
              polynomial with a different labeling. For polynomials of
              $\vecf_{\calC}$, this holds because the output variable of a gate is
              not simultaneously an input variable and all
              input variables have the same labeling. For polynomials in a system of
              the form $\vecx_i - \vecx_{i+1}$, each polynomial contains two different
              variables.
        \item[Condition 3.] \emph{If $\lambda(i, j) = -1$, then $t_{i j}$ is multilinear.} \\
              For polynomials of $\vecf_{\calC}$, only the output variable of a
              gate has label $-1$ and by definition this monomial is linear. For
              polynomials in a system of the form $\vecx_i - \vecx_{i+1}$, all monomials
              are linear, so the condition holds trivially.
        \item[Condition 4.] \emph{If $x_k$ is a variable in the monomials $t_{i j}$,
              $t_{i' j'}$,
              with $i \neq i'$ and $\lambda(i, j) = -1$, then $\lambda(i', j') = +1$.} \\
              Observe that all monomials with label $-1$ contain only a single
              variable, so we refer to a monomial $x_k$ with label $-1$.
              For a polynomial in $\vecf_{\calC}$, a monomial $x_k$ with label $-1$
              corresponds to the output of a gate. Hence, if $x_k$ appears in other
              monomials of $\vecf_{\calC}$, these monomials correspond to inputs
              and have label $+1$. Also, if $x_k$ is an output variable of
              $\vecf_{\calC}$, then it might appear in a polynomial of the form
              $a_1 - a_2$. However, by construction the monomials of $x_i - x_{i+1}$ that
              correspond to output variables of $\vecf_{\calC}$ have label $+1$.
        \item[Condition 5.] \emph{If $\lambda(i, j) \neq 0$, then there exists a $j'$ such
              that $\lambda(i, j') = -1$.} \\
              By the definition of $\lambda$, all polynomials of $\vecg$ have a monomial
              with label $-1$. These are the monomials that correspond to the outputs of
              a gate for the systems of the form $\vecf_{\calC}$ and the monomials
              that correspond to $\vecx_{i+1}$ for the systems of the form
              $\vecx_{i}- \vecx_{i+1}$.
        \item[Condition 6.] \emph{The labeling graph $G_{\lambda}$ contains no cycles.}\\
              Each system of the form $\vecx_i - \vecx_{i+1}$ has incoming
              edges with variables appearing only in the $i$-th copy of
              $\vecf_{\calC}$ and outgoing edges with variables
              appearing only in the $(i+1)$-th copy of $\vecf_{\calC}$. Also, the
              variables appearing on the $i$-th copy of $\vecf_{\calC}$ might
              appear only in the systems $\vecx_{i-1} - \vecx_i$ and
              $\vecx_i - \vecx_{i+1}$. Hence, $G_{\lambda}$ has no cycles that contain
              vertices of two different copies of $\vecf_{\calC}$ or of a copy of
              $\vecf_{\calC}$ and a system of the form $\vecx_{i-1} - \vecx_i$.

              It is left to argue that the labeling graph restricted to a copy of
              $\vecf_{\calC}$  does not have any cycles. Let the vertices of
              $\vecf_{\calC}$ be ordered according to the topological ordering of
              $\calC$. This restricted part of $G_{\lambda}$ corresponds exactly to
              the graph of $\calC$, which by definition is a DAG. Hence, $G_{\lambda}$
              contains no cycles.\qedhere
      \end{description}
\end{proof}

        We also need to show that for this labeling $\vecg$ satisfies the labeled
        Chevalley condition.
        \begin{claim}
          \label{cl:ChevalleyForf1}
          The labeled Chevalley condition
          $\sum_{i = 1}^{m_g} \deg^{\lambda}(g_i) < N$ holds for $\vecg$ with
          labeling $\lambda$.
        \end{claim}
        \begin{proof}              Each polynomial of $\vecg$ has a unique monomial with $\lambda(i,j) = -1$
              and this monomial has degree $1$. Thus,
              $\sum_{i = 1}^{m_g} \deg^{\lambda}(g_{i}) = m_g$.
              On the other hand, the $i$-th polynomial of $\vecg$ has exactly one
              variable that has not appeared in any of the previous polynomials. More
              specifically, the number of variables is equal to $m_g + n$, where $n$
              is the size of the input of $\calC$. Hence, the labeled Chevalley
              condition holds for $\vecg$.
        \end{proof}

        \paragraph{Definition of $\vech$.} The system of polynomials $\vecg$ allows
        us to compute the $p$ vertices given by ${\calC}^{i}(\vecx)$ for
        $i \in \bracks{p+1}$. From the definition of $\lonelyp$ and our pre-processing on
        $\calC$, this group of $p$ vertices is a hyperedge if and only if
        $\calC(\vecx) \neq \vecx$. Since solutions of $\lonelyp$ are lonely
        vertices, we define $\vech$ to exclude $\vecx$ such that
        $\calC(\vecx) \neq \vecx$. Namely, we set $\vech$ to be the system of
        polynomials
                          \[ \vecx_1 - \vecx_2.\]

        \paragraph{Definition of permutation $\boldsymbol{\sigma}$.} In the description
        of $\vecf = (\vecg, \vech)$, we have used the following vector of variables:
                    \[  \vecx = (\vecx_1, \vecx_2, \dots, \vecx_{2p},
                          \vecz_{1,2}, \vecz_{3,4}, \dots, \vecz_{2p-1,2p} )
                    \]
      	We define the permutation $\sigma$ such that $\sigma(\vecx) = (\vecx_3, \vecx_4, \dots, \vecx_{2p},
                \vecx_1, \vecx_2, \vecz_{3,4}, \vecz_{5,6}, \dots, \vecz_{2p-1,2p},
                \vecz_{1,2})$, as illustrated in the following figure. The \textcolor{Gblue}{blue} arrows indicate the polynomials $\vecg$ and the \textcolor{Ggreen}{green} arrows indicate the permutation $\sigma$ in the case of $p=3$.

		
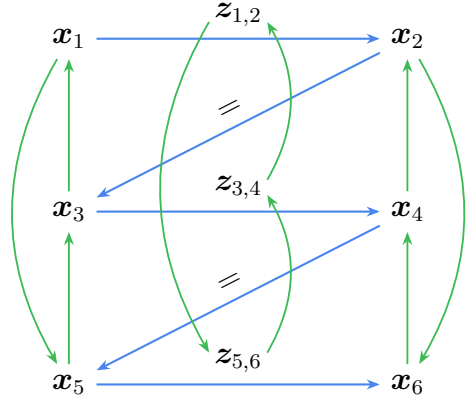
\begin{center}
		\begin{tikzpicture}
			\tikzset{
				zij/.style = {inner sep=0pt, outer sep=2pt, midway, above=0.1cm},
				comp/.style = {{Stealth[length=5pt]}-, draw=Gblue, thick}
			}
			\def\hgap{4.5}
			\def\vgap{2.3}
			\node (x1) at (0,0) {$\vecx_1$};
			\node (x2) at (\hgap, 0) {$\vecx_2$} edge[comp] node[zij] (z12) {$\vecz_{1,2}$} (x1);
			\node (x3) at (0,-\vgap) {$\vecx_3$} edge[comp] node[midway, sloped, above] (eq23) {$=$} (x2);
			\node (x4) at (\hgap, -\vgap) {$\vecx_4$} edge[comp] node[zij] (z34) {$\vecz_{3,4}$} (x3);
			\node (x5) at (0,-2*\vgap) {$\vecx_5$} edge[comp] node[midway, sloped, above] (eq45) {$=$} (x4);
			\node (x6) at (\hgap, -2*\vgap) {$\vecx_6$} edge[comp] node[zij] (z56) {$\vecz_{5,6}$} (x5);

			\path[-{Stealth[length=5pt]},line width=0.7pt,Ggreen]
			(x3) edge (x1)
			(x5) edge (x3)
			(x1) edge[out=-120,in=120] (x5)
			(x4) edge (x2)
			(x6) edge (x4)
			(x2) edge[out=-60,in=60] (x6)
			([shift={(0,0.1)}]z34.east) edge[out=60,in=-60] ([shift={(0,-0.1)}]z12.east)
			([shift={(0,0.1)}]z56.east) edge[out=60,in=-60] ([shift={(0,-0.1)}]z34.east)
			([shift={(0,-0.1)}]z12.west) edge[out=-120,in=120] ([shift={(0,0.1)}]z56.west);
		\end{tikzpicture}
		\end{center}

          \begin{claim}
            \label{cl:sigmaActsFreelyAndOrderp}
            The group $\a{\sigma}$ has order $p$ and acts freely on
            $\calV_{\vecg} \cap \overline{\calV_{\vech}}$.
          \end{claim}

          \begin{proof}              In order to see that $\abs{\sigma} = p$, note that the input of $\sigma$ consists
              of $3p$ blocks of variables. The permutation $\sigma$ performs a rotation of the
              first $2p$ blocks by two positions and of the last $p$ blocks by one
              position.
              All that remains is to show that $\a{\sigma}$ acts freely on
                    $\calV_{\vecg} \cap \overline{\calV_{\vech}}$.
                    First, we show that $\a{\sigma}$ defines a group action on
                    $\calV_{\vecg} \cap \overline{\calV_{\vech}} $, that is
                    for all $\vecx \in \calV_{\vecg} \cap
                    \overline{\calV_{\vech}}$, it holds that $\sigma(\vecx) \in
                    \calV_{\vecg} \cap \overline{\calV_{\vech}}$.
                    Let $\vecx = (\vecx_1, \vecx_2,
                    \dots, \vecx_{2p-1}, \vecx_{2p}, \vecz_{1,2}, \vecz_{3,4}, \dots, \vecz_{2p-
                    1,2p} ) \in \calV_{\vecg} \cap \overline{\calV_{\vech}} $,
                    then
                    \begin{itemize}
                        \item $\vecx \in \calV_{\vecg}$ implies that
                        $\vecf_{\calC}(\vecx_{2i-1},\vecx_{2i},
                        \vecz_{2i-1,2i}) = 0$ for $i \in \bracks{p}$ and
                        $\vecx_{2i} = \vecx_{2i+1}$ for $i \in \bracks{p-1}$
                        \item $\vecx \in \overline{\calV_{\vech}}$ implies that
                        $\vecx_1 \neq \vecx_2$, that is,
                        $\calC(\vecx_{1}) \neq \vecx_1$ since
                        $\vecf_{\calC}(\vecx_{1},\vecx_{2},
                        \vecz_{1,2}) = 0 \Leftrightarrow \vecx_2 =
                        \calC(\vecx_{1}) $.
                    \end{itemize}
                Now, $\sigma(\vecx) = (\vecx_{3}, \vecx_{4},
                \ldots, \vecx_{1}, \vecx_{2}, \vecz_{3,4}, \vecz_{5,6}, \ldots,
                \vecz_{1,2}) \in  \calV_{\vecg} \cap
                \overline{\calV_{\vech}}$ holds because
                \begin{itemize}
                    \item $\vecf_{\calC}(\vecx_{2i-1},\vecx_{2i},
                    \vecz_{2i-1,2i}) = 0$ for $i \in \bracks{p}$ and
                    $\vecx_{2i} = \vecx_{2i+1}$ for $i \in \bracks{p-1}$, which holds
                    from $\vecx \in \calV_{\vecg}$. Additionally, $\vecx_1 = \vecx_{2p}$ holds because we pre-processed $\calC$ such that $\calC^p(\vecx_1) = \vecx_{1}$,
                    \item $\vecx_{3} \neq \vecx_{4}$, which holds because
                    $\vecx_{4}  = \calC(\vecx_{3})$ for $i \in \bracks{p}$
                    and from the definition of $\calC$, $\calC(\vecx_{1})
                    \neq \vecx_1$ implies that $\vecx_{2i} \neq \vecx_{2i-1}$ for all
                    $i \in \bracks{p}$.
                \end{itemize}

              Finally, if $\vecx \in \calV_{\vecg} \cap
              \overline{\calV_{\vech}}$, by
              construction of $\cal{C}$, we have that $\vecx_{2k} \neq \vecx_{2j}$ for
              $k\neq j$ and thus $\sigma(\vecx) \ne \vecx$ simply because $\vecx_3 \ne \vecx_1$. Thus, we conclude that $\a{\sigma}$ acts freely on
              $\calV_{\vecg} \cap \overline{\calV_{\vech}}$.
          \end{proof}

        		\paragraph{Putting it all together.} The solution of this instance of $\cChevalleyp$ cannot be a vector
                $\vecx \in \calV_{\vecg} \cap \overline{\calV_{\vech}}$  with
                $\sigma(\vecx)
                \not\in \calV_{\vecg} \cap \overline{\calV_{\vech}}$ or
                $\sigma(\vecx) = \vecx$, since we know from
                \autoref{cl:sigmaActsFreelyAndOrderp} that $\a{\sigma}$ acts freely on
                $\calV_{\vecg} \cap \overline{\calV_{\vech}}$.
                We also have from \autoref{thm:labeledChevalleyWarningTheorem} that the solution also cannot be a max-degree monomial in the expansion of $\CW_{\vecg}(\bx) = \prod (1-g_i^{p-1})$.
                Thus, the solution must be an $\vecx \neq \vec{0}$ such that
                $\vecf(\vecx) = \vec{0}$. Let $\vecx_1$ denote the first $n$
                coordinates of $\vecx$, then $\vecf(\vecx) = \vec{0}$ implies that
                $\vecx_1 = \calC(\vecx_1)$ and $\vecx \neq \vec{0}$ implies that
                $\vecx_1 \neq \vec{0}$. Hence, $\vecx_1$ corresponds to
                a lonely vertex of the $\lonelyp$ instance.
\end{proof} 
\section{Complete Problems via Small Depth Arithmetic Circuits} \label{sec:simplification}

We now illustrate the significance of the $\PPAp$-completeness of $\cChevalley_p$, by showing that we can reformulate any of the proposed definitions of $\PPAp$, by restricting the circuit in the input to be just constant depth arithmetic formulas with gates $\times \mod{p}$ and $+ \mod{p}$ (we call this class $\ACz_{\F_p}$\footnote{Note that $\ACz_{\F_p}$ is strictly more powerful than $\ACz$ since the Boolean operations of $\set{\land,\lor,\lnot}$ can be implemented in $\ACz_{\F_p}$, but $+ \mod{p}$ cannot be implemented in $\ACz$.}). This result is analogous to the
$\NP$-completeness of $\textsc{SAT}$ which basically shows that
$\textsc{CircuitSAT}$ remains $\NP$-complete even if we restrict the input circuit to be a (CNF) formula of depth $2$.

We define $\sucBipartitep[\ACzp]$ to be the same as $\sucBipartitep$ but with the input circuit being a formula
in $\ACzp$. Similarly, we define $\lonelyp[\ACzp]$, $\leafp[\ACzp]$, etc.
\begin{theorem}\label{thm:simplificationPPAp} For all primes $p$, $\sucBipartitep[\ACzp]$ is $\PPAp$-complete.
\end{theorem}

\begin{remark}
  In \cite{rubinstein16}, a similar simplification theorem was shown for $\PPAD$. In fact, this simplification involves only the $\textsc{End-of-Line}$ problem and does not go through a natural complete problem for $\PPAD$ (see Theorem 1.5 in \cite{rubinstein16}). A similar result can be shown for other $\TFNP$ subclasses, including $\PPA$. However, it is unclear if these techniques also apply to $\PPAp$ classes.
\end{remark}

\smallskip

\noindent \autoref{thm:simplificationPPAp} follows directly from the proof of
\autoref{lem:combinedChevalleyPPAp} by observing that the reduction can be perfomed by an $\ACzp$ circuit. For completeness, we include this proof in \autoref{sec:simplificationProof}.

Since the reductions between $\sucBipartite_p$ and other problems studied in this work (refer to \autoref{sec:equivalences}) can also be implemented as $\mathsf{AC}^0$ circuits, we get the following corollary.

\begin{corollary} \label{cor:simplificationPPAp} For all primes $p$, $\lonelyp[\ACzp]$, $\leafp[\ACzp]$ and $\bipartitep[\ACzp]$ are all $\PPAp$-complete.
\end{corollary}

\newcommand{\NC}{\mathsf{NC}}

\noindent Since $+\mod{p}$ and $\times\mod{p}$ can be simulated in $\NC^1$, we also get the following corollary.

\begin{corollary} \label{cor:simplificationPPApNC1} For all primes $p$, $\lonelyp[\NC^1]$, $\leafp[\NC^1]$ and $\bipartitep[\NC^1]$ are all $\PPAp$-complete.
\end{corollary}

Thus, \autoref{thm:simplificationPPAp} allows us to consider reductions from these $\PPA_p$-complete problems with instances encoded by a shallow formulas rather than an arbitrary circuit. We believe this could be a useful starting point for finding other $\PPAp$-complete problems.

\newcommand{\myshortprob}[5]{
	\noindent
	\begin{minipage}{\linewidth}
		\begin{definition}
			\label{def:#1} \contour{black}{(#2)}\vspace{-5pt}
			\begin{itemize}[leftmargin=2.15cm,itemsep=4pt]
				\item[\slshape Input:] #3
				\item[\slshape Condition:] #4
				\item[\slshape Output:] #5
			\end{itemize}
		\end{definition}
	\end{minipage}\vspace{3mm}
}

\section{Applications of Chevalley-Warning}
\label{sec:applications}

For most of the combinatorial applications mentioned in \autoref{sec:intro-sis}, the proofs utilize restricted versions of the Chevalley-Warning Theorem
that are related to finding binary or short solutions in a system of modular equations. We define
two computational problems to capture these restricted cases. The first problem is about finding
binary non-trivial solutions in a modular linear system of equations, which we call $\bis_q$. The
second is a special case of the well-known short integer solution problem in $\ell_\infty$ norm,
which we denote by $\sis_q$. The computational problems are defined below, where
$N(q)$ denotes the sum of the exponents in the canonical prime factorization of $q$, e.g. $N(4) =
N(6) = 2$. In particular, $N(p) = 1$ for prime $p$ and $N(q_1 q_2) = N(q_1) + N(q_2)$ for all $q_1, q_2$.

\bigskip

\myshortprob{bis}{$\bis_{q}$}
{$\bA \in \bbZ_q^{m \times n}$, a matrix over $\bbZ$}
{$n \geq (m+1)^{N(q)}(q-1)$}
{$\vecx \in \bit^n$ such that $\vecx \neq \vec{0}$ and
              $\bA\vecx \equiv \vec{0} \pmod q$}

\medskip

\myshortprob{sis}{$\sis_{q}$}
{$\bA \in \bbZ_q^{m \times n}$, a matrix over $\bbZ$}
{$n \geq ((m + 1)/2)^{N(q)}(q-1)$}
{$\vecx \in \set{-1, 0, 1}^n$ such that $\vecx \ne \vec{0}$ and $\bA\vecx \equiv \vec{0} \pmod q$}

\noindent $\sis_q$ is a special case of the well-known short integer solution problem in
$\ell_\infty$ norm from the theory of lattices. The totality of this problem is guaranteed even when
$n > m \log_2 q$ by pigeonhole principle; thus, $\sis_{q}$ belongs also to $\PPP$ (for this regime of parameters). However, for the parameters considered in above definitions, the
existence of a solution in the $\bis_q$ and $\sis_q$ is guaranteed through modulo $q$ arguments, which we formally show in the following theorem.

\begin{theorem}
\label{thm:sis/bisPPAq}
For the regime of parameters $n$, $m$ as in \hyperref[def:bis]{Definitions \ref{def:bis}} and \ref{def:sis},
\begin{enumerate}
	\item For all primes $p$ : $\bis_p,\ \sis_p \preceq \Chevalley_p$.\\[-10pt]
	\item For all $q$ : $\bis_{q},\ \sis_{q} \in \FP^{\PPA_q}$, \\[-10pt]
	\item For all $k$ : $\bis_{2^k} \in \FP$, \\[-10pt]
	\item For all $k$, $\ell$ : $\sis_{2^k 3^\ell} \in \FP$.
\end{enumerate}
\end{theorem}
\begin{proof}
{\bf Part 1.} For all primes $p$, $\bis_{p},\ \sis_p \reducible \eChevalleyp$.

\noindent Given an $\bis_p$ instance $\bA = (a_{ij})$, we define a zecote polynomial system as follows
\[
\vecf := \set{f_i(\vecx) = \sum\limits_{j = 1}^n a_{ij} x_{j}^{p-1} \ : \ i \in [m]}
\]
Clearly, $\deg(f_i) = p-1$, so
$\sum_{i = 1}^m \deg(f_i) = m (p-1)$. Since $n \geq (m+1) (p-1)> m (p-1)$, \eqref{eq:CWcondition} is satisfied.
Hence the output of $\eChevalleyp$ is a  solution
$\vecx \ne \mathbf{0}$ such that $\vecf(\vecx) = \vec{0}$. This gives us that
$\vecx^{p-1} \coloneqq (x_1^{p-1}, \dots, x_n^{p-1})$ is binary and satisfies $A \vecx \equiv 0 \mod{p}$.

The reduction of $\sis_p \reducible \eChevalley_p$ follows similarly by defining $f_i(\vecx) := \sum_{j=1}^m a_{ij} x_j^{(p-1)/2}$. This satisfies the \eqref{eq:CWcondition} because $\sum_i \deg(f_i) = m(p-1)/2 < ((m+1)/2)(p-1) \le n$. This ensures that any $\vecx \in \calV_{\vecf}$ satisfies $\vecx^{(p-1)/2} \in \set{-1, 0, 1}^n$ and $A\vecx \equiv 0\mod{p}$.

\noindent {\bf Part 2.} For all $q$ : $\bis_{q},\ \sis_{q} \in \FP^{\PPA_q}$.\\[-3mm]

\noindent We show that $\bis_{q_1q_2} \reducible \bis_{q_1} \bamp \bis_{q_2}$. Hence if $\bis_{q_1} \in \FP^{\PPA_{q_1}}$ and $\bis_{q_2} \in \FP^{\PPA_{q_2}}$, then $\bis_{q_1q_2} \in \FP^{\PPA_{q_1q_2}}$. The proof of Part 2 now follows by induction.

Given a $\bis_{q_1q_2}$ instance $\bA \in \bbZ^{m \times n}$, we divide $\bA$ along the columns into $n_1 = (m+1)^{N(q_1)} (q_1 - 1)$ submatrices denoted by
$\bA_1, \dots, \bA_{n_1}$, each of size at least $m \times n_2$,
with $n_2 = \floor{n/n_1}$ (if
$n/n_1$ is not an integer, then we let $\bA_{n_1}$ has more than $n_2$ columns).
Each $\bA_i$ is an instance of $\bis_{q_2}$, since
\[n_2 = \floor{n/n_1} \geq (m+1)^{N(q_2)}\floor{(q - 1)/(q_1 - 1)}
      \geq (m+1)^{N(q_2)}(q_2 - 1).\]
Let $\vecy_i \in \bit^{n_2}$ be any solution to $\bA_i \vecy_i \equiv 0 \mod{q_2}$. We define the matrix $\bB \in \bbZ^{m \times n_1}$ where the
$i$-th column is equal to $\bA_i \vecy_i / q_2$; this has integer entries since $\bA_i \vecy_i \equiv 0 \pmod {q_2}$.
Now, by our choice of $n_1$, we have that $\bB$ is an instance of $\bis_{q_1}$. Let $\vecz = (z_1, \dots, z_{n_1}) \in \bit^{n_1}$ be any solution to $\bB \vecz = 0 \pmod {q_1}$.

Finally, we define $\vecx := (z_1 \vecy_1, \dots, z_{n_1}\vecy_{n_1}) \in \bit^n$.
Observe that since $\vecy_i$ and $\vecz$ are binary, $\vecx$ is also binary.
Additionally,
\[
\bA \vecx
~=~ \sum\limits_{i = 1}^{n_1} (\bA_i \vecy_i) z_i
~=~ q_2 \sum\limits_{i = 1}^{n_1}\frac{ \bA_i \vecy_i}{q_2} z_i
~=~ q_2 \bB \vecy
~\equiv~ \vec{0} \mod{q_1q_2}.
\]
 Hence, $\vecx$ is a solution of the original $\bis_{q_1q_2}$ instance $\bA\vecx \equiv 0 \mod{q_1 q_2}$.
This concludes the proof of $\bis_q \in \FP^{\PPA_q}$. The proof of $\sis_{q} \in \FP^{\PPA_q}$ follows similarly, by observing that if $\vecy_i$ and $\vecz$ have entries in $\set{-1, 0, 1}$ then so does $\vecx$.\\

\noindent {\bf Parts 3, 4.} For all $k$, $\ell$ : $\bis_{2^k} \in \FP$ and $\sis_{2^k 3^\ell} \in \FP$.\\[-3mm]

\noindent Observe that $\bis_{2}$ (hence also $\sis_2$) and
$\sis_{3}$ are solvable in polynomial time via Gaussian elimination.
Combining this with the reduction $\bis_{q_1q_2} \reducible \bis_{q_1} \bamp \bis_{q_2}$ completes the proof (similarly for $\sis$).
\end{proof}

\noindent Note that for a prime $p$ and any $k$, we have from
\autoref{thm:prime-characterization}, that $\PPA_{p^k} = \PPA_p$. Additionally, \autoref{thm:turing-closure} shows that $\PPAp$ is closed under Turing reductions, so we have the following corollary.

\begin{corollary}
For all primes $p$ and all $k$ : $\bis_{p^k},\ \sis_{p^k} \in \PPA_p$.
\end{corollary}

\noindent Even though the $\sis_q$ problem is well-studied in lattice theory, not many results
are known in the regime we consider where $q$ is a constant. Our results show that solving $\eChevalley_\pnum$
is at least as hard as finding short integer solutions in $\pnum$-ary lattices for a
specific range of parameters. More specifically, our reduction assumes that $\qnum$
is a constant and, thus, it does not depend on the input lattice, and that the
dimension $n$ of lattice is related to the number of constraints in the dual as
$n > ((m+1)/2)^{N(q)}(q-1)$. On the other hand, we showed (in Parts 3, 4) that there are $q$-ary lattice for which finding short integer solutions is easy.

\section{\boldmath Structural Properties of $\PPA_q$} \label{sec:structural}

In this section, we prove the structural properties of $\PPA_q$ outlined in \autoref{sec:intro-struc}.

\paragraph{\boldmath Relation to $\PMOD_q$.}

Buss and Johnson \cite{buss12propositional,johnson11reductions} defined a problem $\textsc{Mod}_q$, which is almost identical to $\lonely_q$, with the only difference being that the $q$-dimensional matching is over a power-of-$2$ many vertices encoded by $C : \bit^n \to \bit^n$, with no designated vertices, except when $q$ is a power of $2$ in which case we have one designated vertex. The class $\mathsf{PMOD}_q$ is then defined as the class of total search problems reducible to $\modn_q$. The restriction of number of vertices to be a power of $2$, which arises as an artifact of the binary encoding of circuit inputs, makes the class $\mathsf{PMOD}_q$ slightly weaker than $\PPA_q$.

To compare $\PPA_q$ and $\PMOD_q$, we define a restricted version of $\lonely_q$, where the number of designated vertices is exactly $k$; call this problem $\lonely_q^k$. Clearly, $\lonely_q^k$ reduces to $\lonely_q$. We show that a converse holds, but only for prime $p$; see \autoref{sec:pmod-equivalence} for proof.

\begin{lemma}\label{lem:restricted-lonely}
	For all primes $p$ and $k \in \set{1, \ldots, p-1}$, $\lonely_p$ reduces to $\lonely_p^k$.
\end{lemma}

\begin{corollary}\label{cor:ppa-pmod}
	For all primes $p$, $\PPA_p = \PMOD_p$.
\end{corollary}

\newcommand{\revand}{\rotatebox[origin=c]{180}{\amp}}

For composite $q$, however, the two classes are conceivably different. In contrast to \autoref{thm:prime-characterization}, it is shown in \cite{johnson11reductions} that $\PMOD_q = \revand_{p|q}~\PMOD_p$, where the operator `$\revand$' is defined as follows: For any two search problem classes $\M_0$, $\M_1$ with complete problems $S_0$, $S_1$, the class $\M_0~\revand~\M_1$ is defined via the complete problem $S_0~\revand~S_1$ defined as follows: Given $(x_0, x_1) \in \Sigma^* \times \Sigma^*$, find a solution to either $x_0$ interpreted as an instance of $S_0$ or to $x_1$ interpreted as an instance of $S_1$. In other words, $\M_1 \ \revand\ \M_2$ is no more powerful than either $\M_1$ or $\M_2$.  In particular, it holds that $\M_1~\revand~\M_2 = \M_1~\cap~\M_2$, whereas $\M_1~\amp~\M_2 \supseteq \M_1~\cup~\M_2$. Because of this distinction, unlike \autoref{thm:prime-characterization}, the proof of $\PMOD_{p^k} = \PMOD_p$ in \cite{johnson11reductions} follows much more easily since for any odd prime $p$ it holds that $2^n \not\equiv 0 \mod{p}$ and hence a $\lonely_{p^k}$ instance readily reduces to a $\lonely_p$ instance.

\subsection{\boldmath $\PPAD\subseteq\PPA_q$}\label{sec:ppad-inclusion}

Johnson \cite{johnson11reductions} already showed that $\PPAD \subseteq \PMOD_q$ which implies that $\PPAD \subseteq \PPA_q$. We present a simplified version of that proof.

We reduce the $\PPAD$-complete problem $\textsc{End-of-Line}$ to $\lonely_q$. An instance of $\textsc{End-of-Line}$ is a circuit $C$ that implicitly encodes a directed graph $G = (V,E)$, with in-degree and out-degree at most $1$ and a designated vertex $v^*$ with in-degree $0$ and out-degree $1$.

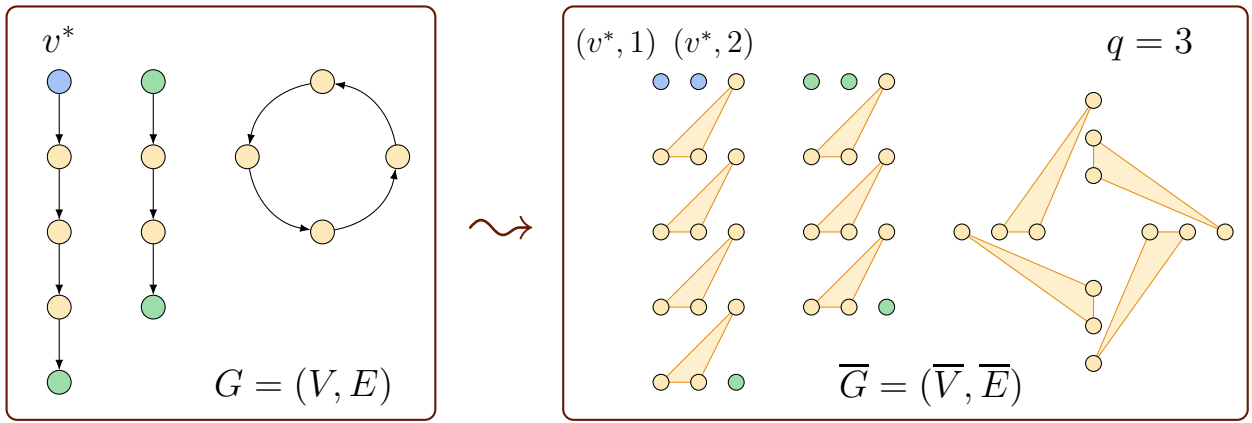
\begin{center}
	\begin{tikzpicture}
\tikzset{
	inner sep=0,outer sep=0,
	vert/.style={circle, draw, minimum size=9pt, fill=Gyellow!30},
	mvert/.style={circle, draw, minimum size=6pt, fill=Gyellow!30},
	sol/.style={fill=Ggreen!50},
	designated/.style={fill=Gblue!50},
	hedge/.style={draw=myGold, fill=Gyellow!20}}

\def\hgap{1.25}
\def\vgap{1}
\node[vert,designated] (v0) at (0,4) {};
\node[vert] (v1) at ([shift={(0,-\vgap)}]v0) {};
\node[vert] (v2) at ([shift={(0,-\vgap)}]v1) {};
\node[vert] (v3) at ([shift={(0,-\vgap)}]v2) {};
\node[vert,sol] (v4) at ([shift={(0,-\vgap)}]v3) {};

\node[vert,sol] (v5) at ([shift={(\hgap,0)}]v0) {};
\node[vert] (v6) at ([shift={(0,-\vgap)}]v5) {};
\node[vert] (v7) at ([shift={(0,-\vgap)}]v6) {};
\node[vert,sol] (v8) at ([shift={(0,-\vgap)}]v7) {};

\node[vert] (v9) at ([shift={(\hgap,-1)}]v5) {};
\node[vert] (v10) at ([shift={(1,-1)}]v9) {};
\node[vert] (v11) at ([shift={(2,0)}]v9) {};
\node[vert] (v12) at ([shift={(1,1)}]v9) {};

\node[above=0.4] at (v0) {\large $v^*$};
\node[below] at ([shift={(2,-0.8)}]v8) {\large $G = (V,E)$};
\draw[rounded corners=4pt, draw=Sepia, thick] ([shift={(-0.7,1)}]v0) rectangle ([shift={(5,-4.5)}]v0);

\path[-{Latex[length=4pt]},line width=.4pt,black]
(v0) edge (v1)
(v1) edge (v2)
(v2) edge (v3)
(v3) edge (v4)
(v5) edge (v6)
(v6) edge (v7)
(v7) edge (v8)
(v9)  edge[bend right=35] (v10)
(v10) edge[bend right=35] (v11)
(v11) edge[bend right=35] (v12)
(v12) edge[bend right=35] (v9);

\node[Sepia] at ([shift={(5.9,-2*\vgap)}]v0) {\Huge $\leadsto$};

\def\hgap{0.5}
\def\vgap{1}
\coordinate (u0-1) at ([shift={(8,0)}]v0);
\coordinate (u0-2) at ([shift={(\hgap,0)}]u0-1);
\coordinate (u0-3) at ([shift={(\hgap,0)}]u0-2);
\coordinate (u1-1) at ([shift={(0,-\vgap)}]u0-1);
\coordinate (u1-2) at ([shift={(\hgap,0)}]u1-1);
\coordinate (u1-3) at ([shift={(\hgap,0)}]u1-2);
\coordinate (u2-1) at ([shift={(0,-\vgap)}]u1-1);
\coordinate (u2-2) at ([shift={(\hgap,0)}]u2-1);
\coordinate (u2-3) at ([shift={(\hgap,0)}]u2-2);
\coordinate (u3-1) at ([shift={(0,-\vgap)}]u2-1);
\coordinate (u3-2) at ([shift={(\hgap,0)}]u3-1);
\coordinate (u3-3) at ([shift={(\hgap,0)}]u3-2);
\coordinate (u4-1) at ([shift={(0,-\vgap)}]u3-1);
\coordinate (u4-2) at ([shift={(\hgap,0)}]u4-1);
\coordinate (u4-3) at ([shift={(\hgap,0)}]u4-2);

\filldraw[hedge] (u0-3) -- (u1-1) -- (u1-2) -- cycle;
\filldraw[hedge] (u1-3) -- (u2-1) -- (u2-2) -- cycle;
\filldraw[hedge] (u2-3) -- (u3-1) -- (u3-2) -- cycle;
\filldraw[hedge] (u3-3) -- (u4-1) -- (u4-2) -- cycle;

\node[mvert,designated] at (u0-1) {};
\node[mvert,designated] at (u0-2) {};
\node[mvert] at (u0-3) {};
\node[mvert] at (u1-1) {};
\node[mvert] at (u1-2) {};
\node[mvert] at (u1-3) {};
\node[mvert] at (u2-1) {};
\node[mvert] at (u2-2) {};
\node[mvert] at (u2-3) {};
\node[mvert] at (u3-1) {};
\node[mvert] at (u3-2) {};
\node[mvert] at (u3-3) {};
\node[mvert] at (u4-1) {};
\node[mvert] at (u4-2) {};
\node[mvert,sol] at (u4-3) {};

\def\hgap{0.5}
\def\vgap{1}
\coordinate (u5-1) at ([shift={(4*\hgap,0)}]u0-1);
\coordinate (u5-2) at ([shift={(\hgap,0)}]u5-1);
\coordinate (u5-3) at ([shift={(\hgap,0)}]u5-2);
\coordinate (u6-1) at ([shift={(0,-\vgap)}]u5-1);
\coordinate (u6-2) at ([shift={(\hgap,0)}]u6-1);
\coordinate (u6-3) at ([shift={(\hgap,0)}]u6-2);
\coordinate (u7-1) at ([shift={(0,-\vgap)}]u6-1);
\coordinate (u7-2) at ([shift={(\hgap,0)}]u7-1);
\coordinate (u7-3) at ([shift={(\hgap,0)}]u7-2);
\coordinate (u8-1) at ([shift={(0,-\vgap)}]u7-1);
\coordinate (u8-2) at ([shift={(\hgap,0)}]u8-1);
\coordinate (u8-3) at ([shift={(\hgap,0)}]u8-2);

\filldraw[hedge] (u5-3) -- (u6-1) -- (u6-2) -- cycle;
\filldraw[hedge] (u6-3) -- (u7-1) -- (u7-2) -- cycle;
\filldraw[hedge] (u7-3) -- (u8-1) -- (u8-2) -- cycle;

\node[mvert,sol] at (u5-1) {};
\node[mvert,sol] at (u5-2) {};
\node[mvert] at (u5-3) {};
\node[mvert] at (u6-1) {};
\node[mvert] at (u6-2) {};
\node[mvert] at (u6-3) {};
\node[mvert] at (u7-1) {};
\node[mvert] at (u7-2) {};
\node[mvert] at (u7-3) {};
\node[mvert] at (u8-1) {};
\node[mvert] at (u8-2) {};
\node[mvert,sol] at (u8-3) {};

\coordinate (u9-1) at ([shift={(1,0)}]u7-3);
\coordinate (u9-2) at ([shift={(\hgap,0)}]u9-1);
\coordinate (u9-3) at ([shift={(\hgap,0)}]u9-2);
\coordinate (u10-1) at ([shift={(1.75,-1.75)}]u9-1);
\coordinate (u10-2) at ([shift={(0,\hgap)}]u10-1);
\coordinate (u10-3) at ([shift={(0,\hgap)}]u10-2);
\coordinate (u11-1) at ([shift={(3.5,0)}]u9-1);
\coordinate (u11-2) at ([shift={(-\hgap,0)}]u11-1);
\coordinate (u11-3) at ([shift={(-\hgap,0)}]u11-2);
\coordinate (u12-1) at ([shift={(1.75,1.75)}]u9-1);
\coordinate (u12-2) at ([shift={(0,-\hgap)}]u12-1);
\coordinate (u12-3) at ([shift={(0,-\hgap)}]u12-2);

\filldraw[hedge] (u9-1) -- (u10-2) -- (u10-3) -- cycle;
\filldraw[hedge] (u10-1) -- (u11-2) -- (u11-3) -- cycle;
\filldraw[hedge] (u11-1) -- (u12-2) -- (u12-3) -- cycle;
\filldraw[hedge] (u12-1) -- (u9-2) -- (u9-3) -- cycle;

\node[mvert] at (u9-1){};
\node[mvert] at (u9-2) {};
\node[mvert] at (u9-3) {};
\node[mvert] at (u10-1) {};
\node[mvert] at (u10-2) {};
\node[mvert] at (u10-3) {};
\node[mvert] at (u11-1) {};
\node[mvert] at (u11-2) {};
\node[mvert] at (u11-3) {};
\node[mvert] at (u12-1) {};
\node[mvert] at (u12-2) {};
\node[mvert] at (u12-3) {};

\node at ([shift={(-0.6,0.5)}]u0-1) {$(v^*,1)$};
\node at ([shift={(0.2,0.5)}]u0-2) {$(v^*,2)$};
\node[below] at ([shift={(0.6,-0.75)}]u8-3) {\large $\overline{G} = (\overline{V},\overline{E})$};
\node at ([shift={(6.5,0.5)}]u0-1) {\large $q=3$};
\draw[rounded corners=4pt, draw=Sepia, thick] ([shift={(-1.3,1)}]u0-1) rectangle ([shift={(7.8,-4.5)}]u0-1);
\end{tikzpicture} \end{center}

\noindent We construct a $q$-dimensional matching $\overline{G} = (\overline{V}, \overline{E})$ on vertices $\overline{V} = V \times [q]$, such that for every edge $(u \to v) \in E$, we include the hyperedge $\set{(u,q), (v,1), \ldots, (v,q-1)}$ in $\overline{E}$. The designated vertices are $\overline{V}^* = \set{(v^*, 1), \ldots, (v^*,q-1)}$. Note that $|\overline{V}| \equiv 0 \mod{q}$ and $|\overline{V}^*| = q-1 \not\equiv 0 \mod{q}$. It is easy to see that a vertex $(v,i)$ is isolated in $\overline{G}$ if and only if $v$ is a source or a sink in $G$. This completes the reduction, since $\overline{V}$ is efficiently representable and indexable and the neighbors of any vertex in $\overline{V}$ are locally computable using black-box access to $C$ (see \autoref{rem:simplifications}).

\subsection{Oracle separations} \label{sec:oracle-sep}

Here we explain how $\PPA_q$ can be separated from other $\TFNP$ classes relative to oracles, as summarized in \autoref{fig:classes}. That is, for distinct primes $p, p'$, there exist oracles $O_1,\ldots,O_5$ such that
\[
(1)\ \PLS^{O_1} \nsubseteq \PPA_p^{O_1} \quad
(2)\ \PPA_p^{O_2} \nsubseteq \PPP^{O_2} \quad
(3)\ \PPA_{p'}^{O_3} \nsubseteq \PPA_p^{O_3}
\]
\[
(4)\ \PPADS^{O_4} \nsubseteq \PPA_p^{O_4} \qquad
(5)\ \bigcap_p\, \PPA_p^{O_5} \nsubseteq \PPAD^{O_5}
\]

The usual technique for proving such oracle separations is propositional proof complexity (together with standard diagonalization arguments)~\cite{beame98relative,bureshoppenheim04relativized,buss12propositional}. The main insight is that if a problem $S_1$ reduces to another problem $S_2$ in a black-box manner, then there are ``efficient proofs'' of the totality of $S_1$ starting from the totality of $S_2$. The discussion below assumes some familiarity with these techniques.

\paragraph{\boldmath $\PLS^{O_1} \nsubseteq \PPA_p^{O_1}$, $\PPA_p^{O_2} \nsubseteq \PPP^{O_2}$, $\PPA_{p'}^{O_3} \nsubseteq \PPA_p^{O_3}$.} Johnson \cite{johnson11reductions} showed all the above separations with respect to $\PMOD_p$. Since we showed $\PPA_p = \PMOD_p$ (\autoref{cor:ppa-pmod}), the same oracle separations hold for $\PPA_p$.

\paragraph{\boldmath $\PPADS^{O_4} \nsubseteq \PPA_p^{O_4}$.} G\"o\"os et al. \cite[\S4.3]{goos19adventures} building on \cite{beame98more} showed that the contradiction underlying the $\PPADS$-complete search problem $\textsc{Sink-of-Line}$ requires $\bbF_p$-Nullstellensatz refutations of high degree. This yields the oracle separation.

\paragraph{\boldmath $\bigcap_p\, \PPA_p^{O_5} \nsubseteq \PPAD^{O_5}$.}
For a fixed $k\geq 1$, consider the problem $S_k\coloneqq \revand_{i\in[k]}\, \lonely_{p_i}$ where $p_i$ are the primes. Buss et al.~\cite{buss01linear} showed that the principle underlying $S_i$ is incomparable with the principle underlying $\lonely_{p_{i+1}}$. This translates into an relativized separation $\bigcap_{i\in[k]}\PPA_{p_i} \nsubseteq \PPA_{p_{i+1}}$ which in particular implies $\bigcap_{i\in[k]}\PPA_{p_i}\nsubseteq\PPAD$. Finally, one can consider the problem $S\coloneqq S_{k(n)}$ where $k(n)$ is a slowly growing function of the input size $n$. This problem is in $\bigcap_p\PPA_p$ since for each fixed $p$ and for large enough input size, $S$ reduces to the $\PPA_p$-complete problem. On the other hand, the result of Buss et al.~\cite{buss01linear} is robust enough to handle a slowly growing $k(n)$; we omit the details.

\subsection{Closure under Turing reductions}\label{sec:turing-closure}

\autoref{thm:turing-closure} says that for any prime $p$, the class $\PPA_p$ is closed under Turing reductions. In contrast, Buss and Johnson showed that $\PPA_{p_1}~\amp~\PPA_{p_2}$, for distinct primes $p_1$ and $p_2$, is not closed under {\em black-box} Turing reductions \cite{buss12propositional,johnson11reductions}. In particular, they define the `{\boldmath $\otimes$}' operator as follows. For two total search problems $S_1$ and $S_2$, the problem $S_1$ {\boldmath $\otimes$} $S_2$ is defined as: Given $(x_0, x_1) \in \Sigma^* \times \Sigma^*$, find a solution to both $x_0$ (instance of $S_0$) and to $x_1$ (instance of $S_1$). Clearly the problem $\lonely_{p_1}$ {\boldmath $\otimes$} $\lonely_{p_2}$ can be solved with two queries to the oracle $\PPA_{p_1} \bamp \PPA_{p_2}$. However, Buss and Johnson \cite{buss12propositional,johnson11reductions} show that $\lonely_{p_1}$ {\boldmath $\otimes$} $\lonely_{p_2}$ cannot be solved with one oracle query to $\PPA_{p_1} \bamp \PPA_{p_2}$ under {\em black-box} reductions. In particular, this implies that $\PPA_q$ is not closed under {\em black-box} Turing reductions, when $q$ is not a prime power. We now prove \autoref{thm:turing-closure}, which is equivalent to the following.

\begin{theorem}
For any prime $p$ and total search problem $S$, if $S \preceq_T \lonely_p$, then $S \preceq_m \lonely_p$.
\end{theorem}
\begin{proof}
The key reason why this theorem holds for prime $p$ is \autoref{lem:restricted-lonely}: In a $\lonely_p$ instance, we can assume w.l.o.g. that there are exactly $p-1$ distinguished vertices.

On instance $x$ of the problem $S$, suppose the oracle algorithm sequentially makes at most $t = \poly(|x|)$ queries to $\lonely_p$ oracle. The $i$-th query consists of a tuple $(C_i, V_i^*)$ where $C_i$ encodes a $p$-dimensional matching graph $G_i = (V_i,E_i)$ and $V_i^* \subseteq V_i$ is the set of $p-1$ designated vertices, and let $y_i \in V_i$ be the solution returned by the $\lonely_p$ oracle. The query $(C_i, V_i^*)$ is computable in polynomial time, given $x$ and valid solutions to all previous queries. Finally, after receiving all answers the algorithm returns $L(x, y_1, \ldots, y_t)$ that is a valid solution for $x$ in $S$.\\[-2mm]

\noindent We make the following simplifying assumptions.\vspace{-2mm}
\begin{itemize}
\item Each hypergraph $G_i$ is on $p^n$ vertices, where $n = \poly(|x|)$ (thanks to instance extension property -- see \autoref{rem:simplifications}).
\item For any query the vertices $V_i^*$ are always isolated in $G_i$ (if some vertex in $V_i^*$ were to not be isolated, the algorithm could be modified to simply not make the query).
\item Exactly $t$ queries are made irrespective of the oracle answers.
\end{itemize}

\noindent We reduce $x$ to a single instance of $\lonely_p$ as follows.\vspace{-4mm}

\paragraph{Vertices.} The vertices of the $\lonely_p$ instance will be $V = [p]^n \cup [p]^{2n} \cup \cdots \cup [p]^{tn}$, which we interpret as $\overline{V} = V_1 \cup (V_1 \times V_2) \cup (V_1 \times V_2 \times V_3) \cup \cdots \cup (V_1 \times \cdots \times V_t)$. The designated vertices will be $\overline{V}^*\coloneqq V_1^*$. Note that $|\overline{V}^*| = |V_1^*| \not\equiv 0 \mod{p}$.\vspace{-3mm}

\paragraph{Edges.} We'll define the hyperedge for vertex $\overline{v} = (v_1, \ldots, v_k)$ for any $k \le t$. Let $j \le k$ be the last coordinate such that for all $i < j$, the vertex $v_i$ is a valid solution for the $\lonely_p$ instance $(C_i, V_i^*)$, which the algorithm creates on receiving $v_1, \ldots, v_{i-1}$ as answers to previous queries.
\begin{itemize}[leftmargin=2.3cm]
\item [Case $j < k$:] Let $u_1, \ldots, u_{p-1}$ be the neighbors of $v_k$ in a canonical trivial matching over $[p]^n$; e.g. $\set{[p] \times w : w \in [p]^{n-1}}$. The neighbors of $\overline{v}$ are $\set{(v_1, \ldots, v_{k-1}, u_i)}_i$.
\item [Case $j = k$:] We consider three cases, depending on whether $v_k$ is designated, non-isolated or isolated in the $\lonely_p$ instance $(C_k, V_k^*)$.
	\begin{itemize}[leftmargin=1.4cm]
	\item [Non-isolated $v_k$:] For $u_1, \ldots, u_{p-1}$ being the neighbors of $v_k$ in $G_k$, the neighbors of $\overline{v}$ are $\set{(v_1, \ldots, v_{k-1}, u_i)}_i$.
	\item [Isolated $v_k$:] Such a $v_k$ is a valid solution for $(C_k, V_k^*)$.
		\begin{itemize}[leftmargin=0.8cm]
		\item [If $k < t$:] the algorithm will have a next oracle query $(C_{k+1}, V_{k+1}^*)$. In this case, for $u_1, \ldots, u_{p-1}$ being the designated vertices in $V_{k+1}^*$, the neighbors of $\overline{v}$ are $\set{(v_1, \ldots, v_{k-1}, v_k, u_i)}_i$.
		\item [If $k = t$:] there are no more queries, and we leave $\overline{v}$ isolated.
		\end{itemize}
	\item [Designated $v_k$:] Let $u_1, \ldots, u_{p-2}$ be the other designated vertices in $V_k^*$. The neighbors of $\overline{v}$ are $\set{(v_1, \ldots, v_{k-1}, u_i)}_i \cup \set{(v_1, \ldots, v_{k-1})}$.
	\end{itemize}
\end{itemize}

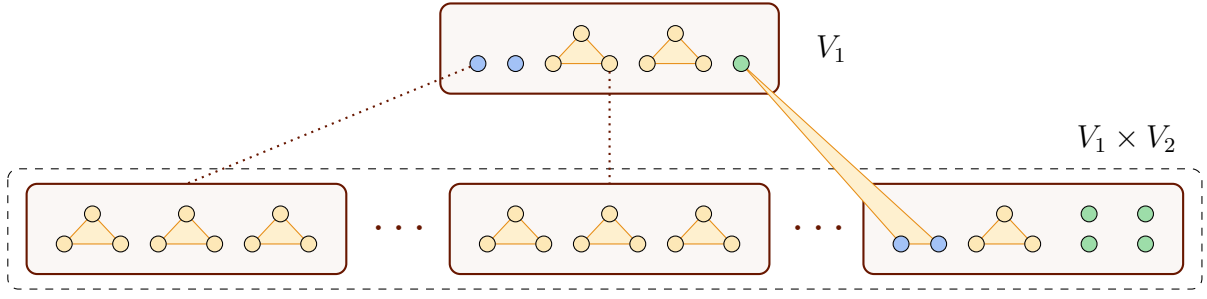
\begin{center}
	\begin{tikzpicture}
\tikzset{
	inner sep=0,outer sep=0,
	mvert/.style={circle, draw, minimum size=6pt, fill=Gyellow!30, outer sep=0pt},
	sol/.style={fill=Ggreen!50},
	designated/.style={fill=Gblue!50},
	hedge/.style={draw=myGold, fill=Gyellow!20},
	myrect/.style={rounded corners=4pt, draw=Sepia, thick, fill=Sepia!2},
	dedge/.style={dotted, Sepia, thick}}

\def\hgap{0.5}
\def\vgap{0.4}
\coordinate (u-1) at (0,0);
\coordinate (u-2) at ([shift={(\hgap,0)}]u-1);
\coordinate (u-3) at ([shift={(\hgap,0)}]u-2);
\coordinate (u-4) at ([shift={(\hgap*0.75,\vgap)}]u-3);
\coordinate (u-5) at ([shift={(\hgap*1.5,0)}]u-3);
\coordinate (u-6) at ([shift={(\hgap,0)}]u-5);
\coordinate (u-7) at ([shift={(\hgap*0.75,\vgap)}]u-6);
\coordinate (u-8) at ([shift={(\hgap*1.5,0)}]u-6);
\coordinate (u-9) at ([shift={(\hgap,0)}]u-8);

\coordinate (u5-1) at ([shift={(-3.25*\hgap,-6*\vgap)}]u-5);
\coordinate (u5-2) at ([shift={(\hgap*0.75,\vgap)}]u5-1);
\coordinate (u5-3) at ([shift={(\hgap*1.5,0)}]u5-1);
\coordinate (u5-4) at ([shift={(\hgap,0)}]u5-3);
\coordinate (u5-5) at ([shift={(\hgap*0.75,\vgap)}]u5-4);
\coordinate (u5-6) at ([shift={(\hgap*1.5,0)}]u5-4);
\coordinate (u5-7) at ([shift={(\hgap,0)}]u5-6);
\coordinate (u5-8) at ([shift={(\hgap*0.75,\vgap)}]u5-7);
\coordinate (u5-9) at ([shift={(\hgap*1.5,0)}]u5-7);

\coordinate (u1-1) at ([shift={(-11*\hgap,-6*\vgap)}]u-1);
\coordinate (u1-2) at ([shift={(\hgap*0.75,\vgap)}]u1-1);
\coordinate (u1-3) at ([shift={(\hgap*1.5,0)}]u1-1);
\coordinate (u1-4) at ([shift={(\hgap,0)}]u1-3);
\coordinate (u1-5) at ([shift={(\hgap*0.75,\vgap)}]u1-4);
\coordinate (u1-6) at ([shift={(\hgap*1.5,0)}]u1-4);
\coordinate (u1-7) at ([shift={(\hgap,0)}]u1-6);
\coordinate (u1-8) at ([shift={(\hgap*0.75,\vgap)}]u1-7);
\coordinate (u1-9) at ([shift={(\hgap*1.5,0)}]u1-7);

\coordinate (u9-1) at ([shift={(4.5*\hgap,0)}]u5-9);
\coordinate (u9-2) at ([shift={(\hgap,0)}]u9-1);
\coordinate (u9-3) at ([shift={(\hgap,0)}]u9-2);
\coordinate (u9-4) at ([shift={(\hgap*0.75,\vgap)}]u9-3);
\coordinate (u9-5) at ([shift={(\hgap*1.5,0)}]u9-3);
\coordinate (u9-6) at ([shift={(1.5*\hgap,0)}]u9-5);
\coordinate (u9-7) at ([shift={(0,\vgap)}]u9-6);
\coordinate (u9-8) at ([shift={(1.5*\hgap,\vgap)}]u9-6);
\coordinate (u9-9) at ([shift={(0,-\vgap)}]u9-8);

\draw[myrect] ([shift={(-0.5,2*\vgap)}]u-1) rectangle ([shift={(0.5,-1*\vgap)}]u-9);
\node at ([shift={(1.2,0.5*\vgap)}]u-9) {$V_1$};
\filldraw[hedge] (u-3) -- (u-4) -- (u-5) -- cycle;
\filldraw[hedge] (u-6) -- (u-7) -- (u-8) -- cycle;

\draw[dashed, rounded corners=4pt] ([shift={(-1.5*\hgap,2.5*\vgap)}]u1-1) rectangle ([shift={(1.5*\hgap,-1.5*\vgap)}]u9-9);
\node at ([shift={(-0.5*\hgap,3.5*\vgap)}]u9-9) {$V_1 \times V_2$};

\draw[myrect] ([shift={(-0.5,2*\vgap)}]u1-1) rectangle ([shift={(0.5,-1*\vgap)}]u1-9);
\filldraw[hedge] (u1-1) -- (u1-2) -- (u1-3) -- cycle;
\filldraw[hedge] (u1-4) -- (u1-5) -- (u1-6) -- cycle;
\filldraw[hedge] (u1-7) -- (u1-8) -- (u1-9) -- cycle;

\draw[myrect] ([shift={(-0.5,2*\vgap)}]u5-1) rectangle ([shift={(0.5,-1*\vgap)}]u5-9);
\filldraw[hedge] (u5-1) -- (u5-2) -- (u5-3) -- cycle;
\filldraw[hedge] (u5-4) -- (u5-5) -- (u5-6) -- cycle;
\filldraw[hedge] (u5-7) -- (u5-8) -- (u5-9) -- cycle;

\draw[myrect] ([shift={(-0.5,2*\vgap)}]u9-1) rectangle ([shift={(0.5,-1*\vgap)}]u9-9);
\filldraw[hedge] (u9-3) -- (u9-4) -- (u9-5) -- cycle;
\filldraw[hedge] (u-9) -- (u9-1) -- (u9-2) -- cycle;

\node[Sepia] at ([shift={(-2.25*\hgap,0.5*\vgap)}]u5-1) {\Large$\cdots$};
\node[Sepia] at ([shift={(2.4*\hgap,0.5*\vgap)}]u5-9) {\Large$\cdots$};

\node[mvert,designated] at (u-1) {} edge[dedge] ([shift={(0,\vgap)}]u1-5);
\node[mvert,designated] at (u-2) {};
\node[mvert] at (u-3) {};
\node[mvert] at (u-4) {};
\node[mvert] at (u-5) {} edge[dedge] ([shift={(0,\vgap)}]u5-5);
\node[mvert] at (u-6) {};
\node[mvert] at (u-7) {};
\node[mvert] at (u-8) {};
\node[mvert,sol] at (u-9) {};

\node[mvert] at (u1-1) {};
\node[mvert] at (u1-2) {};
\node[mvert] at (u1-3) {};
\node[mvert] at (u1-4) {};
\node[mvert] at (u1-5) {};
\node[mvert] at (u1-6) {};
\node[mvert] at (u1-7) {};
\node[mvert] at (u1-8) {};
\node[mvert] at (u1-9) {};

\node[mvert] at (u5-1) {};
\node[mvert] at (u5-2) {};
\node[mvert] at (u5-3) {};
\node[mvert] at (u5-4) {};
\node[mvert] at (u5-5) {};
\node[mvert] at (u5-6) {};
\node[mvert] at (u5-7) {};
\node[mvert] at (u5-8) {};
\node[mvert] at (u5-9) {};

\node[mvert,designated] at (u9-1) {};
\node[mvert,designated] at (u9-2) {};
\node[mvert] at (u9-3) {};
\node[mvert] at (u9-4) {};
\node[mvert] at (u9-5) {};
\node[mvert,sol] at (u9-6) {};
\node[mvert,sol] at (u9-7) {};
\node[mvert,sol] at (u9-8) {};
\node[mvert,sol] at (u9-9) {};
\end{tikzpicture} \end{center}

\noindent It is easy to see that our definition of edges are consistent and the only vertices which are isolated (apart from those in $\overline{V}^*$) are of the type $(y_1, \ldots, y_t)$ where each $y_i$ is a valid solution for the $\lonely_p$ instance $(C_i, V_i^*)$. Thus, given an isolated vertex $\overline{y}$, we can immediately infer a solution for $x$ as $L(x, y_1, \ldots, y_t)$. This completes the reduction since $\overline{V}$ is efficiently representable and indexable --- see \autoref{rem:simplifications}.
\end{proof}
 
\subsection*{Acknowledgements}
We thank Christos Papadimitriou, Robert Robere, Dmitry Sokolov and Noah Stephens-Davidowitz for helpful discussions. We also thank anonymous referees for valuable suggestions.

MG was supported by NSF grant CCF-1412958 (this work was done while MG was at IAS). PK was supported in parts by NSF Award numbers CCF-1733808 and IIS-1741137 and MIT-IBM Watson AI Lab and Research Collaboration Agreement No. W1771646 (this work was done while PK was a student at MIT). MZ is supported by a Google PhD Fellowship. KS is supported in part by NSF/BSF grant \#1350619, an MIT-IBM grant, and a DARPA Young Faculty Award, MIT Lincoln Laboratories and Analog Devices.

\appendix

\section{Appendix: Reductions Between Complete Problems}
\label{sec:equivalences}

In order to prove \autoref{thm:PPAqEquivalentProblems}, we introduce an additional problem that will serve as intermediate problem in our reductions.\\[1mm]

\mydef{leaf-prime}{$\leaf'_q$}
{Same as $\leaf_q$, but degrees are allowed to be larger (polynomially bounded).}
{$q$-uniform hypergraph $G = (V, E)$. Designated vertex $v^* \in V$.}
{$\triangleright$ $C:\bit^n \to (\bit^{nq})^k$, with $(\bit^{nq})^k$ interpreted as $k$ many $q$-subsets of $\bit^n$\\
	$\triangleright$ $v^* \in \bit^n$ (usually $0^n$)}
{$V\coloneqq\bit^n$.\\
	For distinct $v_1, \ldots, v_q$, edge $e\coloneqq\set{v_1, \ldots, v_q}\in E$ if $e \in C(v)$ for all $v \in e$}
{$v^*$ if $\deg(v) \equiv 0 \mod{q}$ and\\
	$v \ne v^*$ if $\deg(v) \not\equiv 0 \mod{q}$}

\begin{proof}[Proof of \autoref{thm:PPAqEquivalentProblems}]
We show the following inter-reducibilities: (1) $\leaf_q \asymp \leaf_q'$, (2) $\leaf_q' \asymp \bipartite_q$ and (3) $\leaf_q \asymp \lonely_q$.

\medskip

\parasc{(1a) $\boldsymbol{\leaf_q \reducible \leaf_q'}$} Each instance of $\leafq$ is trivially an instance of $\leaf'_q$.

\medskip

\parasc{(1b) $\boldsymbol{\leaf'_q \reducible \leafq}$.} We start with a $\leaf'_q$ instance $(C, v^*)$, where $C$ encode a $q$-uniform hypergraph $G = (V, E)$ with degree at most $k$. Let $t = \ceil{k/q}$. We construct a $\leaf_q$ instance encoding a hypergraph $\overline{G} = (\overline{V}, \overline{E})$ on vertex set $\overline{V}\coloneqq V \times [t]$, intuitively making $t$ copies of each vertex.

In order to locally compute hyperedges, we first fix a canonical algorithm that for any vertex $v$ and any edge $e \in E$ incident on $v$, assigns it a label $\ell_v(e) \in [t]$, with at most $q$ edges mapping to the same label --- e.g. sort all edges incident on $v$ in lexicographic order and bucket them sequentially in at most $t$ groups of at most $q$ each. Note that we can ensure that for any vertex $v$ at most one label gets mapped to by a non-zero, non-$q$ number of edges. Moreover, if $\deg(v) \equiv 0 \mod{q}$, then exactly $q$ or $0$ edges are assigned to any label.

We'll assume that $\deg(v^*) \not\equiv 0 \mod{q}$, as otherwise, a reduction wouldn't be necessary. We let $(v^*, \ell^*)$ be the designated vertex of the $\leaf_q$ instance, where $\ell^*$ is the unique label that gets mapped to by a non-zero, non-$q$ number of edges incident on $v^*$.

For any vertex $(v, i) \in \overline{V}$, we assign it at most $q$ edges as follows: For each edge $e = \set{v_1, \ldots, v_q}$ such that $\ell_v(e)=i$, the corresponding hyperedge of $(v,i)$ is ${(v_1, \ell_{v_1}(e)), \ldots, (v_q, \ell_{v_q}(e))}$. It is easy to see that the designated vertex $(v^*, \ell^*)$ indeed has non-zero, non-$q$ degree. Moreover, a vertex $\deg(v,i) \notin \set{0, q}$ in $\overline{G}$ only if $v$ has a non-multiple-of-$q$ degree in $G$. Thus, solutions to the $\leaf_q$ instance naturally maps to solutions to the original $\leaf_q'$ instance.

By \autoref{rem:simplifications}, this completes the reduction since the edges are locally computable with black-box access to $C$ and $\overline{V}$ is efficiently indexable.

\medskip

\parasc{(2a) $\boldsymbol{\leaf'_q \reducible \bipartiteq}$.}
We start with a $\leaf'_q$ instance $(C, v^*)$, where $C$ encode a $q$-uniform hypergraph $G = (V, E)$. We construct a $\bipartite_q$ instance encoding a graph $\overline{G} = (\overline{V} \cup \overline{U}, \overline{E})$ such that $\overline{V} = V$ and $\overline{U} = \binom{V}{q}$, i.e. all $q$-sized subsets of $V$. We include the edge $(v,e) \in \overline{E}$ if $e \in E$ is incident on $v$. The designated vertex for the $\bipartite_q$ instance is $v^*$ in $\overline{V}$.

Clearly, all vertices $e \in \overline{U}$ have degree either $q$ or $0$. For any $v \in \overline{V}$, the degree of $v$ in $\overline{G}$ is same as its degree in $G$. Thus, any solution to the $\bipartite_q$ instance immediately gives a solution to the original $\leaf'_q$ instance. By \autoref{rem:simplifications}, this completes the reduction since the edges are locally computable with black-box access to $C$ and $\overline{V}$ and $\overline{U}$ are efficiently indexable (cf. \cite[\S2.3]{kreher98combinatorial} for efficiently indexing $\overline{U}$).

\medskip

\parasc{(2b) $\boldsymbol{\bipartiteq \reducible \leaf'_q}$.} We start with a $\bipartite_q$ instance $(C, v^*)$ encoding a bipartite graph $\calG = (V \cup U,E)$ with maximum degree of any vertex being at most $k$. We construct a $\leaf'_q$ instance encoding a hypergraph $\overline{G} = (\overline{V}, \overline{E})$ such that $\overline{V} = V$ with designated vertex $v^*$.

First, we fix a canonical algorithm that for any vertex $u \in U$ with $\deg_{G}(u) \equiv 0 \mod{q}$ produces a partition of it's neighbors with $q$ vertices of $V$ in each part. Now, the set of $q$-uniform hyperedges incident on any vertex $v \in \overline{V}$ in $\overline{E}$ can be obtained as: for all neighbors $u$ of $v$, with $\deg_G(u) \equiv 0 \mod{q}$, we include a hyperedge consisting of all vertices in the same partition as $v$ among the neighbors of $u$ (we ignore neighbors $u$ with $\deg(u) \not\equiv 0 \mod{q}$).

Observe that $\deg_{\overline{G}}(v) \le \deg_{G}(v)$ and equality holds if and only if all neighbors of $v$ in $G$ have degree $\equiv 0 \mod{q}$. Hence for any $v \in \overline{V}$, if $\deg_{\overline{G}}(v) \ne \deg_G(v) \mod{q}$, then there exists a neighbor $u \in U$ of $v$ in $G$ such that $\deg(u) \not\equiv 0 \mod{q}$. Thus, if $v = v^*$ and $\deg_{\overline{G}}(v^*) \equiv 0 \mod{q}$, then either $\deg_G(v) \equiv 0 \mod{q}$ or we can find a neighbor $u$ of $v$ in $G$ with $\deg(u) \not\equiv 0 \mod{q}$. Similarly if for some $v \ne v^*$, we have $\deg_{\overline{G}}(v^*) \not\equiv 0 \mod{q}$, then either $\deg_G(v) \not\equiv 0 \mod{q}$ or we can find a neighbor $u$ of $v$ in $G$ with $\deg(u) \not\equiv 0 \mod{q}$. Thus, any solution to the $\leaf'_q$ instance gives us a solution to the original $\bipartite_q$ instance. This completes the reduction since $\overline{V} = \bit^n$ and the edges are locally computable with black-box access to $C$.

\medskip

\parasc{(3a) $\boldsymbol{\leaf_q \reducible \lonely_q}$.} We start with a $\leaf_q$ instance $(C, v^*)$, where $C$ encode a $q$-uniform hypergraph $G = (V, E)$ with degree at most $q$. If $\deg_G(v^*) = q$ or $0$, then we don't need any further reduction. Else, we construct a $\lonely_q$ instance encoding a $q$-dimensional matching $\overline{G} = (\overline{V}, \overline{E})$ on vertex set $\overline{V} = V \times [q]$. The designated vertices will be $V^* = \set{(v,q-i) : 1 \le i \le q - \deg(v^*)}$. Note that, $|V^*| = q - \deg_G(v^*)$ and hence $1 \le |V^*| \le q-1$.

In order to locally compute hyperedges, we first fix a canonical algorithm that for any vertex $v$ and any edge $e \in E$ incident on $v$, assigns it a unique label $\ell_v(e) \in [q]$ --- e.g. sort all edges incident on $v$ in lexicographic order and label them sequentially in $[q]$. In fact, we can ensure that an edge incident on $v$ get labeled within $\set{1, \ldots, \deg_G(v)}$.

For any vertex $(v, i) \in \overline{V}$, we assign it at most one hyperedge as follows:
\begin{itemize}[label=$\triangleright$]
\item If $\deg_G(v) = 0$, we include the hyperedge $\set{(v,i) : i \in [q]}$.
\item Else if $\deg_G(v) \ge i$, then for edge $e = \set{v_1, \ldots, v_q}$ incident on $v$ such that $\ell_v(e)=i$, the corresponding hyperedge of $(v,i)$ is ${(v_1, \ell_{v_1}(e)), \ldots, (v_q, \ell_{v_q}(e))}$.
\item Else if $0 < \deg_G(v) < i$, we leave it isolated.
\end{itemize}

It is easy to see that our definition of hyperedges is consistent and that the designated vertices $V^*$ are indeed isolated. Moreover, a vertex $(v,i)$ is isolated in $\overline{G}$ only if $1 \le \deg_G(v)\le q-1$. Thus, solutions to the $\leaf_q$ instance naturally maps to solutions to the original $\leaf_q'$ instance.

By \autoref{rem:simplifications}, this completes the reduction since the edges are locally computable with black-box access to $C$ and $\overline{V}$ is efficiently indexable.

\medskip

\parasc{(3b) $\boldsymbol{\lonely_q \reducible \leaf_q}$.} We start with a $\lonely_q$ instance $(C, V^*)$, where $C$ encode a $q$-dimensional matching $G = (V, E)$. We construct a $\leaf_q$ instance encoding a $q$-uniform hypergraph $\overline{G} = (\overline{V}, \overline{E})$ on vertex set $\overline{V}$ that will be specified shortly. We describe the hyperedges in $\overline{G}$ and it'll be clear how to compute the hyperedges for any vertex locally with just black-box access to $C$.

We start with $\overline{V} = V$. Our goal is to transform all vertices of degree $1$ to degree $q$, while ensuring that vertices of degree $0$ are mapped to vertices of degree not a multiple of $q$. Towards this goal we let $\overline{E}$ to be set of edges in $E$ in addition to $q-1$ canonical $q$-dimensional matchings over $V$. For example, for a vertex $v \coloneqq (x_1, \ldots, x_n) \in V = [q]^n$, the corresponding edges in $\overline{E}$ include an edge in $E$ (if any) and edges of the type $e_i = \set{(x_1, \dots, x_{i-1}, j, x_{i+1}, \dots, x_n) : j \in [q]}$ for $i \in [q-1]$ (note, this requires us to assume $n \geq q-1$). Adding the $q-1$ matchings increases the degree of each vertex by $q-1$. Therefore, vertices with initial degree $1$ now have degree $q$ and vertices with initial degree $0$ now have degree $q-1$. However, a couple of issues remain in order to complete the reduction, which we handle next.

{\em Multiplicities.} An edge $e \in E$ might have gotten added twice, if it belonged to one of the canonical matchings. To avoid this issue altogether, instead of adding edges directly on $V$, we augment $\overline{V}$ to become $\overline{V} \coloneqq V \cup \inparen{\binom{V}{q} \times [q-1]}$, i.e. in addition to $V$, we have $q-1$ vertices for every potential hyperedge of $G$. For any edge $e\coloneqq\set{v_1, \ldots, v_q} \in E$, instead of adding it directly in $\overline{G}$, we add hyperedge $\set{v, (e,1), (e,2), \ldots, (e,q-1)}$ for each $v \in e$. Note that, all vertices $(e,i) \in \binom{V}{q} \times [q-1]$ have degree $q$ if $e \in E$ and degree $0$ if $e \notin E$, so they are non-solutions for the $\leaf_q$ instance. For vertices in $V$, we still have as before that vertices with initial degree $1$ now have degree $q$ and vertices with initial degree $0$ now have degree $q-1$.

{\em Designated vertex.} In a $\leaf_q$ instance, we need to specify a single designated vertex $v^* \in \overline{V}$. If the $\lonely_q$ instance had a single designated vertex then we would be done. However, in general it is not possible to assume this (for non-prime $q$). Nevertheless, we provide a way to get around this. We augment $\overline{V}$ with $t = (q-1)(q-k) + 1$ additional vertices to become $\overline{V} \coloneqq V \cup  \inparen{\binom{V}{q} \times [q-1]} \cup \set{w_{i,j} : i \in [q-k],\ j \in [q-1]} \cup \set{v^*}$, where $v^*$ will eventually be the single designated vertex for the $\leaf_q$ instance.

Let $V^* = \set{u_1, \ldots, u_k} \subseteq V$ be the set of designated vertices in the $\lonely_q$ instance (note $1 \le k < q$). So far, note that $\deg_{\overline{G}}(u_i) = q-1$. The only new hyperedges we add will be among $u_i$'s, $w_{i,j}$'s and $v^*$, in such a way that $\deg_{\overline{G}}(u_i)$ will become $q$, the degree of all $w_{i,j}$'s will also be $q$ and degree of $v^*$ will be $q-k$.
\begin{itemize}[label=$\triangleright$, itemsep=3pt]
\item For each $u \in V^*$, include $\set{u, w_{1,1}, \ldots, w_{1,q-1}}$. So far, $\deg_{\overline{G}}(u) = q$ and $\deg_{\overline{G}}(w_{1,j}) = k$.
\item For each $j \in [q-1]$ and each $i \in \set{2, \ldots, q-k}$, include $\set{w_{1,j}, w_{i,1}, \ldots, w_{i,q-1}}$.\\
So far, $\deg_{\overline{G}}(w_{i,j}) = q-1$ for all $(i,j) \in [q-k] \times [q-1]$.
\item Finally, for each $(i,j) \in [q-k] \times [q-1]$, include $\set{v^*, w_{i,1}, \ldots, w_{i,q-1}}$.\\
Now, $\deg_{\overline{G}}(w_{i,j}) = q$ for all $(i,j) \in [q-k] \times [q-1]$ and $\deg_{\overline{G}}(v^*) = q-k$.
\end{itemize}

Thus, we have finally reduced to a $\leaf_q$ instance encoding the graph $\overline{G} = (\overline{V}, \overline{E})$ with $\overline{V} \coloneqq V \cup  \inparen{\binom{V}{q} \times [q-1]} \cup \set{w_{i,j} : i \in [q-k],\ j \in [q-1]} \cup \set{v^*}$. By \autoref{rem:simplifications}, this completes the reduction, since $\overline{V}$ is efficiently indexable (again, see \cite{kreher98combinatorial} for a reference on indexing $\binom{V}{q}$) and the edges are locally computable using black-box access to $C$.
\end{proof}

\subsection{\boldmath Completeness of Succinct Bipartite}\label{sec:sucBipartite}

We introduce a new intermediate problem to show $\PPA_p$--completeness of $\sucBipartite_p$.\\[1mm]

\mydef{two-matchings}{$\twoMatchings_p$}
{Two $p$-dimensional matchings over a common vertex set, with a vertex in exactly one of the matchings, has another such vertex.}
{Two $p$-dimensional matchings $G_0 = (V, E_0)$, $G_1 = (V, E_1)$. Designated vertex $v^* \in V$.}
{$\triangleright$ $C_0:\bit^n \to (\bit^n)^p$ and $C_1:\bit^n \to (\bit^n)^p$\\
	$\triangleright$ $v^* \in \bit^n$}
{$V\coloneqq\bit^n$. For $b \in \bit$, $E_b \coloneqq \set{e : C_b(v) = e \text{ for all } v \in e}$	}
{$v^*$ if $\deg_{G_0}(v^*) \ne 1$ or $\deg_{G_1}(v^*) \ne 0$ and\\
	$v \ne v^*$ if $\deg_{G_0}(v^*) \ne \deg_{G_1}(v^*)$}

\noindent Observe that in the case of $p=2$, $\twoMatchings_p$ can be readily seen as equivalent to $\leaf_2$.

\begin{theorem}\label{thm:succ-bip-complete}
	For any prime $p$, $\sucBipartite_p$ and $\twoMatchings_p$ are $\PPA_p$--complete.
\end{theorem}
\begin{proof}
We show that $\bipartite_p \reducible \sucBipartite_p \reducible \twoMatchings_p \reducible \lonely_p$.\medskip

\parasc{$\boldsymbol{\bipartitep \reducible \sucBipartitep}$.} Since $p$ is a prime, we can assume that the designated vertex $v^{*}$ has degree $1 \mod{p}$ (similar to \autoref{lem:restricted-lonely}). Since the
number of neighbors in a $\bipartitep$ instance are polynomial, we can
check if an edge exists and canonically group them efficiently for all vertices
with degree being a multiple of $p$. The designated edge $e^{*}$ is the unique ungrouped edge incident on $v^*$. Thus, valid solution edges to $\sucBipartite_p$ must have at least one endpoint which is a solution to the original $\bipartite_p$ instance.\medskip

\parasc{$\boldsymbol{\sucBipartitep \reducible \twoMatchingsp}$.}
We reduce to a $\twoMatchings_p$ instance encoding two $p$-dimensional matchings $\overline{G}_0 = (\overline{V}, \overline{E}_0)$ and $\overline{G}_1 = (\overline{V}, \overline{E}_1)$, over the vertex set $\overline{V} = V \times U \times [p-1]$, that is, all possible edges producible in the $\sucBipartite_p$ instance. The designated vertex $v^*$ is the designated edge $e^*$ in the $\sucBipartite_p$ instance.

For any edges $e_1, \ldots, e_p$, which are grouped by $\phi_V$ pivoted at some $v \in V$, we include the hyperedge $\set{e_1, \ldots, e_p}$ in $\overline{E}_0$. Similarly, for any edges $e_1, \ldots, e_p$, which are grouped by $\phi_U$ pivoted at some $u \in U$, we include the hyperedge $\set{e_1, \ldots, e_p}$ in $\overline{E}_1$. It is easy to see that points in exactly one of the two matchings $\overline{G}_0$ or $\overline{G}_1$ correspond to edges of the $\sucBipartite_p$ instance that are not grouped at exactly one end. Thus, we can derive a solution to $\sucBipartite_p$ from a solution to $\twoMatchings_p$. (Remark: while edges which are not grouped at either end are solutions to $\sucBipartite_p$, they do not correspond to a solution in the $\twoMatchings_p$ instance.) \medskip

\parasc{$\twoMatchings_p \reducible \lonely_p$.} Given a $\twoMatchings_p$ instance encoding two $p$-dimensional matchings $G_0 = (V, E_0)$ and $G_1 = (V, E_1)$, we reduce to an instance of $\lonely_p$ encoding a $p$-dimensional matching $\overline{G} = (\overline{V}, \overline{E})$ such that $\overline{V} = V \times [p]$. The designated vertex for the $\lonely_p$ instance is $(v^*,p)$.

For any hyperedge $\set{v_1, \ldots, v_p}$ in $E_0$, we include the hyperedge $\set{(v_1, i), (v_2, i), \ldots, (v_p,i)}$ in $\overline{G}$ for each $i \in \set{1, \ldots, p-1}$. Similarly, for any hyperedge $\set{v_1, \ldots, v_p}$ in $E_1$, we include the hyperedge $\set{(v_1, p), (v_2, p), \ldots, (v_p,p)}$ in $\overline{G}$. If $v \in V$ is isolated in both $G_0$ and $G_1$, then we include the hyperedge $\set{v} \times [p]$.

Observe that, $(v^*, p)$ is isolated by design. A vertex $(v,i)$, for $i < p$ is isolated only if $\deg_{G_0}(v) = 0$ and $\deg(G_1) = 1$. Similarly, the vertex $(v,p)$ is isolated only if $\deg_{G_0}(v) = 1$ and $\deg(G_1) = 0$. Thus, isolated vertices in the $\lonely_p$ instance correspond to solutions of the $\twoMatchings_p$ instance.
\end{proof}

\subsection{Equivalence with $\PMOD_p$}\label{sec:pmod-equivalence}

\begin{proof}[Proof of \autoref{lem:restricted-lonely}]
Consider any prime $p$. Consider a $\lonely_p$ instance $(C, V^*)$, where $C$ encodes a $p$-dimensional matching $G = (V,E)$ and $|V^*| = \ell$. We wish to reduce to an instance of $\lonely_p^k$, where the number of designated vertices is exactly $k$. First, we'll assume that all vertices in $V^*$ are indeed isolated in $G$, otherwise, no reduction would be necessary. The key reason why this lemma holds for primes (and not for composites) is because $\ell$ has a multiplicative inverse modulo $p$. In particular, let $t \equiv \ell^{-1} k \mod{p}$.

We construct a $\lonely_p^k$ instance encoding the $p$-dimensional matching $\overline{G} = (\overline{V}, \overline{E})$ over $\overline{V} = V \times [t]$. We let $\overline{V}^*$ to be the lexicographically first $k$ vertices in $V^* \times [t]$. Note that $|V^* \times [t]| = t.\ell \equiv k \mod{p}$. Thus, we partition the remaining vertices of $V^* \times [t]$ into $p$-uniform hyperedges. For any vertex $v \in V \smallsetminus V^*$, with neighbors $v_1, \ldots, v_{p-1}$ in $G$, the neighbors of $(v,i)$ in $\overline{G}$ are $(v_1, i), \ldots, (v_{p-1}, i)$ for any $i \in [t]$. Thus, a vertex $(v,i)$ is isolated only if it is in $\overline{V}^*$ or $v$ is isolated in $G$. This completes the reduction since $\overline{V}$ is efficiently indexable -- see \autoref{rem:simplifications}.
\end{proof}

\begin{proof}[Proof of \autoref{cor:ppa-pmod}]
It is easy to see that $\modn_q \le \lonely_q$ with number of designated vertices being $k \equiv -2^n \mod{q}$, since $\bit^n$ is efficiently indexable (\autoref{rem:simplifications}). Conversely, using \autoref{lem:restricted-lonely}, we can reduce a $\lonely_q$ instance to a $\modn_q$ instance as follows: Let the $\lonely_q$ instance encode a $q$-dimensional matching over $[q]^n$ with $k$ designated vertices. If any of the designated vertices are not isolated, no further reduction is necessary. Otherwise, we can embed the non-designated vertices of $G$ into the first $q^n - k$ vertices of $\bit^N$ for a choice of $N$ satisfying $2^N > q^n$ and $2^N \equiv -k \mod{q}$. Such an $N$ is guaranteed to exist (and can be efficiently found) when $q$ is a prime. Since $2^N - q^n + k \equiv 0 \mod{q}$, we can partition the remaining vertices into $q$-uniform hyperedges, and thus, solutions to the $\modn_q$ instance readily map to solutions of the original $\lonely'_q$ instance.
\end{proof}

\section{Appendix: Proof of \autoref{thm:simplificationPPAp}}
\label{sec:simplificationProof}

\begin{prevproof}{Theorem}{thm:simplificationPPAp}
    From \autoref{thm:combinedChevalleyPPApCompleteness}, it suffices
  to show that $\cChevalleyp \preceq \sucBipartitep[\ACzp]$.
  Additionally from the proof of
  \autoref{thm:combinedChevalleyPPApCompleteness} we can assume without loss
  of generality that the system of polynomials
  $\vecf = (\vecg, \vech)$ of the $\cChevalleyp$ instance has the
  following properties.
  \begin{Enumerate}
    \item[a.] Each polynomial $f_i$ has degree at most 2.
    \item[b.] Each polynomial $f_i$ has at most 3 monomials.
    \item[c.] Each polynomial $f_i$ has at most 3 variables.
  \end{Enumerate}
  \noindent Hence, we can compute each of the polynomials
  $g_i^{p - 1}$ explicitly as a sum of monomials. The degree
  of this polynomial is $O(p)$ and the number of monomials is at most
  $3^p$. Observe that since $p$ is a constant, $3^p$ is also a
  constant.

    Now we follow the proof of \autoref{lem:combinedChevalleyPPAp}
  that reduces an instance of $\cChevalleyp$ to an instance of
  $\sucBipartitep$. Following this proof there are two circuits that we
  need to replace with formulas in $\ACzp$ to reduce to $\sucBipartitep$.
  The first circuit is the edge counting circuit $\calC$ and the second is the
  grouping function $\phi$. We remind that the bipartite graph $G(U, V)$
  of the $\sucBipartitep$ instance has
  two parts $U$, $V$, where $U$ is the set of all possible assignments,
  i.e. $\F_p^n$, and $V = V_1 \cup V_2$, where $V_1$ in
  the set of all monomials of the polynomial
  $F = \prod_{i = 1}^m (1 - g_{i}^{p - 1})$  and $V_2$ is the set of all
  $p$-tuples of assignments, i.e. $\p{\F_p^n}^p$.
  \medskip

  \paragr{From Edge Counting Circuit To Edge Counting Formula.} As
  described in the proof of \autoref{lem:combinedChevalleyPPAp} the
  edge counting circuit takes as input a vertex $u \in U$ and a
  vertex $v \in V$ and outputs the multiplicity of the edge $\{u, v\}$ in
  $G$. Hence, the edge counting formula $\calC$, that we want to implement,
  takes as input a tuple $(\vecx, s, \veca, \vecy)$. The vector $\vecx$
  corresponds to the assignment in $U$. The vector $\veca$ corresponds
  to the description of a monomial of $F$, as the product
  $\prod_{i = 1}^m t'_{i a_i}$ where $t'_{i a_i}$ is the $a_i$-th
  monomial of the polynomial $1 - g_i^{p - 1}$. The vector
  $\vecy = (\vecy_1, \vecy_2, \dots, \vecy_p)$ and corresponds to a
  $p$-tuple in $V_2$. Finally, $s$ is a selector number to distinguish
  between $v \in V_1$ and $v \in V_2$, namely if $s = 1$, we have
  $v \in V_1$ and if $s = 0$, we have that $v \in V_2$. So, the edge
  counting formula can be written as follows
  \begin{equation} \label{eq:selectorImplementation}
    \calC(\vecx, s, \veca, \vecy) = \p{\prod_{i \in \F_p, i \neq 1} (s - i)} \calC_1(\vecx, \veca, \vecy) + \p{\prod_{i \in \F_p, i \neq 0} (s - i)} \calC_2(\vecx, \veca, \vecy).
  \end{equation}
  \noindent This way we can define the edge counting formula $\calC_1$
  for when $v \in V_1$ and the edge counting formula $\calC_2$ for when
  $v \in V_2$ separately and combine them by using at most two
  additional layers in the arithmetic formula. Now,
  $\calC_1(\vecx, \vecy, \veca) = \chara(\vecy = \vec{0}) \cdot \prod_{i = 1}^m \calQ_i(\vecx, a_i)$
  where $\calQ_i(\vecx, a_i)$ is the formula to compute the
  value $t_{i, a_i}(\vecx)$. Observe that the factor $\chara(\vecy = \vec{0})$ can be easily
  computed and is necessary since $\calC_1$ should consider only neighbors between
  $\vecx$ and monomials in $V_1$. Hence, if $\vecy$ is not equal to $\vec{0}$,
  $\calC_1$ should return $0$.
  As we already explained the number of monomials of
  $1 - g_i^{p - 1}$ is constant, and hence the formula
  $\calQ_i(\vecx, a_i)$ can be easily implemented in constant depth
  using a selector between all different monomials similarly to
  Equation~\eqref{eq:selectorImplementation}. Hence,
  $\calC_1$ is implemented in constant depth.

    The formula $\calC_2$ has a factor $\chara(\veca = 0)$ to
  ensure only neighbors in $V_2$ have non-zero outputs. The main
  challenge in the description of $\calC_2$ is that every distinct
  $p$-tuple $\vecy$ has $p!$ equivalent
  representations, but the modulo $p$ argument of
  \autoref{lem:combinedChevalleyPPAp} applies only when edges appear to precisely one of
  the equivalent copies of the $p$-tuple. Thus, we let $\calC_2$ add edges only to
  the lexicographically ordered version of $\vecy$. It is a simple exercise
  to see that sorting of $p!$ numbers, when $p$ is constant, is possible
  in constant depth. We leave this folklore observation as an exercise to
  the reader. Once we make sure that $\vecy$ is lexicographically
  sorted, we compute a sorted representation of the set
  $\Sigma_{\vecx} = \{\vecx, \sigma(\vecx), \dots, \sigma^{p - 1}(\vecx)\}$, where $\sigma$ is the permutation in the input of the
  $\cChevalleyp$ problem. Then, we can easily check whether the $p$-tuple
  represented by $\vecy$ is the same as the sorted $p$-tuple $\Sigma_{\vecx}$.
              Finally, we observe that
  edges between $\vecx$ and $\Sigma_{\vecx}$ are only used when
  $\vecx \in \calV_{\vecg} \cap \overline{\calV}_{\vech}$ which
  again can be checked with constant depth formulas. If these
  checks pass, then $\calC_2$ outputs $p - 1$, otherwise it outputs $0$.
  \medskip

  \paragraph{From Grouping Circuit to Grouping Formula.}
  For this step we use selectors similarly to
   Equation~\eqref{eq:selectorImplementation} and sorting
  as in the description of $\calC_2$. We consider two
  different cases for the grouping formula $\phi$. When the first argument
  is in $U$, i.e. grouping
  with respect to an assignment, we call the formula $\psi$ and when the
  first argument is in $V$, i.e. grouping with respect to
  monomials/$p$-tuples, we call the formula $\chi$. Then, $\phi$ selects between
  $\psi$ and $\chi$ using a selector. This adds at most two layers to
  $\phi$.
\begin{description}
\item[Grouping formula for $\vecx \boldsymbol{\in U}$.]
      First, we describe $\psi$ with inputs $\vecx \in U$,
      $(s, \veca, \vecy) \in V$ and $r$ be the copy of the input edge. We
      have two cases with respect to whether $s = 1$ or $s = 0$. Let $\psi^1$ be the
      formula for the first case and $\psi^2$ be the formula for the second case. For
      the case $s = 1$, we need again to consider two cases:
      (i) $\vecx \in \overline{\calV}_{\vecg}$ and (ii)
      $\vecx \in \calV_{\vecg}$. For case (i) we describe the
      formula $\psi^1_1$ and for case (ii) we define the formula
      $\psi^1_2$. It is easy to see that computing
      $\chara(\vecx \in \calV_{\vecg})$ can be done using a depth 3
      formula since $\vecg$ is given in an explicit form. Hence, once
      again, we can combine $\psi^1_1$ and $\psi^1_2$ using a selectors.

      \paragraph{\boldmath Case $s = 1$, $\vecx \in \overline{\calV}_{\vecg}$.} The formula
      $\psi^1_1$ first computes
      $i^{\star} = \min\limits_{i:1 - g^{p - 1}_{i}(\vecx) = 0} i$.
      This is doable in constant depth, since we can compute in parallel the
      value $\chara(1 - g^{p - 1}_{i}(\vecx) = 0)$  for all $i \in \bracks{m_1} $
       and then in an
      extra layer compute for every $i$ whether
      $1 - g^{p - 1}_{i}(\vecx) = 0$ and
      $1 - g^{p - 1}_{j}(\vecx) \neq 0$ for all $j < i$, which requires just one
      multiplication gate per $i$.

      Next, we define a formula $\psi^1_{1 i}$ for all $i$ and we use a selector to
      output $\psi^1_{1 i^*}$. In $\psi^1_{1 i}$, we first compute the
      value $C_i(\vecx) = \prod_{j \neq i} t_{j, a_j}(\vecx)$. The
      output of $\psi^1_{1 i}$ is a $p$-tuple, where each of the $p$
      parts differs only on the coordinate $a_i$ of $\veca$, which corresponds to a
      monomial of $1 - g_i^{p-1}$, and the value $r$.
      We need to determine $p$ different values for the tuple $(a_i, r)$ where
      $a_i \in [3^p]$,  $r \in \Z_p$. These values only depend on the
      evaluation of the polynomial $g_i$ on the input $\vecx$, on the value $a_i$
      and on the value $r$.

      Because of the properties of the input system of
      polynomials $\vecf$, each polynomial $g_i$ depends only on three
      variables in $\Z_p$, let these variables be $x_1, x_2, x_3$ for simplicity.
      Then, for every $i$ the grouping function that we want to implement is a
      function with input domain $\Z_p^3 \times [3^p] \times \Z_p$ and
      output domain $\Z_p^2$. The truth-table of this function has size
      that depends only on $p$ and therefore we can explicitly implement this
      function using its truth-table in constant depth. This finishes the
      construction of $\psi^1_{1 i}$.

      \paragraph{\boldmath Case $s = 1$, $\vecx \in \calV_{\vecg}$.} We remind that
      $\veca = \vec{0}$ corresponds to the constant monomial $1$ of the
      polynomial $F$. If  $\veca \neq \vec{0}$, this case is
      similar to the previous, except that we use the polynomials
      $\vecg_{ i}^{p - 1}$ instead of $1 - \vecg_{ i}^{p - 1}$, see also
      the proof of \autoref{lem:combinedChevalleyPPAp}. If
      $\veca = \vec{0}$, $\psi^1_2$ outputs the input edge
      $(1, \veca, \vec{0}, 1)$ and $p - 1$ edges of the form
      $(0, \vec{0}, \vecy, t), t \in \bracks{p-1}$
      where $\vecy$ is the lexicographically ordered set $\Sigma_{\vecx}$.

      \paragraph{\boldmath Case $s = 0$.} In this case, the formula $\psi^2$ checks whether
      the vector $\vecy$ is in lexicographic order as described in
      the edge counting formula $\calC$ and $\veca = \vec{0}$.
      It also checks if $\vecx \in \calV_{\vecf_1} \cap \overline{\calV}_{\vecf_2}$
      as described before. If any of these checks fails, the output is $\vec{0}$.
      Otherwise, if $\vecy = \Sigma_{\vecx}$, then we output $p-1$ copies of the edge
      $( 0, \vec{0}, \vecy, t), t \in \bracks{p-1}$, that connects $\vecx$  with
      $\vecy$, and the edge
      $(1, \vec{0}, \vec{0},1)$, that connects $\vecx$ with the constant term of
      $ F$.
      \medskip

  \item[Grouping formula for vertices in $\boldsymbol{V}$.]
        We describe the grouping formula $\chi$ when the first argument belongs to
        $V$, i.e. the grouping with respect to monomials or $p$-tuples. The input
        again is a triple $(s, \veca, \vecy)$ representing  a vertex in $V$, a
        vertex $\vecx \in U$ and a number $r \in \Z_p$ that denotes the index of the
        edge that we want to group, among its possible multiple copies. Again we
        have two cases, $s = 1$ and $s = 0$, which correspond to
        the formulas $\chi^1$ and $\chi^2$ respectively. In each case, we have
        to check that one of $\veca$, $\vecy$ is equal to $\vec{0}$, which is done
        similarly to the previous formulas.

        \paragraph{\boldmath Case $s = 1$.} In this case, the input is a monomial
        $t_{\veca}(\vecx) = \prod_{i = 1}^{m_1} t_{i, a_i}(\vecx)$ and
        we have to find a variable that appears with degree less than $p - 1$.
        We first construct a formula $\chi^1_j$ that computes $z^k$, where $k$ is
        the degree of $x_j$ in $t_{\veca}(\vecx)$.
        This can be done with a constant size
        formula that for a given index $j$ multiplies the powers of $x_j$
        in the monomials of $1 - g_i^{p-1}$ appearing in $t$.

        Now, we compute all values $\chi^1_j(1)$, $\dots$,
        $\chi^1_j(p - 1)$ and we check in parallel if at least one of them is
        different from $1$. If this is the case, then the degree of
        $x_j$ in $t(\vecx)$ is less than $p - 1$. Hence, we have computed the
        formula $\bar{\chi}^1_j(\veca) = \chara(\text{degree of $x_j$ in
        $t_{\veca}$} \neq p - 1)$. We can find the smallest index $j^*$
        such that $\bar{\chi}^1_j(\veca) = 1$ using the same construction
        as in $\psi^1$. So, we can construct a formula for each $j$ that is equal to
        $1$ if and only if $j = j^*$ is the smallest index such that $x_{j^*}$ has
        degree less than $p - 1$ in $t_{\veca}$.
        Finally, we use a selector to find the value
        $C_{j^*}(\vecx) = x_{j^*}^{-k} t(\vecx)$, by computing $C_{j}(\vecx)$ for
        all $j$. This is done through the product of all variables that appear in
        $t_{\veca}(\vecx)$ excluding $x_{j}$.

        It is left to implement a formula that takes as input the value
        $C_{j^*}(\vecx) \in \Z_p$, the value of $r \in \Z_p$ and the values
        $\chi^1_{j^*}(0)$, $\chi^1_{j^*}(1)$, $\dots$, $\chi^1_{j^*}(p - 1)$ all in
        $\Z_p$ and outputs a group of $p$ values in $\Z_p^2$, which corresponds to
        the values of $x_j$ and $r$ in the output. Observe
        that both the input and the output size of this formula are only a
        function of $p$ and, hence, constant. Therefore, we can explicitly
        construct a constant depth formula to capture this grouping.
                                                
        \paragraph{\boldmath Case $s = 0$.} For constructing the formula $\chi^2$ we first
        check whether $\vecx \in \overline{\calV}_{\vecf_1}$ and whether
        $\vecy$ is the lexicographically sorted version of $\Sigma_{\vecx}$.
        These can both be done as we have described in the construction of
        the formula $\psi$ above. If all checks pass, then we output the
        $p$ edges of the form $(\vecz, r)$ for all $\vecz \in \Sigma_{\vecx}$,
        that correspond to the $r$-th copy of the edge between $\vecz$ and $\vecy$.

        \medskip
        Combining the formulas $\psi$ and $\chi$ through a selector concludes the
        construction of $\phi$.
\end{description}

    Hence, our theorem follows from the observation that
  the instance of the $\cChevalleyp$ problem that we get when
  reducing $\lonelyp$ to $\cChevalleyp$ in
  \autoref{thm:combinedChevalleyPPApCompleteness} reduces to
  $\sucBipartitep[\ACzp]$.
\end{prevproof}

\addcontentsline{toc}{section}{References}
\DeclareUrlCommand{\Doi}{\urlstyle{sf}}
\renewcommand{\path}[1]{\small\Doi{#1}}
\renewcommand{\url}[1]{\href{#1}{\small\Doi{#1}}}
\bibliographystyle{alphaurl}
\bibliography{ppa}

\newcommand{\etalchar}[1]{$^{#1}$}
\begin{thebibliography}{CHK{\etalchar{+}}19}

\bibitem[ABB15]{aisenberg15tucker}
James Aisenberg, Maria~Luisa Bonet, and Sam Buss.
\newblock 2-d tucker is {PPA} complete.
\newblock {\em Electronic Colloquium on Computational Complexity {(ECCC)}},
  22:163, 2015.
\newblock URL: \url{http://eccc.hpi-web.de/report/2015/163}.

\bibitem[AFK84]{alon1984regular}
Noga Alon, Shmuel Friedland, and Gil Kalai.
\newblock Regular subgraphs of almost regular graphs.
\newblock {\em Journal of Combinatorial Theory, Series B}, 37(1):79--91, 1984.

\bibitem[Alo87]{alon87splitting}
Noga Alon.
\newblock Splitting necklaces.
\newblock {\em Advances in Mathematics}, 63(3):247--253, 1987.
\newblock \href {http://dx.doi.org/10.1016/0001-8708(87)90055-7}
  {\path{doi:10.1016/0001-8708(87)90055-7}}.

\bibitem[AW86]{alon86borsukulam}
Noga Alon and Douglas~B. West.
\newblock The {Borsuk-Ulam} theorem and bisection of necklaces.
\newblock {\em Proceedings of the American Mathematical Society},
  98(4):623--628, 1986.
\newblock \href {http://dx.doi.org/10.1090/S0002-9939-1986-0861764-9}
  {\path{doi:10.1090/S0002-9939-1986-0861764-9}}.

\bibitem[BCE{\etalchar{+}}98]{beame98relative}
Paul Beame, Stephen~A. Cook, Jeff Edmonds, Russell Impagliazzo, and Toniann
  Pitassi.
\newblock The relative complexity of {NP} search problems.
\newblock {\em J. Comput. Syst. Sci.}, 57(1):3--19, 1998.
\newblock \href {http://dx.doi.org/10.1006/jcss.1998.1575}
  {\path{doi:10.1006/jcss.1998.1575}}.

\bibitem[BG92]{beigel92counting}
Richard Beigel and John Gill.
\newblock Counting classes: Thresholds, parity, mods, and fewness.
\newblock {\em Theor. Comput. Sci.}, 103(1):3--23, 1992.
\newblock \href {http://dx.doi.org/10.1016/0304-3975(92)90084-S}
  {\path{doi:10.1016/0304-3975(92)90084-S}}.

\bibitem[BGIP01]{buss01linear}
Samuel~R. Buss, Dima Grigoriev, Russell Impagliazzo, and Toniann Pitassi.
\newblock Linear gaps between degrees for the polynomial calculus modulo
  distinct primes.
\newblock {\em J. Comput. Syst. Sci.}, 62(2):267--289, 2001.
\newblock \href {http://dx.doi.org/10.1006/jcss.2000.1726}
  {\path{doi:10.1006/jcss.2000.1726}}.

\bibitem[BIQ{\etalchar{+}}17]{belovs17nullstellensatz}
Aleksandrs Belovs, G{\'{a}}bor Ivanyos, Youming Qiao, Miklos Santha, and Siyi
  Yang.
\newblock On the polynomial parity argument complexity of the combinatorial
  nullstellensatz.
\newblock In {\em 32nd Computational Complexity Conference, {CCC} 2017, July
  6-9, 2017, Riga, Latvia}, pages 30:1--30:24, 2017.
\newblock \href {http://dx.doi.org/10.4230/LIPIcs.CCC.2017.30}
  {\path{doi:10.4230/LIPIcs.CCC.2017.30}}.

\bibitem[BJ12]{buss12propositional}
Samuel~R. Buss and Alan~S. Johnson.
\newblock Propositional proofs and reductions between {NP} search problems.
\newblock {\em Ann. Pure Appl. Logic}, 163(9):1163--1182, 2012.
\newblock \href {http://dx.doi.org/10.1016/j.apal.2012.01.015}
  {\path{doi:10.1016/j.apal.2012.01.015}}.

\bibitem[BM04]{bureshoppenheim04relativized}
Josh Buresh{-}Oppenheim and Tsuyoshi Morioka.
\newblock Relativized {NP} search problems and propositional proof systems.
\newblock In {\em 19th Annual {IEEE} Conference on Computational Complexity
  {(CCC} 2004), 21-24 June 2004, Amherst, MA, {USA}}, pages 54--67, 2004.
\newblock \href {http://dx.doi.org/10.1109/CCC.2004.1313795}
  {\path{doi:10.1109/CCC.2004.1313795}}.

\bibitem[BO06]{bureshoppenheim06factoring}
Joshua Buresh-Oppenheim.
\newblock On the {TFNP} complexity of factoring.
\newblock {\em Manuscript}, 2006.
\newblock URL: \url{http://www.cs.toronto.edu/~bureshop/factor.pdf}.

\bibitem[BPR15]{bitansky2015cryptographic}
Nir Bitansky, Omer Paneth, and Alon Rosen.
\newblock On the cryptographic hardness of finding a nash equilibrium.
\newblock In {\em 2015 IEEE 56th Annual Symposium on Foundations of Computer
  Science}, pages 1480--1498. IEEE, 2015.

\bibitem[BR98]{beame98more}
Paul Beame and S{\o}ren Riis.
\newblock More on the relative strength of counting principles.
\newblock In {\em Proceedings of the DIMACS Workshop on Proof Complexity and
  Feasible Arithmetics}, volume~39, pages 13--35, 1998.

\bibitem[BSS81]{barany81topological}
I.~Bárány, S.~B. Shlosman, and A.~Szücs.
\newblock On a topological generalization of a theorem of tverberg.
\newblock {\em Journal of the London Mathematical Society}, s2-23(1):158--164,
  1981.
\newblock \href {http://dx.doi.org/10.1112/jlms/s2-23.1.158}
  {\path{doi:10.1112/jlms/s2-23.1.158}}.

\bibitem[Che35]{chevalley35demonstration}
Claude Chevalley.
\newblock D{\'e}monstration d'une hypoth{\`e}se de m. artin.
\newblock {\em Abhandlungen aus dem Mathematischen Seminar der Universit{\"a}t
  Hamburg}, 11(1):73--75, Dec 1935.
\newblock \href {http://dx.doi.org/10.1007/BF02940714}
  {\path{doi:10.1007/BF02940714}}.

\bibitem[CHK{\etalchar{+}}19]{choudhuri2019finding}
Arka~Rai Choudhuri, Pavel Hub{\'a}cek, Chethan Kamath, Krzysztof Pietrzak, Alon
  Rosen, and Guy~N Rothblum.
\newblock Finding a nash equilibrium is no easier than breaking fiat-shamir.
\newblock In {\em Proceedings of the 51st Annual ACM SIGACT Symposium on Theory
  of Computing}, pages 1103--1114. ACM, 2019.

\bibitem[DEF{\etalchar{+}}16]{deng16ppa}
Xiaotie Deng, Jack~R. Edmonds, Zhe Feng, Zhengyang Liu, Qi~Qi, and Zeying Xu.
\newblock {Understanding PPA-Completeness}.
\newblock In Ran Raz, editor, {\em 31st Conference on Computational Complexity
  (CCC 2016)}, volume~50 of {\em Leibniz International Proceedings in
  Informatics (LIPIcs)}, pages 23:1--23:25, Dagstuhl, Germany, 2016. Schloss
  Dagstuhl--Leibniz-Zentrum fuer Informatik.
\newblock URL: \url{http://drops.dagstuhl.de/opus/volltexte/2016/5831}, \href
  {http://dx.doi.org/10.4230/LIPIcs.CCC.2016.23}
  {\path{doi:10.4230/LIPIcs.CCC.2016.23}}.

\bibitem[DGP09]{daskalakis09nash}
Constantinos Daskalakis, Paul~W. Goldberg, and Christos~H. Papadimitriou.
\newblock The complexity of computing a nash equilibrium.
\newblock {\em {SIAM} J. Comput.}, 39(1):195--259, 2009.
\newblock \href {http://dx.doi.org/10.1137/070699652}
  {\path{doi:10.1137/070699652}}.

\bibitem[DP11]{daskalakis11continuous}
Constantinos Daskalakis and Christos~H. Papadimitriou.
\newblock Continuous local search.
\newblock In {\em Proceedings of the Twenty-Second Annual {ACM-SIAM} Symposium
  on Discrete Algorithms, {SODA} 2011, San Francisco, California, USA, January
  23-25, 2011}, pages 790--804, 2011.
\newblock \href {http://dx.doi.org/10.1137/1.9781611973082.62}
  {\path{doi:10.1137/1.9781611973082.62}}.

\bibitem[FG18]{filosratsikas18consensus}
Aris Filos{-}Ratsikas and Paul~W. Goldberg.
\newblock Consensus halving is ppa-complete.
\newblock In {\em Proceedings of the 50th Annual {ACM} {SIGACT} Symposium on
  Theory of Computing, {STOC} 2018, Los Angeles, CA, USA, June 25-29, 2018},
  pages 51--64, 2018.
\newblock \href {http://dx.doi.org/10.1145/3188745.3188880}
  {\path{doi:10.1145/3188745.3188880}}.

\bibitem[FG19]{filosratsikas19splitting}
Aris Filos{-}Ratsikas and Paul~W. Goldberg.
\newblock The complexity of splitting necklaces and bisecting ham sandwiches.
\newblock In {\em STOC (to appear)}, 2019.
\newblock URL: \url{http://arxiv.org/abs/1805.12559}.

\bibitem[GKRS19]{goos19adventures}
Mika G{\"{o}}{\"{o}}s, Pritish Kamath, Robert Robere, and Dmitry Sokolov.
\newblock Adventures in monotone complexity and {TFNP}.
\newblock In {\em 10th Innovations in Theoretical Computer Science Conference,
  {ITCS} 2019, January 10-12, 2019, San Diego, California, {USA}}, pages
  38:1--38:19, 2019.
\newblock \href {http://dx.doi.org/10.4230/LIPIcs.ITCS.2019.38}
  {\path{doi:10.4230/LIPIcs.ITCS.2019.38}}.

\bibitem[Gri01]{grigni2001sperner}
Michelangelo Grigni.
\newblock A sperner lemma complete for ppa.
\newblock {\em Information Processing Letters}, 77(5-6):255--259, 2001.

\bibitem[GW85]{goldberg85bisection}
C.~Goldberg and D.~West.
\newblock Bisection of circle colorings.
\newblock {\em SIAM Journal on Algebraic Discrete Methods}, 6(1):93--106, 1985.
\newblock \href {http://dx.doi.org/10.1137/0606010}
  {\path{doi:10.1137/0606010}}.

\bibitem[Hol19]{hollender19ppak}
Alexandros Hollender.
\newblock The classes {PPA-$k$}: Existence from arguments modulo $k$.
\newblock In Ioannis Caragiannis, Vahab Mirrokni, and Evdokia Nikolova,
  editors, {\em Web and Internet Economics}, pages 214--227, Cham, 2019.
  Springer International Publishing.

\bibitem[Jer16]{jerabek16integer}
Emil Jer{\'{a}}bek.
\newblock Integer factoring and modular square roots.
\newblock {\em J. Comput. Syst. Sci.}, 82(2):380--394, 2016.
\newblock \href {http://dx.doi.org/10.1016/j.jcss.2015.08.001}
  {\path{doi:10.1016/j.jcss.2015.08.001}}.

\bibitem[Joh11]{johnson11reductions}
Alan~S. Johnson.
\newblock Reductions and propositional proofs for total {NP} search problems.
\newblock {\em {UC San Diego} Electronic Theses and Dissertations}, 2011.
\newblock URL: \url{https://escholarship.org/uc/item/89r774x7}.

\bibitem[JPY88]{johnson88local}
David~S. Johnson, Christos~H. Papadimitriou, and Mihalis Yannakakis.
\newblock How easy is local search?
\newblock {\em J. Comput. Syst. Sci.}, 37(1):79--100, 1988.
\newblock \href {http://dx.doi.org/10.1016/0022-0000(88)90046-3}
  {\path{doi:10.1016/0022-0000(88)90046-3}}.

\bibitem[KM18]{kothari2018sum}
Pravesh~K Kothari and Ruta Mehta.
\newblock Sum-of-squares meets nash: lower bounds for finding any equilibrium.
\newblock In {\em Proceedings of the 50th Annual ACM SIGACT Symposium on Theory
  of Computing}, pages 1241--1248. ACM, 2018.

\bibitem[KNY19]{komargodski2019white}
Ilan Komargodski, Moni Naor, and Eylon Yogev.
\newblock White-box vs. black-box complexity of search problems: Ramsey and
  graph property testing.
\newblock {\em Journal of the ACM (JACM)}, 66(5):34, 2019.

\bibitem[KS98]{kreher98combinatorial}
Donald~L. Kreher and Douglas~R. Stinson.
\newblock {\em Combinatorial Algorithms: Generation, Enumeration, and Search},
  volume~7 of {\em Discrete Mathematics and Its Applications}.
\newblock CRC Press, 1998.

\bibitem[MP91]{megiddo91total}
Nimrod Megiddo and Christos~H. Papadimitriou.
\newblock On total functions, existence theorems and computational complexity.
\newblock {\em Theor. Comput. Sci.}, 81(2):317--324, 1991.
\newblock \href {http://dx.doi.org/10.1016/0304-3975(91)90200-L}
  {\path{doi:10.1016/0304-3975(91)90200-L}}.

\bibitem[Pap94]{papadimitriou94parity}
Christos~H. Papadimitriou.
\newblock On the complexity of the parity argument and other inefficient proofs
  of existence.
\newblock {\em J. Comput. Syst. Sci.}, 48(3):498--532, 1994.
\newblock \href {http://dx.doi.org/10.1016/S0022-0000(05)80063-7}
  {\path{doi:10.1016/S0022-0000(05)80063-7}}.

\bibitem[Rei07]{Reiher07}
Christian Reiher.
\newblock On kemnitz’conjecture concerning lattice-points in the plane.
\newblock {\em The Ramanujan Journal}, 13(1-3):333--337, 2007.

\bibitem[Rub16]{rubinstein16}
Aviad Rubinstein.
\newblock Settling the complexity of computing approximate two-player nash
  equilibria.
\newblock In {\em {IEEE} 57th Annual Symposium on Foundations of Computer
  Science, {FOCS} 2016, 9-11 October 2016, Hyatt Regency, New Brunswick, New
  Jersey, {USA}}, pages 258--265, 2016.

\bibitem[SZZ18]{sotiraki18ppp}
Katerina Sotiraki, Manolis Zampetakis, and Giorgos Zirdelis.
\newblock Ppp-completeness with connections to cryptography.
\newblock In {\em 59th {IEEE} Annual Symposium on Foundations of Computer
  Science, {FOCS} 2018, Paris, France, October 7-9, 2018}, pages 148--158,
  2018.
\newblock \href {http://dx.doi.org/10.1109/FOCS.2018.00023}
  {\path{doi:10.1109/FOCS.2018.00023}}.

\bibitem[War36]{warning36bemerkung}
Ewald Warning.
\newblock Bemerkung zur vorstehenden arbeit von herrn chevalley.
\newblock {\em Abh. Math. Sem. Univ. Hamburg}, 11:76--83, 1936.

\end{thebibliography}
\end{document}